\documentclass[preprint,12pt]{article}

\usepackage{amssymb}
\usepackage{amsthm}

\usepackage{graphicx}
\usepackage{color,soul}
\usepackage{nicefrac}
\usepackage{amsmath, amssymb}
\usepackage{xcolor}
\usepackage{bold-extra}
\usepackage{xspace}
\usepackage{comment}
\usepackage{amsfonts}
\usepackage{stfloats}

\usepackage{enumitem}
\usepackage{epstopdf}
\usepackage{tabularx}
\usepackage{setspace}
\usepackage{array}
\usepackage{tabulary}
\usepackage[normalem]{ulem} 
\usepackage{etoolbox}
\usepackage{adjustbox}
\usepackage{nicefrac}
\usepackage{calligra}
\usepackage[T1]{fontenc}
\usepackage{mathtools} 
\usepackage{framed}
\usepackage{url}
\usepackage{longtable}
\usepackage{colortbl}
\usepackage{authblk}

\definecolor{light-gray}{gray}{0.96}
\definecolor{LightCyan}{rgb}{0.88,1,1}

\newtheorem{theorem}{Theorem}
\newtheorem{lemma}{Lemma}
\newtheorem{definition}{Definition}

\newenvironment{frcseries}{\fontfamily{bikham}\selectfont}{}

\apptocmd{\sloppy}{\hbadness 10000\relax}{}{}
\usepackage{soul}
\makeatletter
\def\SOUL@hlpreamble{%
\setul{\dimexpr\dp\strutbox+1.5pt}{\dimexpr\ht\strutbox+\dp\strutbox+1.5pt\relax}
\let\SOUL@stcolor\SOUL@hlcolor
\SOUL@stpreamble
}

\newcommand{\whp}{with high probability\xspace}
\newcommand{\uar}{uniformly at random\xspace}
\newcommand{\defID}{\textsc{DefID}\xspace}
\newcommand{\AdvPower}{\kappa\xspace} 
\newcommand{\sma}{\alpha\xspace}
\newcommand{\smb}{\beta\xspace}

\newcommand{\lifetime}{\gamma\xspace}
\newcommand{\AOneL}{1/\sma\xspace}
\newcommand{\AOneH}{\sma\xspace}

\newcommand{\newAlg}{\textsc{Ergo}\xspace}
\newcommand{\SybilFuse}{\textsc{SybilFuse}\xspace}
\newcommand{\SybilExpose}{\textsc{SybilExpose}\xspace}
\newcommand{\POW}{PoW\xspace}

\newcommand{\AlgA}{\textsc{SybilControl\xspace}}
\newcommand{\AlgB}{\textsc{CCom}\xspace}
\newcommand{\AlgC}{REMP\xspace}

\newcommand{\polylog}{\mbox{polylog}\xspace}

\newcommand{\Iters}{\mathcal{I}}

\newcommand{\jAll}{J\xspace}  
\newcommand{\jInterval}{\mathtt{J}\xspace} 
\newcommand{\JoinEst}{\tilde{\mathtt{J}}\xspace} 
\newcommand{\jIntervalAll}{\mathtt{J_a}\xspace} 
\newcommand{\interEll}{\mathcal{L}}
\newcommand{\jIter}{\mathcal{J}}

\newcommand{\joinRate}{J\xspace}  
\newcommand{\iJRate}{J\xspace} 

\newcommand{\advRate}{\ensuremath{J^B}\xspace}

\newcommand{\epochRate}{\rho\xspace}  

\newcommand{\advAveCost}{T\xspace}

\newcommand{\algGM}{\textsc{GMCom}\xspace}
\newcommand{\genID}{\textsc{GenID}\xspace}

\newcommand{\iterIDs}{\mathcal{S}\xspace}

\newcommand{\floor}[1]{\left\lfloor #1 \right\rfloor}
\newcommand{\ceil}[1]{\left\lceil #1 \right\rceil}

\newcommand{\defn}[1]{\textbf{\emph{#1}}}
\newcommand{\goodJest}{\textsc{GoodJEst}\xspace}
\newcommand{\ergo}{\textsc{Ergo}\xspace}
\newcommand{\EstGoodJoin}{\textsc{GoodJEst}\xspace}
    
\title{Bankrupting Sybil Despite Churn}

\author[1]{Diksha Gupta}
\author[2]{Jared Saia}
\author[3]{Maxwell Young}

\affil[1]{IBM Research Singapore, 9 Changi Business Park Central 1, Singapore,486048, \texttt{gupta.diksha.1991@gmail.com}}

\affil[2]{Department of Computer Science, University of New Mexico, Albuquerque, NM, USA, 87131, \texttt{saia@cs.unm.edu}} 
               
\affil[3]{Department of Computer Science and Engineering, Mississippi State University, MS, USA, 39762, \texttt{myoung@cse.msstate.edu}}    

\date{}
\begin{document}
\maketitle

\begin{abstract}

A Sybil attack occurs when an adversary controls multiple identifiers (IDs) in a system.  Limiting the number of Sybil (bad) IDs to a minority is critical to the use of well-established tools for tolerating malicious behavior, such as  Byzantine agreement and secure multiparty computation.

A popular technique for enforcing a Sybil minority is resource burning: the verifiable consumption of a network resource, such as computational power, bandwidth, or memory.  Unfortunately, typical defenses based on resource burning require non-Sybil (good) IDs to consume at least as many resources as the adversary.  Additionally, they have a high resource burning cost, even when the system membership is relatively stable.

Here, we present a new Sybil defense, \ergo, that guarantees (1) there is always a minority of bad IDs; and (2) when the system is under significant attack, the good IDs consume asymptotically less resources than the bad.  In particular, for churn rate that can vary exponentially, the resource burning rate  for good IDs under \ergo is {$O(\sqrt{TJ} + J)$}, where $T$ is the resource burning rate of the adversary, and $J$ is the join rate of good IDs.  We show this resource burning rate is asymptotically optimal for a large class of algorithms. 

We empirically evaluate \ergo alongside prior Sybil defenses.  Additionally, we show that \ergo can be combined with machine learning techniques for classifying Sybil IDs, while preserving its theoretical guarantees.  Based on our experiments comparing \ergo with two previous Sybil defenses,  \ergo improves  on the amount of resource burning relative to the adversary by up to $2$ orders of magnitude without machine learning, and up to $3$ orders of magnitude using machine learning.
\end{abstract}

\clearpage 






\section{Introduction}\label{s:intro}

A \defn{Sybil attack} occurs when an adversary controls multiple identifiers (\defn{IDs}) in the system~\cite{douceur02sybil}.  In practice, these attacks may be motivated by financial gain~\cite{zhang2019double,heilman2015eclipse,bitcoin-sybil} and their mitigation has attracted significant attention from the research community~\cite{mohaisen:sybil}.

A classic defense is resource burning (RB), whereby IDs are periodically required to consume local resources in a verifiable manner~\cite{gupta:resource-burning}. A well-known example of resource burning is proof-of-work (\POW)~\cite{dwork:pricing}, but several other methods exist (see Section~\ref{sec:model-main}).

Unfortunately, current resource burning methods always consume resources, even when the system is not under attack. This non-stop consumption translates to substantial monetary costs~\cite{economistBC,arstechnica}.  

Prior work shows it is sometimes possible for  non-adversarial (\defn{good}) IDs to consume fewer total resources than the adversary~\cite{Gupta_Saia_Young_2019}.  Unfortunately, this work fails to hold when the rate at which system membership changes---often referred to as the \defn{churn rate}--- is high. Many systems vulnerable to Sybil attacks have significant churn~\cite{Stutzbach:2006:UCP:1177080.1177105,falkner:profiling,6688697}.

\subsection{Bankrupting Sybil Despite Churn}\label{s:bankrupting-sybil}

Given this state of affairs, a natural question is: {\it Despite churn, can we design a Sybil defense where the cost to good IDs scales slowly with the cost to the adversary?} 

A key contribution of our work is answering this question in the affirmative. Informally,  we design and analyze a new defense where the rate of RB performed in total by the good IDs:
\begin{enumerate}[leftmargin=16pt]
\item depends only on the churn rate of good IDs, when there is no attack; \vspace{-5pt}
\item grows {\it asymptotically} slower than that of the adversary during an attack.  \vspace{-3pt}
\end{enumerate}
In Section~\ref{sec:contributions}, we provide a formal statement of these guarantees, but their implications are worth highlighting here. Regarding (1), in the absence of an attack, the overhead from our defense is low. Specifically, each good ID incurs only a constant amount of RB to join the system and then a constant amount each time the system size increases or decreases by a constant factor (details are given in Section~\ref{sec:tog}). Therefore, the total amount of RB performed by good IDs grows slowly, despite large changes in the system size.

Regarding (2), given that RB ultimately translates into a monetary cost, our defense places the adversary at an economic disadvantage. A Sybil attack is now provably more expensive given our defense, which is a deterrent against malicious behavior.

\subsection{A Roadmap}

The remainder of this manuscript is structured as follows. In Section~\ref{sec:model-main}, we introduce our model and formal problem. With these preliminaries in place, we give a formal statement of our main results in Section~\ref{sec:contributions}, along with highlighting the novelty of our approach, along with new material not present in preliminary conference versions of our work. To provide additional context for our results, we discuss prior related work in Section~\ref{sec:related-work}. 

The specifics of our defense are given in Section~\ref{sec:tog} and \ref{s:goodJest}, along with a discussion providing intuition for our design decisions. An upper-bound analysis is presented in Section \ref{s:analysis}, and followed up by experimental results in Section \ref{sec:experiments} that illuminate the performance of our defense. In Section~\ref{sec:lower}, we provide a lower bound to show that our result is asymptotically optimal from a large class of natural defenses. 

A discussion of how to fully decentralize our defense is given in Section~\ref{s:committee}. Finally, in Section~\ref{sec:future}, we conclude with several promising avenues for future research.


\section{Model and Problem}\label{sec:model-main}

We now describe a general network model.  The system consists of virtual \defn{identifiers (IDs)}, where each ID is either \defn{good} if it obeys the protocol, or \defn{bad} if it is controlled by the Sybil adversary (or just \defn{adversary}). The amount of resource burning that each good ID can perform is identical; that is, the population of good IDs is homogeneous.\smallskip

\noindent{\bf Resource-Burning Challenges.}\label{sec:puzzle} IDs can construct \defn{resource-burning (RB) challenges}  of varying hardness, whose solutions cannot be stolen or pre-computed; some examples are discussed in Section~\ref{sec:related-work}. To specify hardness, a  {\boldmath{$k$}}\defn{-hard RB challenge} for any integer $k \geq 1$ imposes a resource cost of $k$ on the challenge solver. Our results are agnostic to the type of challenges employed, either those discussed above or new resource-burning schemes available for future use. \smallskip

\noindent{\bf Coordination.} To simplify our presentation, we assume that there is a single server running our algorithms.  The server---or the committee in our decentralized setting---learns when an ID joins or departs the system. Later, in Section~\ref{s:committee}, we show how the server can be replaced with a small committee, thus allowing our algorithms to execute in decentralized settings.

A \defn{round} is the amount of time it takes to solve a $1$-hard challenge plus the time for communication between the server and corresponding ID for issuing the challenge and returning a solution.  As is common in the literature, we assume that good IDs have clocks that are closely synchronized. Intuitively, if there is significant message delay or clock drift, then a challenger cannot accurately measure the response time for the ID solving the challenge; see~\cite{synchrony:malkhi} for further discussion. Techniques for synchronizing on the order of milliseconds are known and suffice for our purposes~\cite{392384}.  We assume that no more than an $\epsilon$-fraction of the good IDs depart in any round, for some small positive constant $\epsilon < 1/12$.

\smallskip

\noindent{\bf Adversary.}\label{sec:adv} A single adversary controls all bad IDs.  This pessimistically represents perfect collusion and coordination by the bad IDs. Bad IDs may arbitrarily deviate from our protocol, including sending incorrect or spurious messages. The adversary can send messages to any ID at will, can read the messages sent by good IDs before sending its own.  The adversary sets the timing of join and departure events by good and bad IDs, and makes these timing choices adaptively over time. Furthermore, the adversary can select which bad ID departs. However,  the adversary cannot select which specific good ID departs, but rather the departing ID is selected uniformly at random (u.a.r.) from the good IDs currently in the system.

By being able to scheduling the departure times of good IDs, we capture a strong adversary. In terms of motivation, one may imagine an attacker who has been monitoring the system for a long period of time and, therefore, possesses statistics on when good IDs join and depart over different times of the day. Such an attacker might not know {\it exactly} when good IDs are going to join/depart, but it could predict this behavior to a degree, and our pessimistic model accounts for this. We also note that considering an adversary without this power does not appear to make it easier to obtain theoretical guarantees on security and performance.

The adversary is resource-bounded: in any single round where all IDs are solving challenges, the adversary can solve a {\boldmath{$\AdvPower$}}-fraction of the challenges; this assumption is common~\cite{nakamoto:bitcoin,andrychowicz2015pow,GiladHMVZ17}. \smallskip

\subsection{ABC Model of Churn}

We now describe our model of churn, which we call the \defn{ABC model} of churn, based on two parameters $\alpha$ (A) and $\beta$ (B) of churn (C).  This churn model is particular relevant for defending against Sybil attacks (see Section~\ref{s:demand} for further discussion).

\subsubsection{Joins and Departures}\label{sec:join}

For simplicity, in our model, every join and departure event is assumed to occur at a unique point in time. In practice, because of the granularity of clocks, it may seem as if events are occurring at the same time, and this can be handled by having the server or committee order these events; all of our results hold even if this ordering is dictated by the adversary.

Whenever the adversary decides to cause a good ID departure event,  the adversary does not get to select which good ID departs. Rather, the departing good ID is selected independently and uniformly at random from the set of good IDs in the system.  Departing good IDs announce their departure. 

In practice, each good ID can issue “heartbeat messages” to the server that indicate they are still alive. These messages are sent periodically, and their absence is interpreted as a departure by the corresponding ID. We note that a bad ID that fails to issue heartbeat messages will be treated by the server as having departed from the system.  Thus, even departures by bad IDs are detectable.

Every joining ID is treated as a new ID.  We  ensure every joining ID is given a unique name by concatenating a join-event counter to the name chosen by the ID. 
We assume that every joining ID initially knows one good ID in the system.  This follows if every joining ID knows a set of IDs with a good majority, since the majority of good IDs can suggest a single good one; further it is needed to avoid eclipse attacks~\cite{urdaneta:survey,heilman2015eclipse}.  Thus it is a standard assumption in peer-to-peer networks~\cite{stoica_etal:chord,guerraoui:highly,JaiyeolaPSYZ17,fiat:making,awerbuch_scheideler:group,castro:secure,young:towards,awerbuch:towards2,awerbuch:random,awerbuch:towards,scheideler:how}; cryptocurrency systems like Ethereum, which uses the structured Kademlia peer-to-peer network~\cite{wang2021ethna, maymounkov:kademlia}, and Bitcoin, which uses an ad-hoc peer-to-peer network~\cite{nakamoto:bitcoin,heilman2015eclipse}; and for decentralized blockchains that use peer-to-peer overlays~\cite{CromanDEGJKMSSS16}.

We define {\boldmath{$n_0$}} to be the minimum number of good IDs in the system at any time. The \defn{system lifetime} is defined over $n_0^{\lifetime}$ joins and departures, for any fixed constant {\boldmath{$\lifetime$}} $>0$.  Our results for \goodJest and the costs for \ergo hold over this system lifetime  (see Section~\ref{sec:contributions}).  A system designer may choose a value of $\gamma$ in order to have the guarantees provided by our results.  

To illustrate, imagine a system where there are always at least $6000$ good IDs participating. If a system designer wishes the guarantees of \ergo to hold over $100$ billion join and departure events, then $\gamma=3$ will suffice. For dynamic systems, such guarantees are common (for examples, see~\cite{awerbuch:towards,awerbuch:random,guerraoui:highly}).


\subsubsection{Epochs,  Smoothness, and Churn} \label{s:churn}

\medskip
Here, we give the details for our model of churn. Later, in Section~\ref{s:a-b-churn}, we argue that our model of  {\it good} churn is quite general, and we provide intuition for why. We also  highlight that we make no assumptions about the bad churn.

Let $A \triangle B$  denote the symmetric difference between any two sets $A$ and $B$, i.e. $A \triangle B = (A-B) \cup (B-A)$. Time is partitioned into \defn{epochs} whose boundaries occur when the symmetric difference between the sets of good IDs at the start and the end of the epoch exceeds  $1/2$ times the number of good IDs at the start. Epochs are important to our model and analysis; however, our approach does {\it not} assume knowledge of when epochs begin or end.

Good churn is specified by two {\it a priori unknown} parameters:  {\boldmath{$\alpha, \beta$}}.  First, the good join rate between two consecutive epochs differs by at most an $\alpha$-factor.  Second, the number of good IDs that join or depart during $\ell$ consecutive seconds within an epoch differs by at most a $\beta$-factor from $\ell$ times the good join rate of the epoch.  Thus, $\alpha$ characterizes how the good join rate changes over epochs; and $\beta$ characterizes the burstiness of good ID arrivals and departures within an epoch. 

\smallskip 

Let {\boldmath{$\epochRate_i$}} be the join rate of good IDs (i.e., \defn{good join rate}) in epoch $i$; that is, the number of good IDs that join in epoch $i$ divided by the number of seconds in epoch $i$. We now formally define the notion of ``smoothness'' for $\alpha$ and $\beta$, which allows our model to capture a churn rate that varies widely over time.

\begin{definition}\label{def:smoothness}
For any $\sma\geq 1$ and $\smb \geq 1$, and any epoch $i>1$, we define the following.
\begin{itemize}[leftmargin=10pt]
\item{\boldmath{$\sma$}\defn{-smoothness}:} $(\AOneL) \epochRate_{i-1}  \leq \epochRate_{i} \leq \AOneH\, \epochRate_{i-1}$.

\item {\boldmath{$\smb$}\defn{-smoothness}:}  For any duration of $\ell$ seconds in the epoch, the number of good IDs that join is at least $ \lfloor \ell \epochRate_{i} / \beta\rfloor $ and at most $\lceil\beta \ell \epochRate_{i}\rceil$. Also, the number of good IDs that depart during this duration is at most $\lceil\beta \ell \epochRate_{i}\rceil$. \smallskip
\end{itemize} 
\end{definition}

Our guarantees require that $n_0\geq \max\{6000, (720(\gamma+1))^{4/3}, (41\smb)^2 \}$. This lower bound on $n_0$ arises in our analysis and we discuss this in Section~\ref{sec:parameter-constant-discussion}.

\medskip\smallskip

\noindent{\bf Varying Churn Rate.} In Section~\ref{s:a-b-churn}, we motivate our new churn model, and we provide intuition with examples aided by illustrations. Here, we succinctly highlight the roles of $\sma$ and $\smb$.

The parameter $\sma$ can capture {\it any} possible change in the good join rate between consecutive epochs, since there always exists an $\sma$ that satisfies the definition. Thus, the good join rate may change rapidly. For example, say $\sma=2$.  In this case, the good join rate may increase {\it exponentially} from epoch to epoch. Similarly, the good join rate may also decrease exponentially. The parameter $\beta$ ensures there can be possibly large deviations within an epoch from the average good join rate over the entire epoch. 

Many other churn models have been described in the research literature. In Section~\ref{s:churn-models}, we compare our model of churn against several popular prior models.


\subsection{The \defID Problem}  \label{s:defID}
In the well-studied \genID problem~\cite{AspnesJK2005,katz2014pseudonymous,andrychowicz2015pow,hou2017randomized,aggarwal2019bootstrapping}, there is some initial set of good and bad IDs in a permissionless system.   The problem \genID requires the IDs to all \emph{generate} an initial set of IDs with a bounded fraction of bad.  In particular, all good IDs must decide on a set of $S$ such that: all good IDs are in $S$; and a $O(\kappa)$-fraction of the IDs in $S$ are bad.

We introduce the \defn{\defID (DEFend ID)} problem, which generalizes \genID to handle churn.  Specifically, ID join and departure events follow our ABC model of churn.  Our goal is to ensure that, at any time $t$, all good IDs know a set $S(t)$ such that (1) all good IDs are in $S(t)$; and (2) a $O(\kappa)$-fraction of IDs in $S(t)$ are bad.

\subsubsection{Why is \defID harder than \genID?}\label{s:defID-harder}

\defID is strictly harder than \genID. Now, the fraction of bad IDs increases whenever a bad ID joins or a good ID departs. Since bad and good IDs cannot be differentiated {\it a priori}, the $O(\kappa)$ bound on bad IDs may be violated after some amount of churn. A naive approach to solve \defID would be to run a solution to \genID after every join and departure event.  But this is expensive.  Specifically, it requires a resource-burning cost that is linear in the current system size for \emph{every} ID join and departure event, using the best known bound for \genID~\cite{aggarwal2019bootstrapping}.  Thus, it is necessary to define this new \defID problem in order to be able to design an efficient Sybil defense in the presence of churn.

\section{Our Results}\label{sec:contributions}
\smallskip
 
Recall that $\kappa$ is the fraction of the resource that is controlled by the adversary.  Let the \defn{good spend rate} be the total resource burning cost for all good IDs per second. Similarly, let {\boldmath{$T$}} be the \defn{adversary's spend rate} and let {\boldmath{$J$}} be the \defn{join rate of good IDs} over the system lifetime.  

The first theorem solves the \defID problem using our algorithm \textbf{\ergo}, whose specification is given in Section~\ref{sec:tog}.

\smallskip

\begin{theorem}\label{thm:new-main-upper}
For $\AdvPower \leq 1/18$,  \ergo ensures that the fraction of bad IDs in the system is always less than $3 \kappa \leq 1/6$.  Additionally, with probability of error that is $o(1/n_0)$ over  the system lifetime, the good spend rate is $$O\left(\sma^{11/2}\smb^7 \sqrt{\advAveCost(\jAll + 1)}  + \sma^{11}\smb^{14}\jAll \right).$$
\end{theorem}
 
\smallskip 

A strict upper bound of $1/6$ enables solutions to problems such as Byzantine agreement and secure multiparty computation~\cite{lynch1996distributed, benor_goldwasser_wigderson:completeness}. 

To complement our upper bound, we provide a lower bound in Theorem~\ref{t:lower-bound} that shows that this result is asymptotically optimal for a large class of algorithms. 

To complement the upper bound in Theorem~\ref{thm:new-main-upper}, we show in Section~\ref{sec:lower} that when $\alpha$ and $\beta$ are $\Theta(1)$, \ergo's resource burning rate is asymptotically optimal over a large class of algorithms. We also show in Section~\ref{s:committee} how to remove reliance on a single server so that \ergo can be executed in a decentralized fashion, while providing similar guarantees (see Theorem~\ref{thm:decent}).
\medskip 

\ergo makes use of a second algorithm, {\textbf{\goodJest}}, that may be of independent interest.  \goodJest correctly estimates the good join rate, provided that the fraction of bad IDs is always less than $1/6$. {\bf Since \ergo always ensures that the fraction of bad IDs is less than {\boldmath{$1/6$}}}, it is able to make use of \goodJest.  As addressed in Sections~\ref{sec:entrance-cost} and~\ref{s:analErgo}, \ergo uses \goodJest to achieve the good spend rate of Theorem~\ref{thm:new-main-upper}.  The main properties of \goodJest are as follows. 

\smallskip

\begin{theorem} \label{t:JoinEst}
Assume the fraction of bad IDs is always less than $1/6$. Then with probability of error that is $o(1/n_0)$ over the system lifetime the following holds.
Fix any epoch. Let $\epochRate$ be the good join rate during that epoch.  Then, if $\JoinEst$ is the estimate from \goodJest at any time during that epoch: 
    $$ 1/(88 \sma^4\smb^3) \epochRate \leq \JoinEst \leq 1867\sma^4\smb^5 \epochRate.$$
\end{theorem}

This theorem holds no matter how the adversary injects bad IDs. Based on our experiments on multiple networks, \goodJest always provides an estimate within a factor of $10$ of the true good join rate, and often much closer (cf. Section~\ref{sec:evalEst}).

We validate our theoretical results by  comparing \ergo against prior RB-based defenses using real-world data from several networks  (Section~\ref{section:empasym}). In terms of the amount of RB performed relative to the adversary, we find that \ergo performs up to $2$ orders of magnitude better than previous defenses, according to our simulations. Using insights from these first experiments, we engineer and evaluate several heuristics aimed at further improving the performance of \ergo.  Our best heuristic has \ergo leveraging a prior machine-learning method~\cite{gao2018sybilfuse}, and we observe an improvement by up to 3 orders of magnitude in comparison to previous algorithms for large-scale attacks. 

\subsection{A Note on Incentives}\label{sec:incentives}
In this paper, we do not explore the challenging problem of how to \emph{incentivize} good IDs to solve challenges.  This is an issue that concerns  resource-burning techniques in general (see Section~\ref{sec:related-work} for other examples of resource-burning), and there are a variety of ways in which good IDs might be motivated to solve puzzles.  Even though it is beyond the scope of this paper, we do {\it sketch} some ideas of how incentivization could happen, based on techniques used in cryptocurrency systems. For clarity of presentation, we defer this discussion to Section~\ref{sec:future}.

\subsection{New Contributions}\label{sec:new}

This manuscript contains results previously published in the proceedings of the {\it 41st IEEE International Conference on Distributed Computing Systems} (ICDCS'21)  \cite{Gupta_Saia_Young_2021} and in the proceedings of the {\it 33rd IEEE International Parallel and Distributed Processing Symposium} (IPDPS'19)~\cite{Gupta_Saia_Young_2019}. In particular, the majority of our results presented here appeared recently at ICDCS'21, while our lower-bound result appeared at IPDPS'19.

This manuscript has been substantially revised and expanded.  Specifically, we highlight the following new material:

\begin{itemize}[leftmargin=17pt]

\item Full proofs of all our results, including proofs of several lemmas that were not included in the conference proceedings due to space constraints (Section~\ref{s:analysis}).

\item Exposition giving intuition for the guarantees provided by \goodJest (Section~\ref{s:goodJestIntuition} and Figure~\ref{fig:gJest}).

\item A complete presentation of our approach for decentralizing \ergo. We have  included full proofs and added additional exposition (Section~\ref{s:committee}).

\item Additional empirical results for our main algorithm, \ergo, that defends against the Sybil attack. We design and evaluate several heuristics aimed at further lowering the cost of \ergo, while still retaining its guarantees (Section~\ref{section:heuristics} and Figure~\ref{fig:Heuristic}). 

\item Additional empirical results for the algorithm, \goodJest. We examine the performance of \goodJest over new data sets for the Ethereum and Gnutella systems (Section~\ref{sec:evalEst} and Figure~\ref{fig:PI}). 

\item Additional discussion of several pertinent issues related to our results:  incentives for good IDs to solve puzzles (Section~\ref{sec:incentives}), additional discussion of classification approaches for defending against the Sybil attack (Section~\ref{sec:related-work}), and future directions of research (Section~\ref{sec:future}).

\end{itemize}

\subsection{Technical Contributions and Novelty of Approach}\label{sec:novelty}

\medskip

\noindent The technical contributions of our results include:

\begin{itemize}[leftmargin=15pt]
\item \textbf{Formal Problem Definition.}  To the best of our knowledge, this paper is the first to formally define the Sybil defense problem with churn.  Our model of churn (Section~\ref{s:churn}) is quite general, simple to describe and mathematically tractable.  Our \defID problem (Section~\ref{s:defID}) formally defines what is needed to prevent Sybil attacks in the presence of churn.

\item \textbf{Spending Less than the Attacker.}   Before the results in this paper, it was unknown whether a Sybil defense algorithm could spend asymptotically less than an attacker.  In this paper, we describe and analyze the first Sybil defense algorithm that achieves this result.  The fact that such a result is even possible is surprising, and perhaps a useful contribution to the general science of distributed security.
\end{itemize}

\noindent
Achieving these results requires solving many technical issues.  Below, we summarize some of the novel techniques we have developed in order to address these problems.

We note in the following that, initially, we assume a single server maintains the system membership information (recall Section~\ref{sec:model-main}); in Section~\ref{s:committee} we describe how a committee can maintain system membership information.

\begin{itemize}[leftmargin=15pt]
    \item \textbf{Estimate and Set.} We introduce a new, general approach for the design of algorithms to secure a system in the presence of churn.  Our technique consists of two parts.  First, \emph{estimating good activity}: here, estimating the rate at which good IDs join the system.  Second, \emph{setting costs using this estimate}: here, costs to join the system are set carefully based on this estimate.
    \item \textbf{Estimating the Good Join Rate.}  A major technical challenge is estimating the good join rate (Sections~\ref{s:goodJest} and~\ref{s:analGoodJest}).  How can we achieve this estimate when we do not know, a priori, whether joining IDs are good or bad?  Our solution is to estimate the ``flow" of new, good IDs into the network during a certain time period, by bounding the flow of bad IDs in; and the flow of new, good IDs out.  Technical tools we use to achieve these bounds are: (1) defining the time period for measurement using symmetric difference of membership sets, in order to achieve both upper and lower bounds on the flows; and (2) deriving careful stochastic, concentration results to bound the flow of good IDs out, with high probability.  Section~\ref{s:goodJestIntuition} gives more intuition on our approach, which we believe again may be of independent interest.
    \item \textbf{Setting Costs based on Good Join Rate.}  We use our estimate of the good join rate to set the costs to enter the system.  Whenever the system membership changes by a constant fraction, we also charge all IDs a cost of $1$, in order to purge the system of a possible excess of bad IDs.  Setting the entrance cost too high means that good IDs pay too much when there is no attack.   Setting it too low means that bad IDs can cheaply join and increase costs in the future due to more frequent purges.  We can resolve this tension mathematically by setting the entrance costs to the ratio of the join rate of all IDs over the (estimated) join rate of good IDs.  Section~\ref{sec:entrance-cost} gives more intuition about our approach for setting costs, and Section~\ref{s:analErgo} gives our formal analysis, which makes use of bounds on our estimate of the good join rate, and also makes critical use of the Cauchy-Schwartz inequality. 
\end{itemize}

\section{Churn Models}\label{s:churn-models}

As we highlighted in above in Section~\ref{sec:novelty}, our churn model is an important contribution. Here, we summarize other popular churn models in the literature (Sections \ref{s:DNC} and \ref{s:other-pop}), and then we motivate and discuss our new model of churn (Section~\ref{s:a-b-churn}).

\subsection{Dynamic Network Model with Churn (DNC)}\label{s:DNC}

The \emph{Dynamic Network Model with Churn} (DNC) model was first proposed by Augustine, Pandurangan, Robinson, and Upfal in 2012~\cite{augustine2012towards}.  Then, subsequent papers described algorithms in this model to solve Byzantine agreement~\cite{augustine2012towards, augustine2015distributed}; leader election~\cite{augustine2015fast}; maintain a Distributed Hash Table (DHT)~\cite{jacobs2013stochastic}; maintain a bounded-degree expander topology~\cite{augustine2015enabling}  and solve storage and retrieval problems~\cite{augustine2013storage}.  The survey paper~\cite{augustine2016distributed} gives an overview of the DNC model and results.  

For completeness, we now describe aspects of the DNC model that are relevant to this paper.  DNC assumes that the network size is fixed.  An adversary is assumed to control what nodes join and leave and at what time, and the adversary has complete knowledge of the algorithm and unlimited computational power.  The maximum node churn rate is parameterized: in any round, up to $C(n)$ nodes can be replaced by new nodes.  Typically, algorithmic results in this model can handle $C(n) = \Theta(n)$ for an adversary that is oblivious: unaware of the past random choices of the algorithm; and $C(n) = \Theta(\sqrt{n})$ for an adversary that knows the random choices of the algorithm.  Finally, we note that the DNC model also considers edge churn in the network topology; we omit discussion of this aspect of the model since it is not relevant to this paper.

Byzantine IDs occur in the DNC model; typically with the assumption that for some positive $\epsilon$, there are $O(n^{1/2 - \epsilon})$ Byzantine IDs in every round.  The adversary controlling the Byzantine IDs also controls churn.  This is similar to the adversary in our own ABC model.

Results in the DNC generally hold for all but a $1/\polylog(n)$ fraction of the good IDs.  For example, in the leader election
results, all but a $1/\polylog(n)$ fraction of the IDs agree on the correct leader~\cite{augustine2015fast}; in 
Byzantine agreement~\cite{augustine2015fast, augustine2012towards, augustine2015distributed} all but a $1/\polylog(n)$ fraction of the IDs agree on a correct output bit.  These types of results are necessary given that the DNC model (1) assumes connectivity in sparse networks with edge churn; and (2) allows targeted deletions of good IDs by the adversary. 

\subsubsection{Differences between the DNC and ABC model}\label{s:whyABC}
We now discuss differences between the DNC model and the ABC model.  First, while the ABC model allows the adversary to schedule when good ID deletions will occur, it does not allow the adversary to target specific good ID.  In particular, in the ABC model, when a good ID deletion event occurs, the good ID that is deleted is selected \uar from the set of all good IDs.  

Byzantine IDs are another key difference.  In the DNC model, in every round, by assumption, for some positive $\epsilon$, there are at most $O(n^{1/2-\epsilon})$ Byzantine IDs in the system.  In contrast, in the ABC model, there is no such assumption.  Instead, the ABC model assumes that the adversary controls a constant fraction of the RB resource, and must use this fact to try to constrain the fraction of Byzantine IDs.

\subsection{Churn Models without Byzantine IDs}~\label{s:other-pop}
Several other models have been proposed for churn which, unlike the DNC and ABC models, do not consider Byzantine IDs.  We now discuss these other results.

First, work by Ko, Hoque and Indranil, from 2008, describes two churn models, \emph{TRAIN} and \emph{CROWD}~\cite{ko2008using}. In both models, the system size is fixed, and each ID join event occurs in parallel with an ID deletion event.  Additionally, in both models, the churn rate, or number of processes joining per unit time, is assumed to be fixed over the entire lifetime of the system.  In the TRAIN model, there is some fixed, positive integer $K$ and the join and leave events can only occur at times that are integer multiples of $K$.  In the more challenging CROWD model, join and leave events can happen at any time. In contrast, the ABC model allows for system size increase and decrease over time, and these changes are not restricted to specific multiples of events.

We note that the DNC model (recall Section~\ref{s:DNC}) is more general than both the TRAIN and CROWD models since it allows up to a certain number of join and deletion events to occur in a single round, but does not require that exactly a $K$ number of events occur.  Thus, it does not require that the churn rate stay fixed throughout the lifetime of the system, only that there is some upper bound on the churn rate.
 
Second, in 2004, Aguilera proposed several churn models allowing for infinitely many IDs~\cite{aguilera2004pleasant}.  Four key models are proposed in this paper, based on: whether or not (1) the system has a finite number of IDs; and (2) whether or not each run has a finite number of IDs.  For example, if the system has infinite IDs, then for every integer $N$, there are runs with more than $N$ IDs seen throughout the run.  IDs are assumed to communicate via shared memory, and the paper describes how to use shared memory to implement counters, atomic snapshots, group membership, and mutual exclusion in several of their proposed churn models. These aspects of the model differ substantially from the ABC model because of the problems for which they are designed.   In particular, join and departure event timing is not parameterized as it is in the ABC model, and so resource costs for their results are not given as a function of model parameters such as $\alpha$ and $\beta$.

Third, in 2002, Liben-Nowell, Balakrishnan and Karger~\cite{liben2002analysis} defined the notion of \emph{half-life} as follows. Consider a system with $N$ IDs at time $t$. The time elapsed until another $N$ IDs join is the doubling time from time $t$. The time until $N/2$ IDs that are present in the system at time $t$ depart is the halving time from time $t$. The minimum of the doubling time and the halving time is the half-life from time $t$, and the {\it half-life of the system} is the minimum half-life over all $t$.  

The half-life as defined in this way is closely related to our notion of an epoch, but there are technical differences, which we now describe.  There is always at least one epoch in every half-life.  To see this, first note that after $N/2$ additions or $N/2$ deletions, the symmetric difference has changed by at least $N/2$, which satisfies the criteria to end an epoch.  However, there may be multiple epochs in one half-life, since IDs added over one epoch may be deleted in subsequent epochs, thereby avoiding the criteria needed for the end of a half-life.  

Fourth,  the distribution of ID \defn{session times} can be used to characterize churn; that is, the times for which IDs remain in the system. Smaller session times are indicative of higher churn, and vice-versa.  Real-world measurements can inform the distribution of session times, although these findings are specific to the system in question. For example, a measurement study of the peer-to-peer (P2P) system Gnutella found that session time was distributed exponentially with a mean exceeding two hours, whereas the distribution for the P2P system Kazaa is heavy-tailed, with an average session time of roughly a few minutes~\cite{gummadi2003measurement}. Another example involves the peer-to-peer system KAD~\cite{Stutzbach:2006:UCP:1177080.1177105} and the Bitcoin network~\cite{imtiaz2019churn}, where session times for both are well-fit by a Weibull distribution; however, the parameters of these distributions are very different. Therefore, even within the P2P domain, the session-time distribution---and, thus, the characterization of churn---can differ significantly between specific systems. By comparison, the ABC model does not incorporate a particular session-time distribution; rather, churn is defined more generally via the two parameters $\alpha$ and $\beta$.

\section{Motivating $\alpha$, $\beta$-Churn}\label{s:a-b-churn}

Here, we argue that the Sybil attack necessitates a new model of churn. Then, we motivate our notion of $\alpha$, $\beta$-churn, illustrating how it captures the key aspects of a Sybil attack, while remaining mathematical tractable. 

\subsection{Sybil Attack Demands New Churn Models}\label{s:demand}

During a Sybil attack, the system size may increase as the adversary injects large numbers of IDs.  Additionally, a Sybil adversary can also decrease the system size by removing IDs it controls.  In all cases, an algorithm that depends on static system sizes will be fragile when faced with a Sybil adversary. Thus, models that assume a static system size---such as DNC, TRAIN, and CROWD---seem inappropriate for addressing the Sybil attack.

\subsection{Some Gentle Intuition for the $\alpha$, $\beta$-Churn (\emph{ABC}) Model}

If the system size is not fixed, how should we model churn?  First, it seems clear that there should be no constraints on the timing of events for the bad IDs, since these are controlled by an adversary.

So, how should we constrain the timing of {\it good}  ID join and departure events?  A simple idea is to assume a fixed rate for good events, and to assume that this rate never changes.  This is similar to some models in~\cite{ko2008using}.  However, this is unrealistic in many settings.  For example, a system might experience an unusually high good join rate during certain times of day or the week, or during unpredictable events, such as a sudden spike in popularity in some system service.  Additionally, there may be periods of time during which systems size grows or shrinks {\it non-linearly}.  These phenomena cannot be captured by a fixed event rate.

A  more sophisticated approach is to allow for some change in the good event rate.  For example, the good join rate could itself change according to some separate rate.  But what should be this new change rate be?   

One idea is that the rate of \emph{change} could depend on the good event rate.  For example, if the good event rate is initially $1$ ID per millisecond, then the change rate could be by a factor of $2$ every \emph{millisecond}.  But this seems too fast - it allows the rate to double with each new event.  

An alternative would be to have the change rate set to a factor of $2$, say every $1,000$ events ($1,000$ milliseconds).  But setting the denominator to some absolute number of events is artificial: in systems with large populations, we might not expect much change over $1,000$ events, but in systems with small populations, we might expect significant change over $1,000$ events.

\subsection{The epoch}\label{s:the-epoch}
To resolve this issue, the fundamental time period we use for defining rate of change is the \emph{epoch}.  As defined previously, an epoch is the amount of time over which the symmetric difference of the system population changes by half the initial size of the system (see Section~\ref{s:churn}).  Recall that this definition is closely related to the ``half-life" as defined in~\cite{liben2002analysis}.  

\begin{figure}[t]
\begin{center}
\includegraphics[width=0.9\textwidth]{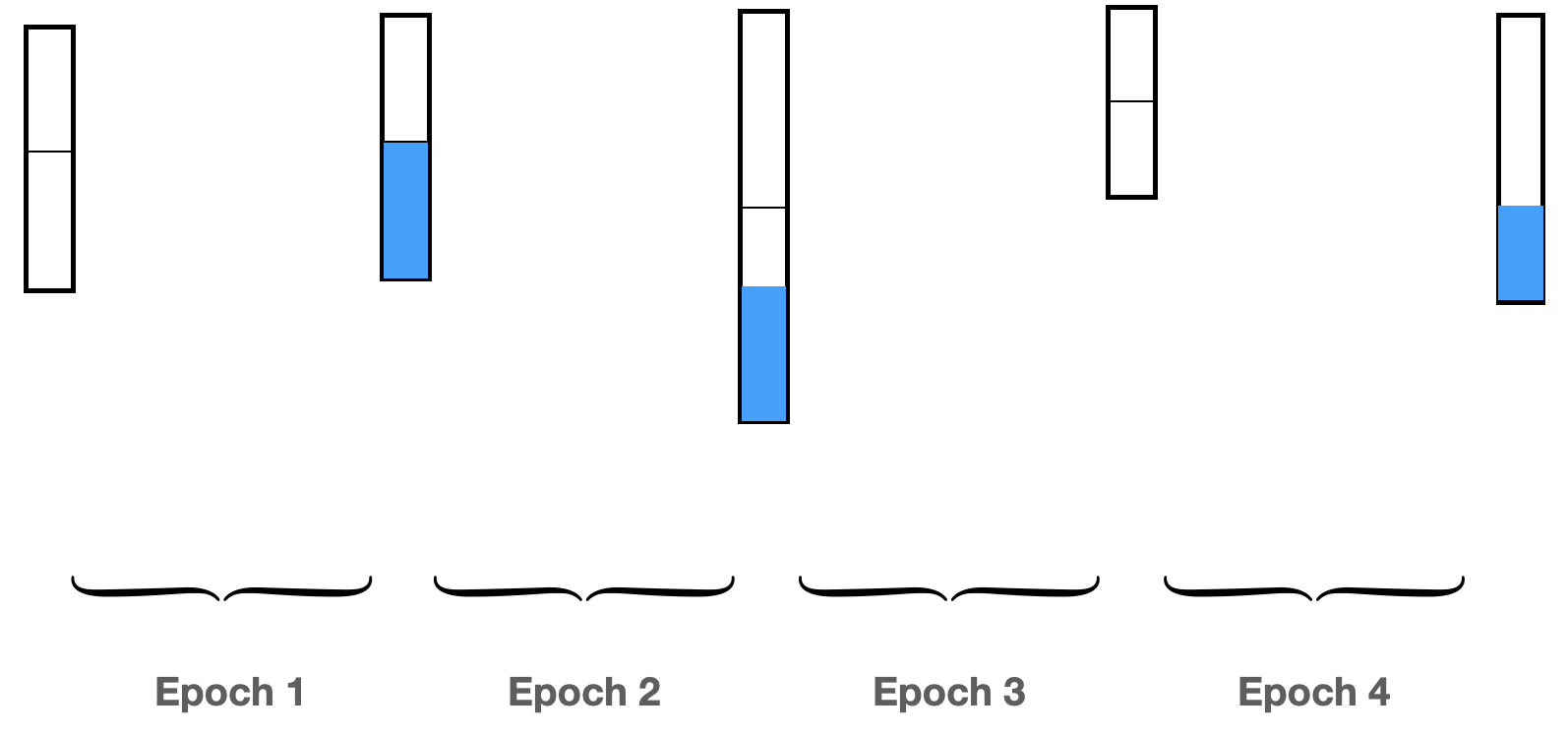}
\end{center}\vspace{-8pt}
\caption{Illustration of four epochs.  Time proceeds from left to right, and each rectangle illustrates the state of the system at the beginning of an epoch, with half of the rectangle delineated by the thin horizontal line. The blue rectangles indicate the new IDs that have been added during the epoch; see the discussion at the end of Section~\ref{s:the-epoch}.
}\label{fig:epoch}
\end{figure}

By letting the event rate change by a fixed amount in each \emph{epoch},
we obtain a churn model that generalizes to both small and large systems.

We now provide some intuition about the notion of an epoch.
Figure~\ref{fig:epoch} illustrates four example epochs.  In this figure, time moves from left to right, and each rectangle illustrates the state of the system at the beginning of an epoch, with half of the rectangle delineated by the thin horizontal line.  

The first epoch ends with half of the IDs changed;  these IDs are illustrated with the blue bar.  Notice that the system size has not changed at the end of this first epoch.  The second epoch ends when new IDs total half of the system size at the start of the epoch.  Notice that the system size has grown by $50\%$.  The third epoch ends when  departing IDs total half of the system size at the start of the epoch.  Notice that the system size has decreased by half.  The fourth epoch ends when new IDs total half the size at the epoch start.  Notice throughout that the system size is multiplied by some value in the range $[.5,1.5]$ during the course of any epoch.

\subsection{$\alpha$-smoothness}
We now describe how the value of $\alpha$ can allow for the good event rate to change across epochs. Recall from Section~\ref{s:churn} that in each epoch $i$, there is a good event rate, $\rho_i$.  Additionally, recall the $\alpha$-smoothness criteria:
\begin{itemize}
\item{\boldmath{$\sma$}\defn{-smoothness}:} $(\AOneL) \epochRate_{i-1}  \leq \epochRate_{i} \leq \AOneH\, \epochRate_{i-1}$.
\end{itemize}

This ensures that the good event rate changes by no more than an $\alpha$-factor from epoch to epoch.  For example, when $\alpha = 2$, if the good event rate is $1$ event per millisecond in epoch $i-1$, then in epoch $i$, the good event rate will be in the range from $1$ event per two milliseconds up to $2$ events per millisecond.

\begin{figure}[t]
\includegraphics[width=1\textwidth]{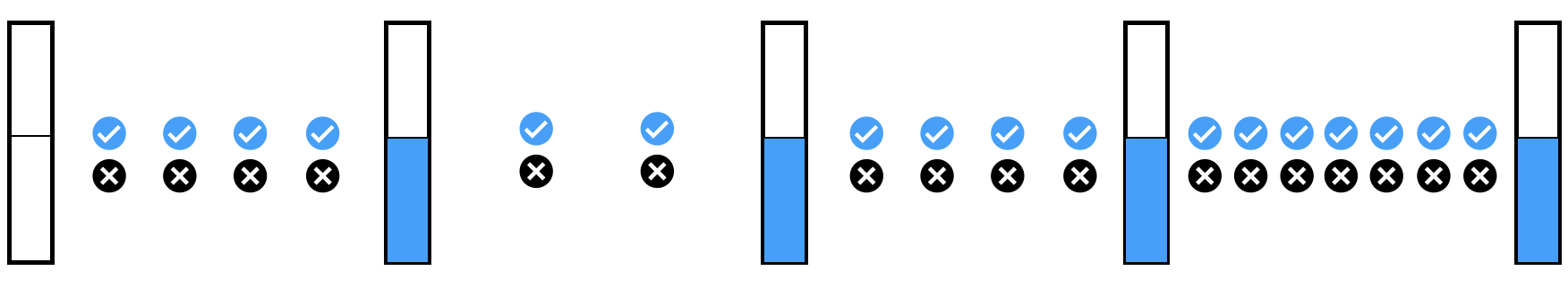}
\caption{Example for $\alpha=2$, $\beta = 1$. 
 Blue checks represent good ID join events; black x's represent good ID deletion events.  In this example, for simplicity, the system size does not change, so the number of good join events equals the number of good deletion events.  Since $\alpha = 2$, the number of events can change by a factor of $2$ from epoch to epoch.  Since $\beta = 1$, all these events are uniformly distributed in time.}\label{fig:ABCAlpha}
\end{figure}

Figure~\ref{fig:ABCAlpha} illustrates an example for $\alpha = 2$.  In this figure, the checks are good ID additions and the x's are good ID deletions.  Again, the system state at the beginning of epochs is represented by a rectangle; the blue half of the rectangle illustrates the half of the system IDs that are new in comparison to the system population at the start of the epoch. For simplicity, this figure illustrates a case where the system size does not change.

All epoch lengths are the same: $1$ hour.  Finally, in order to focus solely on $\alpha$, we have set $\beta = 1$.  This means that all events are spread out evenly over the duration of the epoch.  To keep things simple, in the figure, the epochs all have the same length.  This can happen even when the value of $\rho$ changes, because deletion events can either be for IDs that have been added during the epoch, or from IDs that were around at the start of the epoch.  In the former case, the ID that joined and then departed does not hasten the end of the epoch.

Since $\alpha = 2$, it is possible for the $\rho_i$ values vary  by multiplicative factors of $2$ from one epoch to another.  In the figure, $\rho_1 = 4$, $\rho_2 = 2$, $\rho_3 = 4$ and $\rho_4 = 7$.

A key observation is that, even for a small value such as $\alpha = 2$, the good join  rate can increase (or decrease), \emph{exponentially} over multiple epochs. In particular, over $x$ epochs the event rate can decrease by a factor of $2^{-x}$ or increase by $2^{x}$.

\subsection{$\beta$-smoothness}\label{ss:beta}
In the previous section, good events were evenly spread over time in each epoch.  Next, we show how this can change by discussing the final part of our model: $\beta$-smoothness. Recall the definition from Section~\ref{s:churn}:

\begin{itemize}
\item {\boldmath{$\smb$}\defn{-smoothness}:}  For any duration of $\ell$ seconds in the epoch, the number of good IDs that join is at least $ \lfloor \ell \epochRate_{i} / \beta\rfloor $ and at most $\lceil\beta \ell \epochRate_{i}\rceil$. Also, the number of good IDs that depart during this duration is at most $\lceil\beta \ell \epochRate_{i}\rceil$.
\end{itemize}

\begin{figure}[t]
\includegraphics[width=1\textwidth]{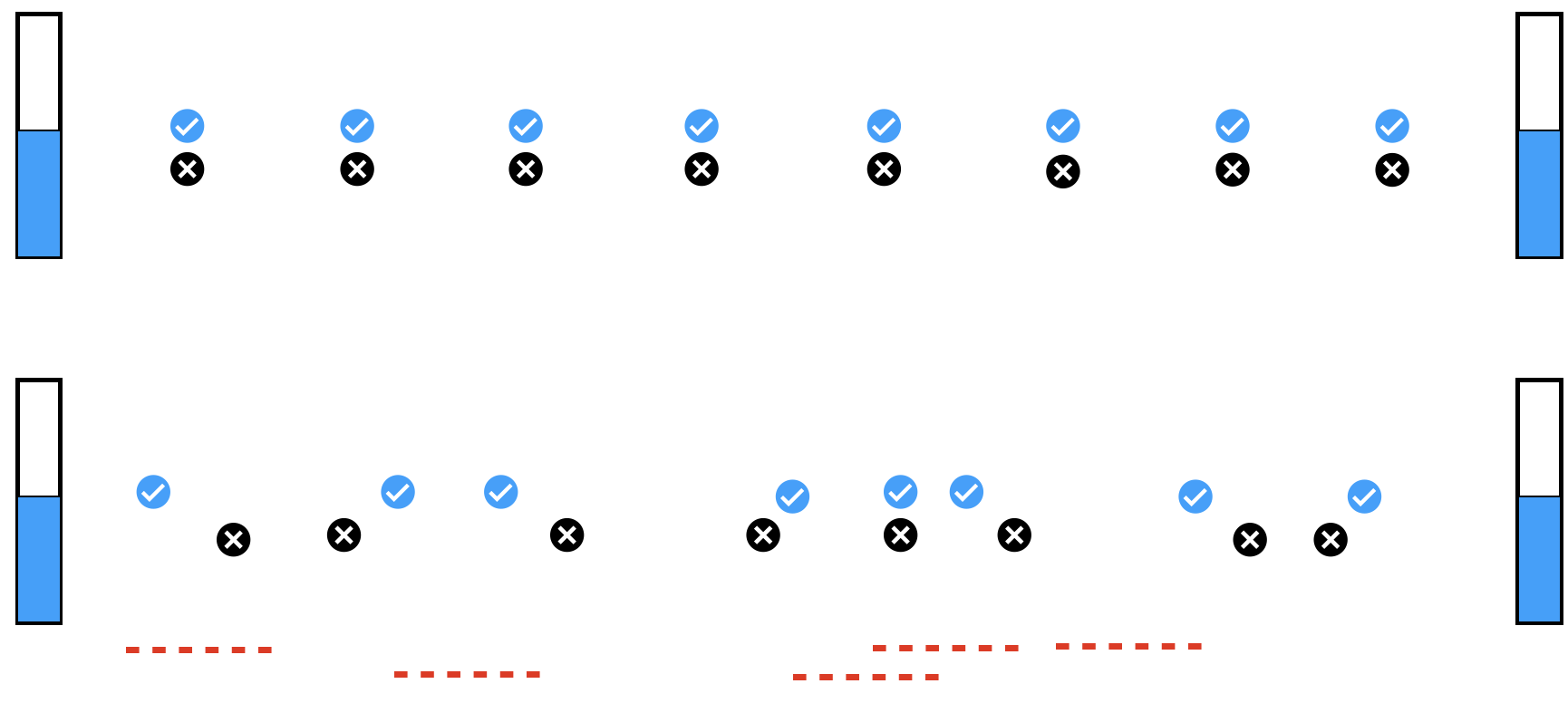}
\caption{Single Epoch with Different $\beta$ values. The dashed red lines at the bottom represent ``durations" of time equal to the epoch length divided by the number of good join events.  Top: $\beta=1$; each duration overlaps exactly $1$ join event and at most $1$ deletion event; events are spread evenly.  Bottom: $\beta = 2$; each duration overlaps between $\floor{1/2} = 0$ and $\ceil{2} = 2$ join events, and at most $2$ deletion events; 
 events may clump.}\label{fig:ABCBeta}
\end{figure}

Figure~\ref{fig:ABCBeta} illustrates how different values of $\beta$ effect the spacing of good events over a single epoch lasting $1$ hour. Both the top and bottom sub-figures illustrate epochs with the same $\rho$ value: $\rho = 8$ events per hour.  In the top figure, $\beta=1$ and in the bottom figure, $\beta = 2$.  

The red dashed lines at the bottom of the figure are a few example ``durations" of length $\ell = 3600/8 = 450$ seconds.  In the top epoch, $\beta$-smoothness for $\beta = 1$ requires that within each such duration, there must be at least $ \lfloor \ell \epochRate \rfloor  = 1$ , and at most $\lceil \ell \epochRate \rceil = 1$ good join event during this red duration.  Additionally, there must also be at most $\lceil \ell \epochRate \rceil = 1$ good deletion.  Note that these bounds hold for each duration in the top subfigure.

In the bottom epoch, $\beta = 2$, so  $\beta$-smoothness requires that within each duration, there must be at least $ \lfloor (1/2) \ell \epochRate \rfloor  = 0$ , and at most $\lceil 2 \ell \epochRate \rceil = 2$ ID join events during this red duration.  Additionally, there must also be at most $\lceil \ell \epochRate \rceil = 2$ good deletions.  Note that these bounds hold for each duration in the bottom subfigure.

In sum, larger values of $\beta$, such as in the bottom epoch, allow for more ``clumping" of the good events.  Smaller values of $\beta$ require the events to be more evenly spread.  Finally, we note that for simplicity, our duration lengths in Figure~\ref{fig:ABCBeta}---the red dashed lines---are all the same length.  However, the $\beta$-smoothness criteria holds for any duration length that is completely contained in the epoch.

\section{Related Work} \label{sec:related-work}

A preliminary version of our results appeared  in~\cite{Gupta_Saia_Young_2021,Gupta_Saia_Young_2019}. Specifically, the problem definition for \defID, model of churn, upper-bound analysis, and experimental results appeared recently in~\cite{Gupta_Saia_Young_2021}, while our lower-bound result (Section~\ref{sec:lower}) appeared in~\cite{Gupta_Saia_Young_2019}.
Our presentation here includes new material as detailed previously in Section~\ref{sec:new}.

\medskip

\noindent\textbf{Sybil Attacks.} There is significant prior work on Sybil attacks~\cite{douceur02sybil}.  For example, see surveys~\cite{mohaisen:sybil,john:soft},~\cite{newsome:sybil,mohaisen:sybil,john:soft} and additional work documenting real-world Sybil attacks~\cite{bitcoin-sybil,6503215,Yang:2011:USN:2068816.2068841}. 

Several results leverage social networks for Sybil defense \cite{yu:sybilguard,mohaisen:improving,wei:sybildefender}, including recent work using machine learning to {\it classify} likely Sybil IDs, such as \SybilExpose~\cite{misra2016sybilexposer} and \defn{\SybilFuse}~\cite{gao2018sybilfuse}.  Since social-network data may not always be available, in this paper, we focus on Sybil defense without it. We also note that, by themselves,  classification methods do not solve \defID. In particular, a classifier that is wrong with even a small probability, say $10^{-6}$, still allows the adversary to obtain a bad majority, over a large number of attempted join events.

That said, in Section~\ref{sec:experiments}, we do show that \ergo combines well with classification algorithms like \SybilFuse.  In particular, \SybilFuse significantly reduces costs for \ergo, when social network data is available to use \SybilFuse (Section~\ref{sec:experiments}, Heuristic 4).  Additionally, \ergo enables a classification algorithm like \SybilFuse to be leveraged to create a full-fledged Sybil defense algorithm that can withstand significant, long-term attack.

Other Sybil defenses use network measurements~\cite{sherr:veracity,liu:mason,Gil-RSS-15}.  These defenses rely on accurate measurements of latency, signal strength, or round-trip times to try to detect Sybil IDs.  Again, since such data may not always be available, \ergo does not rely on its use.  But we expect that, when this type of data is available, these results could also be used to reduce costs for \ergo.

Finally, Danezis et al.~\cite{danezis:sybil} and Scheideler and Schmid~\cite{scheideler:shell} describe containment strategies for Sybil attacks in overlay networks.  However, these results focus on isolating older IDs from newer bad IDs during a Sybil attack, and so do not ensure that the fraction of bad IDs in the network is always bounded.

\medskip

\noindent\textbf{Resource Burning.}  Many resource burning schemes for Sybil defense exist. {\it Computational puzzles} consume CPU cycles~\cite{nakamoto:bitcoin, li:sybilcontrol, andrychowicz2015pow}.  {\it Proof of Space-Time}, requires allocation of storage capacity~\cite{moran2019simple}.  {\it Proof of useful-work} consumes CPU cycles to solve challenges applicable to real-world scientific or engineering problems~\cite{shoker2017sustainable, ball2018proofs}.

A {\it completely automated public Turing test to tell computers and humans apart (CAPTCHA)} is a resource - burning tool where the resource is human effort~\cite{von2003captcha, moradi2015captcha}.  CAPTCHAs of tunable hardness have been proposed~\cite{baird2005scattertype}, as have CAPTCHAs that channel human effort into practical problems such as deciphering scanned words or detecting spam~\cite{von2008recaptcha}.

In a wireless network with multiple communication channels, Sybil attacks can be mitigated via \emph{radio - resource testing} if the adversary cannot listen to all channels simultaneously~\cite{monica:radio,gilbert:sybilcast,gilbert:who}; the resource here is listening capacity.

Finally, we note that Proof of Stake~\cite{GiladHMVZ17,Kiayias2017, ethereum-pos} is \emph{not} a resource burning technique.  It requires that the ``stake" of each ID to be a globally known quantity and thus is likely to remain relevant primarily for cryptocurrencies.  Moreover, even in that domain, it is controversial~\cite{posdm}.  

\medskip

\noindent\textbf{Guaranteed Spend Rate.}
In~\cite{gupta2017proof} and \cite{Gupta_Saia_Young_2019}, Gupta et al. proposed two algorithms \AlgB and \algGM that ensure that the fraction of bad IDs is always small, with respective good spend rates of $O(T+ \jAll)$ and $O(\jAll + \sqrt{T(\jAll+1)})$.  Unfortunately, the second result holds for the case where (1) churn is sufficiently small; and (2) there is a fixed constant amount of time that separates all join events by good IDs (i.e., non-bursty arrivals). \ergo does not require these assumptions.

We also note that, outside of the Sybil attack, several prior works address network security challenges with results that are parameterized by the adversary's cost~\cite{gilbert:making,gilbert:near,king:conflict,bender:how,ICALP15,daniICJournal17,aggarwal2016secure,gilbert:resource,zamani2017torbricks}. Such results are referred to as {\it resource-competitive}, and many examples are provided in the survey by Bender et al.~\cite{Bender:2015:RA:2818936.2818949}.


\section{ERGO}\label{sec:tog}

\begin{figure*}[t]
\centering
\begin{framed}
\begin{minipage}[h]{0.99\textwidth}
\noindent \textbf{\textsc{\underline{E}ntire by \underline{R}ate of \underline{Go}od}} ({\bf{\ergo}})
\medskip
\small

$S(0) \leftarrow$ set of IDs that returned a valid solution to 1-hard
RB challenge. 

\medskip

\noindent $\JoinEst$ is maintained by running $\goodJest$ in parallel to the following code;\\
\noindent $\tau \leftarrow$ time at system initialization;\\
\noindent $\tau'$ is the current time.

\medskip

\noindent For each iteration, do:
\begin{enumerate}[leftmargin=14pt] 
    \colorlet{light-gray}{gray!25}
    \sethlcolor{light-gray}
    \setlength{\lineskip}{0pt}
	\item[1.] Each joining ID is assigned a RB challenge of hardness $1$ plus the number of IDs that have joined in the last $1/\JoinEst$ seconds of the current iteration.
	
	\item[2.] When number of joining and departing IDs in this iteration exceeds $\vert S(\tau) \vert /11$, perform a purge as follows:
			\begin{itemize}[leftmargin=15pt]
			\item[(a)] Issue all IDs a $1$-hard RB challenge. 
			\item[(b)] $\iterIDs(\tau)$ $\leftarrow$  IDs solving this RB challenge in $1$ round.
			\item[(c)]  $\tau \leftarrow \tau'$
			\end{itemize}
            
\end{enumerate}
\end{minipage}
\end{framed}
\vspace{-13pt}\caption{Pseudocode for \ergo.}
\label{alg:gmcom}
\end{figure*}


We begin by walking through the execution of \ergo and providing intuition for its design.  Recall from Section~\ref{sec:model-main} that we are presenting \ergo with coordination provided by a server. Therefore, the server executes the pseudocode for \ergo presented in Figure~\ref{alg:gmcom}. The server initializes system membership with all IDs that solve a 1-hard RB challenge.   Then, execution occurs over disjoint periods of time called \defn{iterations}, where an iteration consists of Steps 1 and 2 in the pseudocode. 

In Step 1,  each ID that wishes to join the system must solve an RB challenge of hardness $1$ plus the number of IDs that joined within the last $1/\JoinEst$ seconds, where $\JoinEst$ is the current good join rate estimate obtained from \goodJest. We call the hardness of this RB challenge the \emph{entrance cost}; intuition for its value is in Section~\ref{sec:entrance-cost}.  Once the ID returns a valid solution, the server adds the ID to a membership set that it maintains and the ID is considered to have successfully joined the system.


Step 1 lasts until the number of IDs that join and depart in the iteration is at least $1/11$ times the number of IDs at the start of the iteration;  we note that the value $1/11$ is not special, as discussed later in Section~\ref{sec:parameter-constant-discussion}.  In Step 2, the server performs a \defn{purge} by resetting system membership to all IDs that solve a 1-hard RB challenge within 1 round.  To do this, the server issues a 1-hard RB challenge to all IDs and then removes from the membership set those IDs that fail to respond with a valid solution within 1 round. Note, again, that the server maintains a membership set of all IDs in the system.


\subsection{Intuition for Entrance Cost}\label{sec:entrance-cost} 

To gain intuition, fix an iteration, assume that $\beta = \Theta(1)$, and let $\jInterval$ be the good join rate during the iteration.  Then, in the absence of attack, all entrance costs are $O(1)$ since $O(1)$ good IDs join on average during $1/\JoinEst \approx 1/\jInterval$ seconds.  


If there is a large attack, the adversary pays more than \ergo.  For example, consider the case where the ratio of bad ID joins to good ID joins is $x$ and these bad ID join events are evenly spread out over time. Thus, the number of bad IDs joining in any $1/\JoinEst$ seconds is about $x$. For each good join event, the adversary pays an entrance cost which is at least $1+2+3 + ... + x = \Theta(x^2)$. In contrast, the good ID that joins in this time pays at most $x+1 = O(x)$; recall that, in the worst case, this good ID joins after the bad IDs. Therefore, over this interval of $1/\JoinEst$ seconds, the good ID pays an entrance cost that is asymptotically the square root of what the adversary pays.
 
A challenge we face in analyzing \ergo is establishing this flavor of result more generally. Nonetheless, we can extend this reasoning a bit further to demonstrate more intuition for why Theorem~\ref{thm:new-main-upper} is plausible. During a large attack, the adversary's spend rate is $T = \Theta(\xi  \jIntervalAll)$, where {\boldmath{$\xi$}} is the average entrance cost, and {\boldmath{$\jIntervalAll$}} is the join rate for all IDs. Then, the good spend rate due to entrance costs is $\Theta(\xi \jInterval$), and the good spend rate due to purges is $\Theta(\jIntervalAll)$. When $\xi  = \jIntervalAll/ \jInterval$, these two costs are balanced, and the good spend rate due to the entrance costs and purge costs is within a constant factor of:

$$\xi \jInterval + \jIntervalAll  \leq  2\jIntervalAll = 2\sqrt{\left(\jIntervalAll\right)^2} = 2 \sqrt{\jIntervalAll  \xi \jInterval} =  2\sqrt{\jInterval T}$$

\noindent where the first step holds by our setting of $\xi$, the third step since $\jIntervalAll = \xi \jInterval$, and the final step since $T = \xi  \jIntervalAll$. 

\smallskip
As a side note, our entrance cost approximates the ratio of the total join rate over the good join rate, which motivates the name {\textsc{\underline{E}ntire by \underline{R}ate of \underline{Go}od}}; \ergo is also the Greek word for work.

\medskip
\noindent
\textbf{Technical Difficulties.}  While the above gives intuition for our analysis, challenges remain.  First, how do we get an estimate of the good join rate when we do not know anything about which IDs are good or bad, when epochs begin or end, or the values of $\alpha$ and $\beta$?  Solving this problem is a key technical difficulty, addressed by our algorithm \goodJest, which we describe and analyze in Section~\ref{s:analGoodJest}.

Second, how can the above intuition for analyzing entrance costs generalize when the good join rate is changing, possibly within each iteration of \ergo?  To handle this, we first show that  our estimate of the good join rate updates at least once every $O(1)$ epochs, so it is never too stale.  This implies that the entrance costs---which make use of the good join rate---yield an advantage over the adversary, as sketched above.  Finally, we make use of Cauchy-Schwartz to upper-bound \ergo's total cost based on the bounds for each iteration, thus completing the upper-bound analysis.

\subsection{Intuition for Purging}\label{s:intuition-purge}

The purpose of purging is to ensure that the fraction of bad IDs in the system is less than $1/6$ at all times. When a $1$-hard RB challenge is issued to all IDs, the adversary can only solve a $k$-fraction within a round and thus keeps at most a $\kappa$-fraction of its bad IDs in the system. Our main result for \ergo result holds for $\kappa\leq 1/18$, and so immediately after each purge the fraction of bad IDs in the system is at most $1/18$. During the iteration, the fraction of bad IDs can increase, but the iteration always ends before this fraction can be reach $1/6$. This reasoning is formalized in Lemma~\ref{lem:badbounded}.


\section{\goodJest}\label{s:goodJest}

\goodJest provides an estimate, {\boldmath{$\JoinEst$}},  of the good join rate, when there is at most a constant fraction of bad IDs. We require that the fraction of bad IDs is less than $1/6$. 

Initially, \goodJest sets $\JoinEst$ equal to the number of IDs at system initialization divided by the total time taken for initialization, where initialization consists of the server issuing a $1$-hard RB challenge to all nodes.  In Section~\ref{s:committee}, we show how to decentralize this algorithm.  The value $t$ is set to the system start time. Throughout the protocol, $t$ will equal the last time that $\JoinEst$ was updated, and $t'$ will be the current time.


\begin{figure}[t!]
\centering
\begin{framed}
\begin{minipage}[h]{0.99\textwidth}
\noindent  
\textbf{\textsc{\underline{Good} \underline{J}oin \underline{Est}imate}}  \textbf{(\goodJest)}
\medskip
\small

In the following, $t'$ is the current time and $S(x)$ is the set of IDs in the system at time $x$.

\smallskip
$t \leftarrow$ time at system initialization.\\
$\JoinEst \leftarrow$ $\vert S(t) \vert$ divided by time required for initialization.\smallskip

Repeat forever: whenever $\vert S(t') \triangle S(t) \vert \geq \frac{5}{12} \vert S(t') \vert$, do: \vspace{-5pt}
            \begin{enumerate}[leftmargin=15pt] 
            \item[1.]  $\JoinEst \leftarrow \vert S(t')\vert /(t'-t)$.\vspace{-7pt}
            \item[2.]  $t \leftarrow t'$.
            \end{enumerate}
\end{minipage}
\end{framed}
\vspace{-13pt}\caption{Pseudocode for \goodJest.}\label{alg:estGoodRate}\vspace{-5pt}
\end{figure}

The pseudocode is presented in Figure~\ref{alg:estGoodRate}. There are two aspects that must be addressed in designing  \goodJest. First, at what points in time should $\JoinEst$ be updated?  This occurs whenever the system membership has changed by a constant factor with respect to the current system size. In particular, $\JoinEst$ is updated when $\vert S(t') \triangle S(t)\vert  \geq \frac{5}{12} \vert S(t')\vert $ holds true;  we note that the value $5/12$ is not special, as discussed later in Section~\ref{sec:parameter-constant-discussion}.  Since join and departure events are ordered, this is equivalent to the property that $\vert S(t') \triangle S(t)\vert  = \lceil \frac{5}{12} \vert S(t')\vert  \rceil$.  We refer to $(t,t']$ as an \defn{interval}. The execution of \goodJest divides time into consecutive, disjoint intervals.

Second, how is $\JoinEst$ updated?  This is done by setting $\JoinEst$ to the current system size divided by the amount of time since the last update to $\JoinEst$.  In particular, we set $\JoinEst \leftarrow \vert S(t')\vert /(t'-t)$. 


\subsection{Intuition for \EstGoodJoin}\label{s:goodJestIntuition}

We provide intuition, using Figure~\ref{fig:gJest}, for how \EstGoodJoin estimates the rate at which good IDs join the system, despite the fact that we do not know a priori which IDs are good and which are bad.  The full, formal analysis is in Section~\ref{s:analGoodJest}.

Fix some interval, and let $S(t)$ and $S(t')$ be the set of good IDs at the beginning and end of the interval, respectively; $s = \vert S(t) \vert $; and $a$ be the number of good IDs that join during the interval.  We provide intuition for why $a$ is always $\Theta(s)$, and thus, why dividing $s$ by the interval length yields an estimate of the good join rate.  For simplicity, we only consider here the case where the system size is fixed, but the IDs change; our intuitive reasoning only provides results in expectation. Our full proof handles changing system size and gives results with high probability (see Section~\ref{s:analGoodJest}).  Symmetric difference is critical;  for example, if we delineated intervals simply by the raw number of joins and deletions, this would allow us to obtain that $a = O(s)$, but not that $a = \Omega(s)$.  We now sketch our result.  


\begin{figure}[t]
\includegraphics[width=1\textwidth]{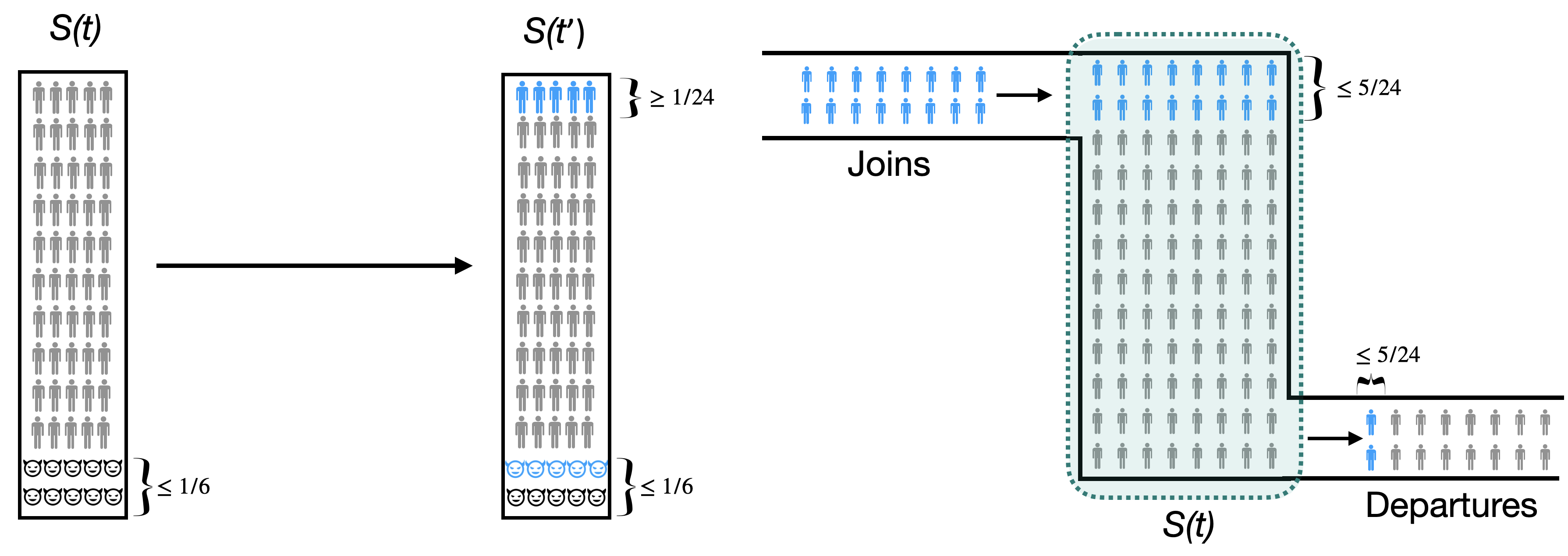}
\caption{Example used to give intuition for \EstGoodJoin, as discussed in Section~\ref{s:goodJestIntuition}. Blue and gray indicate a new and old ID, respectively.  Stick figures and horned-faces indicate good and bad IDs, respectively.}\label{fig:gJest}
\end{figure}


Figure~\ref{fig:gJest} (left) shows why the number of good IDs that join is greater than a constant times the system size at the start of the interval, i.e. why $a = \Omega(s)$.   In the figure, the old IDs are grey and new IDs are blue, where an ID is called \emph{old} if it was in the system at the beginning of the interval, and is \emph{new} otherwise. The good IDs are the stick figures and the bad IDs are the horned faces.  The interval only ends when a $5/24$ fraction of new IDs have joined.  This is true because, since the system size stays constant, the number of departing IDs equals the number of joining IDs, thus, $5/24$ joining and $5/24$ departing IDs yields the total $5/12$-fraction change required to ends an interval.  Note that the fraction of new, good IDs must be at least $5/24 - 1/6 = 1/24$ since the fraction of new, bad IDs is at most $1/6$.  This shows that $a \geq (1/24)s$.

In Figure 6 (left), the system size is $60$. The  number of bad IDs (old and new) is $10$, and the fraction is thus exactly $1/6$. The total number of good new IDs is $5$, giving a fraction that is $5/60 \geq 1/24$.

Figure~\ref{fig:gJest} (right) shows why the number of good IDs that join is smaller than a constant times $s$, i.e. why $a = O(s)$.  This fact does not follow trivially: 
if each new good ID immediately departs, this would increase $a$, but would not alter the symmetric difference.  For simplicity, our figure focuses only on the good IDs, in order to more simply illustrate the interplay between arrivals and departures.
 
The figure illustrates the following.  \textbf{Fact 1:} the fraction of new good IDs in the system is always at most $5/24$ during the interval, otherwise the interval would end.  To see this, first note that \goodJest ends an interval when the symmetric difference is more than a $5/12$ fraction of the system size at the beginning of the interval.  Then, note that if the fraction of new good IDs is ever more than a $5/24$ fraction of the system size, the symmetric difference will exceed a $5/12$ fraction of the beginning system size.

Next, we have the following.  \textbf{Fact 2:} the fraction of new IDs in the set of departing good IDs is at most $5/24$ in expectation.  To see this, first recall that when a good departure occurs, the departing ID is selected uniformly at random from the set of good IDs currently in the system (See Section~\ref{sec:model-main}).  Thus, a departing good ID is new with probability at most equal to the fraction of new, good IDs in the system.  

From these two facts, we note the following.  In expectation, the fraction of good IDs that join and are deleted is at least $a - (5/24)a$, by Fact 2.  Second, the fraction of good new IDs in the system at any point is always no more than $(5/24) s$, by Fact 1.  Putting these together, we get that 
in expectation, $a - (5/24) a \leq (5/24) s$; solving for $a$, we get that $a \leq (5/19) s$. 

The above intuition is formalized in Lemmas~\ref{lem:a_uppbound} and~\ref{lem:a-lowerbound} of Section~\ref{s:analGoodJest}.  The analysis in Section~\ref{s:analGoodJest} solves the following remaining problems: providing sharp concentration bounds; handling changing system sizes; and using the value of $a$ to obtain the good join rate, even when an interval intersects multiple epochs.


\section{Analysis}\label{s:analysis}

 
\begin{table}[t]
\renewcommand{\arraystretch}{1.2} 
\begin{center}
{
\begin{tabular}{ |>{\centering\arraybackslash}p{2cm}|m{10.5cm}|  }
\hline
\rowcolor{LightCyan}\hspace{0pt}{\bf Notation} &  \hspace{3.8cm}{\bf Definition}  \\
\hline
$A \triangle B$ & The symmetric difference between sets $A$ and $B$.  \\
\hline
$\epochRate_i$ & The good join rate in epoch $i$.\\
\hline
$\sma$  &  The good join rate between two consecutive epochs differs by at most an $\alpha$-factor.   \\
\hline
$\smb$ & The number of good IDs that join or depart during $\ell$ consecutive seconds within an epoch differs by at most a $\beta$-factor from $\ell$ times the good join rate of the epoch.  \\
\hline
$\AdvPower$ & In any single round where all IDs are solving challenges, the adversary can solve a $\AdvPower$-fraction of the challenges. \\
\hline
$\epsilon$ & In any single round, at most an $\epsilon$-fraction of good IDs may depart, for $\epsilon<1/12$. \\
\hline
$n_0$ & The minimum number of good IDs in the system at any time. \\
\hline
$\gamma$ & The system lifetime is defined over $n_0^{\lifetime}$ joins and departures, for any fixed constant $\gamma>0$.\\
\hline
$\advAveCost$ & The adversary's spend rate over the system lifetime.\\
\hline
$\joinRate$ &  The join rate of good IDs over the system lifetime.\\
\hline
$\joinRate^B$ &  The join rate of bad IDs over the system lifetime (specific to the lower bound argument in Section~\ref{sec:lower}).\\
\hline
\end{tabular}
}\caption{Table of commonly used notation.}\label{table:notation}
\end{center}
\end{table}


In this section, we provide full proofs of our results. We start with the analysis of \goodJest (Section~\ref{s:analGoodJest}).  Next, we prove the correctness properties for \ergo, and prove the spending rate upper-bounds for \ergo, when using the estimate provided by \goodJest (Section~\ref{s:analErgo}).  For ease of reference, we collect our commonly used notation in Table~\ref{table:notation}. Finally, we conclude this section with a discussion of our choice of values for several parameters and constants used in our analysis and algorithm design (Section~\ref{sec:parameter-constant-discussion}).

\subsection{Analysis of \goodJest}\label{s:analGoodJest}

Why does \goodJest provide a close estimate of the good join rate?   Recall that \goodJest divides time into intervals. We say that an interval \defn{intersects} an epoch if there is a point in time belonging to both the interval and the epoch. 

\begin{lemma}\label{lem:interval-epochs}
	An interval intersects at most two epochs. 
\end{lemma}

\begin{proof}
	Assume that some interval starts at time $t_0$ and intersects at least three epochs; we will derive a contradiction. Given this assumption, there must be at least one epoch entirely contained within the interval.  Consider the earliest such epoch, and let it start at time $t_1 \geq t_0$ and end at time $t_2> t_1$.  Observe that:
 	\begin{eqnarray*} 
 	\vert S(t_2)  \triangle S(t_0) \vert  &\geq & \vert  G(t_2) \triangle G(t_1)\vert \\
 	&\geq & \left(\frac{1}{2}\right)\vert  G(t_2)\vert  \\
 	&\geq & \left(\frac{1}{2}\right)\left(\frac{5}{6}\right)\vert S(t_2)\vert \\
 	&= &\left(\frac{5}{12}\right) \vert S(t_2)\vert 
	\end{eqnarray*}
	In the above, step one holds since $\vert S(t_2)  \triangle S(t_0) \vert \geq  \vert S(t_2) \triangle S(t_1)\vert $; step two holds by the definition of an epoch; and the second to last step holds given that the fraction of bad IDs is always less than $1/6$.

	But the above inequalities show that $\vert S(t_2)  \triangle S(t_0) \vert $ $\geq \frac{5}{12} \vert S(t_2)\vert $.  Therefore, the interval ends by time $t_2$, and there can be no third epoch intersecting the interval; this contradiction completes the argument.  
\end{proof}

At this point, it is useful to foreshadow the relationship between epochs (recall Section~\ref{s:churn}), intervals (used by \EstGoodJoin), and iterations (used in \ergo). Lemma~\ref{lem:interval-epochs} establishes the ``translation'' between the first two. Later, in Section~\ref{s:analErgo} (Lemma~\ref{lem:iterinter}), we prove the translation between intervals and and iterations. Figure~\ref{fig:EpochIntervalIteration} depicts the relationship between epochs, intervals, and iterations.

Why is this translation necessary? In this section, we will show that \EstGoodJoin achieves an estimate of the good join rate that is parameterized by factors of $\sma$ and $\smb$. These factors---which impact the accuracy of the estimate provided by\EstGoodJoin---arise because an interval may overlap more than one epoch. When we go from intervals to iterations, we will similarly accrue additional factors of $\sma$ and $\smb$ in our analysis for the good spend rate under \ergo, since an iteration can overlap more than one interval (see the analysis in Section~\ref{s:analErgo} starting with Lemma~\ref{lem:num_sub_intervals}).


\begin{figure*}[t]
\centering 
\includegraphics[trim = 1cm 0cm 1cm 0, width=0.80\textwidth]{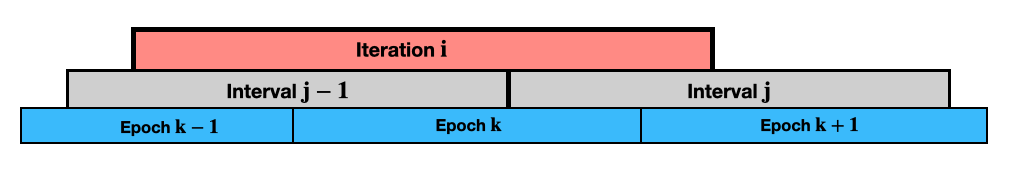}
\caption{\small \emph{Epochs} derive from our model of churn (Section~\ref{s:churn}).  \emph{Intervals} derive from \goodJest, specifically the times at which it sets the variable $\JoinEst$ (Figure~\ref{alg:estGoodRate}, Step 1). \emph{Iterations} derive from \ergo, specifically the duration between purges (Figure~\ref{alg:gmcom}, Step 2).}
  \label{fig:EpochIntervalIteration}  
\end{figure*}

For the remainder of this section, fix an interval that starts at time {\boldmath{$t$}} and ends at time {\boldmath{$t'$}}. Let {\boldmath{$a$}} be the number of good IDs that have joined during the interval.  All lemmas hold \whp in $n_0$.  

\begin{lemma}\label{lem:boundsize}
Assume that $n_0\geq 120$. Then, $\vert S(t')\vert  \geq \frac{7}{10}\vert S(t)\vert $. 
\end{lemma}

\begin{proof}
By the definition of an interval, we know that: 
\begin{align*}
    \vert S(t') \triangle S(t)\vert  &\leq \left\lceil \frac{5}{12}\vert S(t')| \right\rceil  \leq \frac{5}{12}\vert S(t')\vert  + 1
\end{align*}
Note that $\vert S(t') \triangle S(t)\vert  \geq \vert S(t) - S(t')\vert $, which implies:
\begin{align}\label{eq:defInt}
    \vert S(t) - S(t')\vert  &\leq \frac{5}{12}\vert S(t')\vert  + 1
\end{align}
Moreover, $\vert S(t)- S(t')\vert  \geq \vert S(t)\vert  - \vert S(t')\vert $. Rearranging, we get: 
	\begin{align*}
		\vert S(t')\vert  &\geq \vert S(t)\vert  - \vert S(t)- S(t')\vert \\
			&\geq \vert S(t)\vert  - \left(\frac{5}{12}\vert S(t')\vert +1\right)
	\end{align*}
	where the second step follows from Inequality~\ref{eq:defInt}. Finally, isolating $\vert S(t')\vert $ in the last inequality, we get:
	\begin{align*}
		\vert S(t')\vert  &\geq 
		\frac{12}{17}\left(\vert S(t)\vert -1\right) \geq \frac{12}{17}\left(\frac{119}{120}\vert S(t)\vert \right) \geq \left(\frac{7}{10}\right)\vert S(t)\vert 
	\end{align*}
	\noindent where the second inequality holds so long as $n_0 \geq 120$.  
\end{proof}

The next lemma is one of the more technically challenging in our analysis.  Recall that $S(\tau)$ is the set of all IDs in the system at time $\tau$.  In order to upper bound the number of joining good IDs, we need to first upper bound the number of new, good IDs that depart, where an ID is new if it has joined in the current interval.  The key technical difficulty is establishing this bound \whp.  To do so, we compute the expected number of departing new, good IDs, and then use a stochastic dominance argument and Chernoff bounds to show tight concentration around this expectation. 

\begin{lemma}\label{lem:a_uppbound} 
Assume that $n_0\geq \max\{6000, (720(\gamma+1))^{4/3} \}$. Then, for any interval, $a \leq 11\vert S(t)\vert  + 2$.
\end{lemma}

\begin{proof}
Note that:
\begin{align*}
	\left \lceil \frac{5}{12} \vert S(t')\vert  \right \rceil & = \vert S(t') \triangle S(t)\vert   \geq |G(t') \triangle G(t)| \geq\vert G(t') - G(t)\vert 
\end{align*}

\noindent where the first step holds by the definition of an interval.  Thus, we have:
\begin{align*}
	\vert G(t') - G(t)\vert  & \leq \left \lceil \frac{5}{12} \vert S(t')\vert  \right \rceil  < \left(\frac{1}{2}\right) \vert G(t')\vert  + 1
\end{align*}

The last inequality holds given that the fraction of bad IDs is always less than $1/6$,  so  $\frac{\vert G(t')\vert }{\vert S(t')\vert } > \frac{5}{6}$ implies that $\vert S(t')\vert  < \frac{6}{5} \vert G(t')\vert $.  This gives our first key inequality:


\begin{equation} \label{e:numNew}
\vert G(t') - G(t)\vert  < \left(\frac{1}{2}\right) \vert G(t')\vert  + 1
\end{equation}

\noindent
We can bound the probability that a good ID deleted at time $\tau \geq t$ is from the set $G(\tau) - G(t)$ by dividing both sides of Equation~\ref{e:numNew} by $G(\tau)$ to get
\begin{align*}
    \frac{|G(\tau) - G(t)|}{|G(\tau)|} & \leq \frac{1}{2} + \frac{1}{|G(\tau)|} \\
    & \leq 31/40
\end{align*}
\noindent where the final inequality holds assuming that $1/|G(\tau)| \leq 11/40$, which holds when $n_0 \geq 5$.  From this, we know that the probability that a good ID deleted at time $\tau$ is from $G(t)$ is at least $1-31/40 = 9/40$.

Let $d$  be the number of good IDs that have departed in the interval.  Let the random variable $X$ be the number of IDs in $G(t)$ that have departed during the interval. It follows that $E(X) \geq \frac{9}{40} d$, for $n_0 \geq 5$. Additionally, $X$ stochastically dominates a simpler random variable that counts the number of successes when there are $d$ independent trials, each succeeding with probability $\frac{9}{40}$. 

Hence, by a standard Chernoff bound~\cite{dubhashi:concentration}, we have:
\begin{align*}
Pr(X < (1-\delta)(9/40)d) &\leq e^{-\delta^2 (9/40)d/2}\\
& = e^{-(1/81)(9/40)d/2}\\
&= e^{-d/720}
\end{align*}
\noindent where the second inequality follows from setting $\delta=1/9$. Therefore, when $d \geq \vert G(t)\vert \geq n_0 $, it follows that:
\begin{align*}
Pr(X < (1/5)d) &\leq  e^{-d/720}.
\end{align*}
Therefore, $X \geq \frac{1}{5} d$, with probability of failure at most $e^{-n_0/720}$.  This probability of failure is at most $1/n_0^{\gamma+1}$ for $n_0\geq 720(\gamma+1) \ln n_0$.  To derive a sufficient lower bound on $n_0$, note that 
$$n_0/\ln n_0 \geq n_0/n_0^{1/4} = n^{4/3} \geq 720(\gamma+1)$$
\noindent where the first inequality holds for $n_0\geq 6000$, and the last inequality holds so long as $n_0 \geq (720(\gamma+1))^{4/3} \approx 6454 (\gamma+1)^{4/3}$. By a union bound, $X \geq \frac{1}{5}d$ over all intervals during the lifetime of the system, with probability of failure at most $1/n_0$.

Clearly, $X \leq \vert G(t)\vert $.  So by the above, we have that, \whp, $\frac{1}{5} d \leq \vert G(t)\vert $, which gives:
\begin{equation} \label{e:deletions}
d \leq 5 \vert G(t)\vert 
\end{equation}


Since the number of new good IDs in $S(t')$ is at least $a - d$, then $\vert G(t') - G(t)\vert  \geq a - d$.  Thus:
\begin{align*}
	a & \leq \vert G(t') - G(t)\vert  + d\\
 	&\leq \left(\frac{1}{2}\vert G(t')\vert  + 1 \right) + 5\vert G(t)\vert \\
 	&\leq \frac{1}{2}\left(\vert G(t)\vert  + a \right)  + 1 + 5\vert G(t)\vert \\
 	&\leq \left(\frac{11}{2}\right)\vert G(t)\vert  + \frac{a}{2} + 1
\end{align*}

In the above, the second step follows by applying inequalities~\ref{e:numNew} and~\ref{e:deletions}, and  the third step by noting that $\vert G(t')\vert  \leq \vert G(t)\vert  + a$.  Finally, the lemma follows by isolating $a$ in the last inequality, to get $a \leq 11\vert G(t)\vert  + 2 \leq 11\vert S(t)\vert  + 2$. 	
\end{proof}

\begin{lemma}\label{lem:a-lowerbound} 
Assume that $n_0\geq \max\{1681, (41\smb)^2\}$. Then, $a\geq \dfrac{\vert S(t')\vert }{84(1+\smb^2)}- 2 \geq 8$.
\end{lemma}

\begin{proof}
    Let $d$  be the number of good IDs that have departed in the interval. We start by proving that:
    \begin{align}\label{eq:d_upper}
        d & \leq \smb^2(a + 2)+2.	
    \end{align}
    By Lemma \ref{lem:interval-epochs}, an interval intersects at most two epochs. If two epochs are intersected, let $\rho,\rho'$ be the good join rates over the two epochs intersected, and $\ell,\ell'$ be the lengths of the intersection.  If a single epoch is intersected, let $\rho$ and $\rho'$ both equal the good join rate over that epoch, and let $\ell,\ell'$ both be half the length of the intersection of the interval and the epoch.  Then, in every case, from $\smb-$smoothness, we have:
	\begin{align*}
		a \geq \left\lfloor\frac{\rho\ell}{\smb}\right\rfloor + \left\lfloor\frac{\rho'\ell'}{\smb}\right\rfloor \geq \frac{\rho\ell + \rho'\ell'}{\smb} - 2
	\end{align*}
	For which:
	\begin{align}	\label{eq:newupper}
	    \rho\ell + \rho'\ell' & \leq {\smb} (a+2)
	\end{align}
	
 	We can bound the number of departures using $\smb-$ smoothness:
 	\begin{align*}
 		d &\leq \left\lceil\smb \rho\ell \right\rceil + \left\lceil\smb\rho'\ell' \right\rceil \leq \smb(\rho\ell + \rho'\ell') + 2 \leq \smb^2 (a+2) + 2
 	\end{align*}
 	where the last step follows from Inequality \ref{eq:newupper}, and this yields Equation~\ref{eq:d_upper}.  Next, note that:
	\begin{align}\label{eq:uppersymm}
		\vert G(t') \triangle G(t) \vert  &=  \vert G(t') - G(t)\vert  + \vert G(t) - G(t')\vert \nonumber\\
		     &\leq a + d \nonumber \\
			&\leq a + \smb^2(a+2) + 2 \nonumber \\
			&\leq (1 + \smb^2)(a + 2) 
	\end{align}
	where the second to last step follows from Equation~\ref{eq:d_upper}.

   Since the sets of good and bad IDs are disjoint, we have $\vert S(t')\triangle S(t)\vert  = \vert G(t') \triangle G_t\vert $ $+ \vert B(t') \triangle B(t)\vert $. Rearranging, we get: 	
    \begin{align}\label{eq:lowersymm}
		\vert G(t') \triangle G_t\vert 
		&= \vert S(t')\triangle S(t)\vert  - \vert B(t') \triangle B(t)\vert \nonumber\\
		&\geq \left( \frac{5}{12}\right)\vert S(t')\vert  - \vert B(t') \triangle B(t)\vert \nonumber\\
		&= \left(\frac{5}{12}\right)\vert S(t')\vert  - (\vert B(t') - B(t)\vert  + \vert B(t) - B(t')\vert ) \nonumber\\
		&\geq \left(\frac{5}{12}\right)\vert S(t')\vert  - \frac{\vert S(t')\vert }{6} -\frac{\vert S(t)\vert }{6}\nonumber\\
		&\geq \left(\frac{5}{12}\right)\vert S(t')\vert  - \frac{\vert S(t')\vert }{6} - \frac{10}{7}\left(\frac{\vert S(t')\vert }{6}\right)\nonumber\\
		&\geq \frac{\vert S(t')\vert }{84}
	\end{align}

    \noindent The second step follows from the definition of an interval.  The third step by definition of symmetric difference. The fourth step follows by the fact that the fraction of bad IDs is always less than $1/6$, so  $\vert B(t')$ $ - B(t)\vert  \leq \frac{\vert S(t')\vert }{6}$ and $\vert B(t) - B(t')\vert  \leq \frac{\vert S(t)\vert }{6}$.  The fifth step follows from Lemma \ref{lem:boundsize}.
	
	\smallskip
	
	Finally, combining Inequality~\ref{eq:uppersymm} and Inequality~\ref{eq:lowersymm}, we get:
	\begin{align*}
		(1+\smb^2)(a+2) \geq \frac{\vert S(t')\vert }{84}
	\end{align*}
	
	On isolating $a$ in the above, we get:
	$$a \geq \left(\frac{1}{84(1+\smb^2)} \right) \vert S(t')\vert  - 2$$ 
	Since $n_0 \geq (41\smb)^2$, we have that $\smb^2 \leq n_0/1681$.  So, we get:
	
	\begin{align*}
	a & \geq \left(\frac{1}{84(1+\frac{n_0}{1681})} \right)\vert S(t') \vert - 2\\
	  & \geq \left(\frac{1}{84(\frac{2n_0}{1681})} \right)\vert S(t') \vert - 2\\
	   & \geq \left( \frac{10}{n_0}  \right)\vert S(t') \vert - 2\\
	   & \geq \left( \frac{10}{n_0}  \right)\vert n_0 \vert - 2\\
	   & \geq 8
	\end{align*}
	\noindent where the second line follows for $n_0\geq 1681$ and the last line follows since $\vert S(t')\vert  \geq n_0$.
\end{proof}

	

\noindent We now bound $\JoinEst$ with respect to the good join rate over the interval. 

\begin{lemma}\label{lem:boundJoinEst}
Let $\JoinEst$ be the estimated join rate at the end of any interval and $\jInterval$  be the good join rate during the interval. Then:
$$\jInterval/21 \leq \JoinEst \leq 210\smb^2 \jInterval$$ 
\end{lemma}

\begin{proof}
	At the end of the interval, \goodJest sets the estimate of the good join rate as:
	\begin{align*}
	    \JoinEst &= \frac{\vert S(t')\vert }{t' - t}\\ 
	        &\geq \frac{7}{10}\left(\frac{\vert S(t)\vert }{t'-t}\right)\\
	        &\geq \frac{7}{10}\left(\frac{(a-2)/11}{t'-t}\right)\\
	        &\geq \frac{7}{110}\left(\frac{a- (a/4)}{t'-t}\right)\\
	        & \geq \frac{21}{440}\left(\frac{a}{t'-t}\right)\\
	        &\geq \frac{\jInterval}{21}
	\end{align*}
The second step follows from Lemma~\ref{lem:boundsize}, the third step from Lemma \ref{lem:a_uppbound} using $a\leq 11 \vert S(t)\vert  + 2$, and the fourth step from  Lemma~\ref{lem:a-lowerbound} using $a\geq 8$. Similarly:
	\begin{align*}
		\JoinEst &= \frac{\vert S(t')\vert }{t'-t}\\
		    &\leq \frac{84(1+\smb^2)(a+2)}{(t'-t)}\\
		    &\leq {84(1+\smb^2)}\left(\frac{5a}{4(t'-t)}  \right)\\
		    &\leq 210 \smb^2 \jInterval
	\end{align*}
The second and third steps follow from Lemma~\ref{lem:a-lowerbound} using $a \geq \frac{\vert S(t')\vert }{84(1+\smb^2)}- 2$  and $a\geq 8$; the last step holds since $\smb \geq 1$ implies that $1+\beta^2 \geq 2 \beta^2$.
\end{proof}

\begin{lemma}\label{lem:epoch-interval-lim}
Consider an epoch that intersects any interval. Suppose $\rho$ is the good join rate over the  epoch, and $\jInterval$ is the good join rate over the interval. Then:
$$ \frac{4}{5\sma\smb}\rho \leq\jInterval \leq \frac{8}{3}\sma\smb\rho$$  
\end{lemma}

\begin{proof}
    Fix an interval that starts at time $t$ and ends at time $t'$. By Lemma~\ref{lem:interval-epochs}, we know the interval intersects at most two epochs. Let $\rho$ and $\rho'$ be the join rate of good IDs over the two epochs intersecting the interval over say lengths $\ell$ and $\ell'$, respectively.  If only a single epoch is intersected, let $\rho'$ and $\ell'$ both be $0$.  Then, by $\smb$-smoothness (Definition~\ref{def:smoothness}), we have:

	\begin{align*}
		\iJRate &\geq \frac{1}{t'-t} \left(\left\lfloor\frac{\rho\ell}{\smb}\right\rfloor + \left\lfloor\frac{\rho'\ell'}{\smb}\right\rfloor\right) \\
		 &\geq \frac{1}{t'-t} \left(\frac{\rho\ell}{\smb} +  \frac{\rho'\ell'}{\smb} - 2\right) \\
		 		&\geq \frac{1}{\smb}\left(\rho\left(\frac{\ell}{t'-t}\right) +  \rho'\left(\frac{\ell'}{t'-t}\right)\right)-\frac{2}{t'-t}\\
		 &\geq \frac{1}{\smb}\left(\rho\left(\frac{\ell}{t'-t}\right) + \left(\frac{\rho}{\sma}\right) \frac{\ell'}{t'-t}\right)-\frac{2}{\ell}\\
		&\geq \frac{\rho}{\smb\sma}\left(  \frac{\ell+\ell'}{t'-t} \right)-\frac{(\iJRate(t'-t)/4)}{t'-t}\\
		&\geq \frac{\rho}{\sma\smb}-\frac{\iJRate}{4}
	\end{align*}
	The fourth step follows from $\sma$-smoothness (Definition~\ref{def:smoothness}), and the fifth step follows from Lemma \ref{lem:a-lowerbound}, specifically that $a \geq 8$. Isolating $\iJRate$ in the last inequality, we get the lower bound.
	
	\medskip
		
	\noindent Next, we prove the upper bound. Using $\smb$-smoothness:
	\begin{align*}
		\iJRate &\leq \frac{1}{t'-t}\left( \left\lceil\smb\rho\ell\right\rceil + \left\lceil\smb\rho'\ell'\right\rceil \right)\\ 
		&\leq  \smb\left( \rho\left(\frac{\ell}{t'-t} \right)+ \sma\rho\left(\frac{\ell'}{t'-t}\right) \right) + \frac{2}{t'-t}\\
		&\leq  \smb\left( \rho\left(\frac{\ell}{t'-t} \right)+ \sma\rho\left(\frac{\ell'}{t'-t}\right) \right) + \frac{(\iJRate(t'-t)/4)}{(t'-t)}\\
		&\leq (1+\sma)\smb\rho + \frac{\iJRate}{4}
	\end{align*}
The second inequality follows from the $\sma$-smoothness, and the third inequality follows from Lemma \ref{lem:a-lowerbound}, specifically that $a \geq 8$. To see the last step, note that since $\alpha\geq 1$, $1+\alpha\leq 2\alpha$, so we have:
\begin{align*}
		\iJRate &\leq (2\sma)\smb\rho + \frac{\iJRate}{4}.
\end{align*}
Isolating $\iJRate$ in the last inequality, we obtain the upper bound.
\end{proof}

The next lemma makes use of Lemmas~\ref{lem:interval-epochs},~\ref{lem:boundJoinEst} and \ref{lem:epoch-interval-lim}. 

\begin{lemma}\label{lem:prev_current}
Let $\JoinEst$ be the estimated join rate at the end of the interval and $\jInterval$ be the join rate during the \emph{next} interval. Then:
$$ 1/\left( 70 \sma^3\smb^2 \right) \jInterval \leq \JoinEst \leq 700\sma^3\smb^4 \jInterval$$
\end{lemma}

\begin{proof}
    By Lemma~\ref{lem:interval-epochs}, intervals $i$ and $i-1$ will intersect at most $4$ epochs.  Let $\rho_1$ and $\rho_2$ be the join rate of good IDs over the two epochs intersecting interval $i-1$.  If there is only one epoch intersected, let $\rho_2 = \rho_1$.  Let $\rho_3$ be the good join rate over the first epoch that interval $i$ intersects.

	\medskip
	
	\noindent{\bf Lower Bound.} Applying Lemma~\ref{lem:epoch-interval-lim} to interval $i-1$, we have:
	\begin{align}\label{eq:lowprev}
		\iJRate_{i-1} &\geq \frac{4}{5\sma\smb}\rho_2
	\end{align}
	\noindent Applying Lemma~\ref{lem:epoch-interval-lim} to interval $i$, we have:
	\begin{align*}
		\iJRate_i &\leq \frac{4}{3}\left(1 + \sma\right)\smb\rho_3\nonumber \\
		 &\leq \frac{4}{3}\left(1 + \sma\right)\smb\sma\rho_2 \nonumber\\
		 &\leq \frac{4}{3}(1+\sma)\sma\smb\left(\frac{5\sma\smb}{4}\iJRate_{i-1}\right) \nonumber\\
		 &= \frac{5}{3}(1+\sma)\sma^2\smb^2\iJRate_{i-1} \nonumber\\
		 &\leq \frac{10}{3}\sma^3\smb^2\iJRate_{i-1}\\
		 &\leq \frac{10}{3}\sma^3\smb^2 \left( 21 \JoinEst_{i-1} \right)\\
		 &\leq 70\sma^3\smb^2\JoinEst_{i-1}
	\end{align*}
	The second step follows from $\sma$-smoothness, the third step follows by isolating $\rho_2$ in Inequality \ref{eq:lowprev}, the fifth step holds since $\sma \geq 1$ from $\sma$-smoothness, the sixth step follows from Lemma \ref{lem:boundJoinEst}, since $\iJRate_{i-1} \leq 21 \JoinEst_{i-1}$.
	
	\smallskip
    
	\noindent Finally, isolating $\JoinEst_{i-1}$ in the last step of the above equation, we obtain the lower bound of our lemma statement. 
	
	\medskip

	\noindent{\bf Upper Bound.} Similarly, by Lemma \ref{lem:epoch-interval-lim}, the good-ID join rate in interval $i-1$ is:
	\begin{align}\label{eq:upperprev}
		\iJRate_{i-1} &\leq \frac{4}{3}\left(1+\sma\right)\smb\rho_2  
	\end{align}
	\noindent and for interval $i$ is:	\begin{align*}
		\iJRate_i &\geq \frac{4}{5\sma\smb}\rho_3\\
		& \geq \frac{4}{5\smb\sma^2}\rho_2\\ 
		& \geq \frac{4}{5\smb\sma^2}\left(\frac{3}{4(1+\sma)\smb}\iJRate_{i-1}\right) \\
		& \geq \frac{3}{10\sma^3\smb^2}\iJRate_{i-1}\\
		& \geq \frac{\JoinEst_{i-1}}{700\sma^3\smb^4}.
	\end{align*}
	The second step follows from the $\sma$-smoothness, the third step follows from inequality \ref{eq:upperprev}, and the fourth step holds since $\sma \geq 1$ from $\sma$-smoothness. The fifth step follows from Lemma \ref{lem:boundJoinEst}, since: $\iJRate_{i-1} \geq \JoinEst_{i-1}/(210 \smb^2)$.
	
	\medskip
	
	\noindent Finally, isolating $\JoinEst_{i-1}$ in the last step of the above equation, we obtain the lower bound. 
\end{proof}

\noindent Now we give the proof of Theorem~\ref{t:JoinEst}.  We note that this theorem holds with high probability in $n_0$ for any interval, and also, by a union bound, it holds for the entire lifetime of the system over all intervals.

\begin{proof}
	Fix a time step $\tau$ in the interval. Let $\rho_{\tau}$ be the good join rate over the epoch containing $\tau$. Combining the bounds on $\JoinEst$ from Lemma \ref{lem:prev_current} with those from Lemma \ref{lem:epoch-interval-lim},we get:
	\begin{equation*}\label{eq:final}
		\left(\frac{4}{5\sma\smb}\right) \left(\frac{1}{70\sma^3\smb^2}\right)\rho_{\tau} \leq \JoinEst \leq 700 \sma^3 \smb^4 \left( \frac{8}{3}\sma\smb\right) \rho_{\tau}
	\end{equation*}
	Simplifying this equation yields the result. 
\end{proof}

\subsection{Analysis of \ergo} \label{s:analErgo}

For any iteration $i$, let {\boldmath{$B_i$}} and {\boldmath{$G_i$}} respectively denote the number of bad and good IDs in the system at the end of iteration $i$, and we let {\boldmath{$N_i = B_i + G_i$}}. 


\begin{lemma}\label{lem:bound_b}
    For all iterations $i \geq 0$:  
    $$B_i \leq N_i\kappa/({1-\epsilon}).$$
\end{lemma}

\begin{proof}
 Recall that the server issues all IDs a $1$-difficult challenge prior to iteration $0$. This ensures the fraction of bad IDs at the end of iteration $0$ is at most $\kappa$. Thus, our lemma statement holds for $i=0$.
Next, note that in Step 2 ending each iteration $i>0$, the server issues a $1$-difficult challenge to all IDs. Let $N^s_i$ be the number of IDs in the system at the start of the purge. Recall from Section \ref{sec:model-main} that at most an $\epsilon-$fraction of good IDs can depart in a round. Thus, $N_i \geq  N^s_i - \epsilon N^s_i = (1-\epsilon)N^s_i$.

Since the adversary holds at most a $\kappa$ fraction of the resource, the number of bad IDs in the system at the end of iteration $i$ is at most $\kappa N^s_i \leq \kappa N_i/(1-\epsilon)$. 
\end{proof}


\begin{lemma}\label{lem:badbounded}
The fraction of bad IDs is always less than $3\kappa$, provided $\kappa \leq 1/18$. 
\end{lemma}

\begin{proof}
    Fix some iteration $i>0$.  Let $n_i^a, b_i^a$ denote the total, and bad IDs that arrive over iteration $i$. Let  $n_i^d, g_i^d$  denote the total, and good IDs that depart over iteration $i$. 

    Recall that at the end of an iteration $n_i^a + n_i^d \geq N_{i-1} /11$.  Thus, at any point during  the iteration, we have: 
    \begin{eqnarray}
        b_i^a + g_i^d &\leq& n_i^a + n_i^d \leq N_{i-1}/11 \label{eqn:bad-plus-good}
    \end{eqnarray}

    We are interested in the maximum value of the ratio of bad IDs to the total IDs at any point during the iteration.  Thus, we pessimistically assume all additions of bad IDs and removals of good IDs come first in the iteration.  This leads us to examine the maximum value of the ratio over the iteration:
    $$\frac{B_{i-1} + b_i^a}{N_{i-1} + b_i^a - g_i^d} \leq \frac{N_{i-1}\kappa/(1-\epsilon) + b_i^a}{N_{i-1} + b_i^a - g_i^d}$$  
    where the above inequality follows from Lemma~\ref{lem:bound_b}. By Equation~\ref{eqn:bad-plus-good}, we have $g_i^d \leq N_{i-1}/11 - b_i^a$.  Thus, we have:
    \begin{eqnarray*}
         \frac{N_{i-1}\kappa/(1-\epsilon) + b_i^a}{N_{i-1} + b_i^a - g_i^d}
        & \leq &  \frac{N_{i-1}\kappa/(1-\epsilon) + b_i^a}{N_{i-1} + b_i^a - (N_{i-1}/11-b_i^a)}\\
        & = & 3 \kappa \left(\frac{N_{i-1}/(3(1-\epsilon)) + b_i^a/(3\kappa)}{10N_{i-1}/11 + 2b_i^a}\right)\\
        & \leq & 3 \kappa \left(\frac{N_{i-1}/(3(1-\epsilon))}{(10/11)N_{i-1}} + \frac{b_i^a/(3\kappa)}{(10/11)N_{i-1}}\right)\\
        & \leq &  3 \kappa \left(\frac{11}{30(1-\epsilon)} + \frac{(1/(33\kappa))N_{i-1}}{(10/11)N_{i-1}}\right)\\
        & < & 3\kappa \left(\frac{11}{30(11/12)} + 3/5\right)\\
        & = &  3\kappa
    \end{eqnarray*}
    The fourth step follows since $b_i^a \leq N_{i-1}/11$ by Equation~\ref{eqn:bad-plus-good} and the fifth step holds for $\epsilon < 1/12$. 
\end{proof} 
  
Next, we prove the bound on the good spend rate. To begin, we partition every interval of length $\ell$ into $\lceil{\ell \JoinEst}\rceil$ \defn{sub-intervals}  of length at most ${1}/{\JoinEst}$, where $\JoinEst$ is the estimate of the good join rate used in the interval. Then we have the following lemmas.

\begin{lemma}\label{lem:num_bad_sub}
	Fix a sub-interval $j$. Let $\mathcal{T}_{j}$ be the total spending of the adversary in sub-interval $j$. Then, the number of bad IDs that join in this sub-interval is at most $\sqrt{2\mathcal{T}_{j}}$. 
\end{lemma}

\begin{proof}
    Let $b_j$ be the number of bad IDs joining in the sub-interval $j$.  Then, pessimistically assuming all bad IDs join before any good IDs in a sub-interval, we get: \vspace{-3pt}
	$$\mathcal{T}_{j} \geq \sum_{i=1}^{b_{j}} i \geq \frac{b_{j}^2}{2}$$
	Solving the above for $b_j$, we obtain the result. 
\end{proof}

\begin{lemma}\label{lem:iterinter}
    An iteration intersects at most two intervals.
\end{lemma}

\begin{proof}
    We prove this by contradiction.  Assume an iteration starts at time $t_0$ and intersects three or more intervals. Then, there will be at least one interval that is completely contained within the iteration. Let the first such interval start at time $t_1 \geq t_0$ and end at time $t_2 > t_1$.  Let $n^a (n^d)$ be the number of IDs that join (depart) during this interval. Then:
    \begin{align*}
    	n^a + n^d &\geq \vert S(t_1)\triangle S(t_2)\vert \\
    	&\geq \frac{5}{12}\vert S(t_2)\vert \\
    	&\geq \frac{5}{12}\left(\frac{10}{11}\vert S(t_0)\vert \right)\\ 
    	&= \frac{25}{66}\vert S(t_0)\vert 
    \end{align*}
    
    The second step follows from the definition of an interval.  The third step holds since during an iteration, at most $\vert S(t_0)\vert /11$ IDs can depart, and so the system size at time $t_2$ is at least $\frac{10}{11}\vert S(t_0)\vert $.  But the number of joins and departures during the iteration is at most $\vert S(t_0)\vert /11$ by the definition of an iteration.  This gives the contradiction, since $25/66 > 1/11$.
\end{proof}

\begin{lemma}\label{lem:num_sub_intervals}
	Fix an iteration. Let $\interEll$ be the length, and $\jIter$ be the good join rate in this iteration. Then, the number of sub-intervals in the iteration is at most:
	$$700\alpha^3\beta^6 (\jIter\interEll+10)$$ 
\end{lemma}

\begin{proof}
	From Lemma \ref{lem:iterinter}, an iteration intersects at most two intervals. For $i \in \{1,2\}$, let $t_i$ denote time at which the $i^{th}$ interval intersects the iteration for the first time; $\jInterval_i$ be the good join rate in the $i^{th}$ overlapping interval and $\JoinEst_{i}$ be the estimated good join rate set at the end of interval $i$. If there is only one interval intersected, let $\jInterval_1 = \jInterval_2$, $\JoinEst_{1} = \JoinEst_{2}$ and $t_1 = t_2$.
	
	By Lemma \ref{lem:interval-epochs}, an interval intersects at most two epochs. So, let ${\rho}_{1}$ and ${\rho}_{2}$ be the join rate of good IDs over the two epochs that intersect with interval $1$, and let ${\ell}_{1}$ and ${\ell}_{2}$, respectively be the lengths of their intersection, with $\ell_2 =0$ if there is only one such epoch.  Similarly, let ${\rho}_{3}$ and ${\rho}_{4}$ be the join rate of good IDs over the two epochs that intersect with interval $2$, and let ${\ell}_{3}$ and ${\ell}_{4}$, respectively be the lengths of their intersection, with $\ell_4 = 0$ if there is only one such epoch.  Then, from the $\smb$-smoothness property, we have: 
	\begin{align}\label{eq:lowerj}
	    \jIter \interEll 
	    &\geq \sum_{k=1}^4\left\lfloor\frac{\rho_{k}\ell_{k}}{\smb}\right\rfloor\nonumber\\
	    &\geq  \sum_{k=1}^4\left(\frac{\rho_{k}\ell_{k}}{\smb}- 1\right)\nonumber\\ 
	    &= \frac{1}{\smb} \sum_{k=1}^4\rho_{k}\ell_{k} - 4
	\end{align}

	\noindent Next, let $t_0$ denote the time at the start of the iteration. Then, the number of sub-intervals in the iteration is:\vspace{-3pt}
	\begin{align*}
		\sum_{i=1}^2 \lceil{(t_i-t_{i-1}) \JoinEst_{i}\rceil}
		&\leq 700\sma^3\smb^4\sum_{i=1}^2 \left((t_i - t_{i-1})\jInterval_{i} + 1\right)\\
		&\leq 700\sma^3\smb^4\left(2 + \sum_{k=1}^4\lceil\smb\rho_{k}\ell_{k}\rceil\right)\\
		&\leq 700\sma^3\smb^5\left(2+ \left(\sum_{k=1}^4\rho_{k}\ell_{k}\right) + 4\right)\\	
		&\leq 700\alpha^3\beta^6 (\jIter\interEll+10)
	\end{align*}
	In the above, the first step follows from Lemma \ref{lem:prev_current}; the second step follows from $\smb$-smoothness and Lemma~\ref{lem:interval-epochs}; the third step holds since $\smb \geq 1$, we can pull out the $\beta$; and the last step from inequality \ref{eq:lowerj} by isolating the value of $\sum_{k=1}^4\rho_{k}\ell_{k}$. 
\end{proof}

\begin{lemma}\label{lem:num-good-subint}
The number of good IDs that join over any sub-interval is at most $88\sma^4\smb^4 + 1$.	
\end{lemma}

\begin{proof}
 Let $\JoinEst$ be the estimate of the good join rate in the interval containing the sub-interval, and let $\rho$ be the good join rate over the epoch that contains the sub-interval. Then, by $\smb$-smoothness, the number of good IDs that join over the sub-interval is at most: 
	\begin{align*}
		\left\lceil \smb\rho \left(\frac{1}{\JoinEst}\right) \right\rceil 
		& \leq \smb \rho\left( 88\sma^4\smb^3\left(\frac{1}{\rho}\right)\right) + 1\\
		& \leq 88\sma^4\smb^4 + 1
	\end{align*}
	In the above, the first step follows from Theorem \ref{t:JoinEst}. 
\end{proof}

\begin{lemma}\label{fact:cs1}
	Suppose that $u$ and $v$ are $x$-dimensional vectors in Euclidean space. For all $x \geq 1$:
	 $$\sum_{j=1}^x\sqrt{u_jv_j} \leq \sqrt{\sum_{j=1}^x u_j \sum_{j=1}^x v_j}$$
\end{lemma}

\begin{proof}
	Using the Cauchy-Schwarz inequality~\cite{horn2012matrix}:
	$$\left(\sum_{j=1}^n \sqrt{u_j v_j} \right)^2 \leq \sum_{j=1}^n u_j \sum_{j=1}^n v_j $$
	
    \noindent The result follows by taking the square-root of both sides. 
\end{proof}

\begin{lemma}\label{lem:iter_entrance}
	Fix an iteration. Let $\interEll$ be the length of this iteration, $\jIter$ be the join rate of good IDs in the iteration, and $\mathcal{T}$ be the total resource cost to the adversary during the iteration. Then, the total entrance cost to good IDs during the iteration is:
	$$O\left( \sma^{11/2}\smb^7\sqrt{(\jIter\interEll+1)\mathcal{T}} +  \sma^{11}\smb^{14}\jIter\interEll\right). $$
\end{lemma}

\begin{proof}
	Fix a sub-interval $j$ of the iteration. Let $g_j$ ($b_j$) be the number of good (bad) IDs that join in sub-interval $j$, and $\mathcal{T}_j$ be the resource cost to the adversary in sub-interval $j$. Pessimistically assuming all good IDs enter at the end of the sub-interval, the total entrance cost to good IDs in sub-interval $j$ is at most:\vspace{-5pt}
		\begin{eqnarray*}\label{eq:entrance_sub}
            & & \sum_{k=1}^{g_{j}} \left(b_{j} + k\right)\\
            &\leq &  g_{j}\left(\sqrt{2{\mathcal{T}}_{j}} + g_{j}\right)\\ 	
         	&\leq & (88\sma^4\smb^4+1)\left(\sqrt{2{\mathcal{T}}_{j}} + 88\sma^4\smb^4+1\right)	
		\end{eqnarray*}
	The first step follows from Lemma \ref{lem:num_bad_sub}, and the second step follows from Lemma \ref{lem:num-good-subint}. 

 	Let $t$ be the number of sub-intervals in the iteration. Then, the total entrance cost to the good IDs in the iteration is:
 	\begin{eqnarray*}
 		& & \sum_{j=1}^{t} \left( \left(88\sma^4\smb^4+1\right)\left(\sqrt{2{\mathcal{T}}_{j}} + (88\sma^4\smb^4+1)\right)\right)\\
 		& = & \left(88\sma^4\smb^4+1\right) \sum_{j=1}^{t} \sqrt{2\mathcal{T}_{j}} + (88\sma^4\smb^4+1)^2 t\\
 		&\leq & (88\sma^4\smb^4+1)\sqrt{2t\mathcal{T}} + (88\sma^4\smb^4+1)^2 t \\
 		& \leq & (88\sma^4\smb^4+1) \sqrt{\left(1400\sma^3\smb^6(\jIter\interEll+10)\right) \mathcal{T}}\\ 
 		& & + (88\alpha^4\beta^4+1)^2(700\sma^3\smb^6(\jIter\interEll+10))\\
 		& = & O\left( \sma^{11/2}\smb^7\sqrt{(\jIter\interEll+1)\mathcal{T}} +  \sma^{11}\smb^{14}\jIter\interEll\right)
 	\end{eqnarray*}
 	The first step follows from Lemma \ref{fact:cs1} and by noting that $\sum_{j=1}^t {\mathcal{T}}_j =\mathcal{T}$. The third step follows from using Lemma \ref{lem:num_sub_intervals} to upper bound $t$. 
\end{proof}

The previous lemma bounds the entrance cost in a single iteration.  The next lemma bounds the total algorithmic spending.

\begin{lemma}\label{lem:iter_purge}
	Fix an iteration. For this iteration, let $\interEll$ be the length, $\mathcal{D}$ be the rate of departure,  $\jIter$ be the join rate of good IDs, and $\mathcal{T}$ be the total-resource burning cost to the adversary. Then, the total spending for good IDs in this iteration is:
	$$O\left(\mathcal{D}\interEll + \sma^{11/2}\smb^7 \sqrt{(\jIter\interEll+1)\mathcal{T}} + \sma^{11}\smb^{14}\jIter\interEll\right).$$
\end{lemma}

\begin{proof}
    Let $S$ be the set of IDs at the beginning of iteration.  For the iteration, let $t$ be the number of sub-intervals; $g$ and $b$ be the number of good and bad IDs that join, and $d$ be the total number of IDs that depart. For any sub-interval $j$ of the iteration, let $\mathcal{T}_{j}$ be the total resource-burning cost to the adversary in that sub-interval.
	
    Each good ID solves a $1$-hard RB challenge during purges. Hence the cost due to purges is at most the number of good IDs at the end of the iteration, which is at most:
	\begin{eqnarray}
	    \frac{12}{11}\vert S\vert  & \leq & \frac{12}{11}\left(11 \left( d + b + g \right)\right) \nonumber \\
	    & \leq & 12 \left( D\interEll + \sum_{j=1}^{t} \sqrt{2\mathcal{T}_{j}} + \jIter\interEll\right)\nonumber \\
		& \leq & 12 \left( D\interEll + \sqrt{2t\sum_{j=1}^{t}{\mathcal{T}}_{j}} + \jIter\interEll\right)\nonumber \\
		& \leq & 12\left( D\interEll + \sqrt{1400\sma^3\smb^6(\jIter\interEll+10) \mathcal{T}} + \jIter\interEll\right)\label{eq:purge}
	\end{eqnarray}	
	In the above, the first step follows since over an iteration the number of good IDs in the system can increase by at most $\vert S\vert /11$.  The second step follows since the number of ID joins and deletions in an iteration, i.e. $d + b + g$ is at least $\vert S\vert /11$ (Step 2 of \ergo). The third step follows by upper bounding $b$ using Lemma \ref{lem:num_bad_sub} to bound the number of bad IDs joining over all sub-intervals; and noting that $g = \jIter\interEll$ and $b = \mathcal{D}\interEll$. The fourth step follows from Lemma \ref{fact:cs1}. The last step follows from Lemma \ref{lem:num_sub_intervals} and substituting $\sum_{j=1}^t \mathcal{T}_{j} = \mathcal{T}$. 

    Finally, combining Equation~\ref{eq:purge} with the cost from Lemma~\ref{lem:iter_entrance}, for some constant $c$, we have that the entrance costs plus the purge cost paid by all good IDs is:
	\begin{eqnarray*}
		& & c \left(\sma^{11/2}\smb^7\sqrt{(\jIter\interEll+1)\mathcal{T}} + \sma^{11}\smb^{14}\jIter\interEll\right)\\
		& & +12 \left( D\interEll + \sqrt{1400\sma^3\smb^6(\jIter\interEll+10) \mathcal{T}} + \jIter\interEll\right)\\
		& = & O\left(D\interEll + \sma^{11/2}\smb^7 \sqrt{(\jIter\interEll+1)\mathcal{T}} + \sma^{11}\smb^{14}\jIter\interEll\right)
	\end{eqnarray*}
	which proves the result.  
\end{proof}

Consider a long-lived system which undergoes an attack over some limited number of consecutive iterations. A resource bound over the period of attack, rather than over the lifetime of the system, is stronger, and may be of additional value to practitioners.  Thus, we first provide this type of guarantee in Lemma~\ref{l:cost}; Theorem~\ref{thm:new-main-upper} then becomes a simple corollary of this lemma, when considered over all iterations.

For the following lemma, let {\boldmath{$\Iters$}} be a subset of contiguous iterations containing all iterations numbered between $x$ and $y$ inclusive, for any $x$ and $y$, $1 \leq x \leq y$.  Let {\boldmath{$\delta(\Iters)$}} be $\vert S_x - S_y\vert $; and let {\boldmath{$\Delta(\Iters)$}} be  $\delta(\Iters)$ divided by the length of $\Iters$. We note that in the proof of Theorem~\ref{thm:new-main-upper}, $\Delta(\Iters)$ will be $0$.  Let {\boldmath{$T_{\Iters}$}} be the adversarial spend rates over $\Iters$; and let {\boldmath{${J}_\mathcal{I}$}} be the good join rate over $\Iters$.  Then we have the following lemma.

\begin{lemma}\label{l:cost}
   	For any subset of contiguous iterations, $\Iters$, starting after iteration $1$, the good spend rate over $\Iters$ is:
	$$ O\left( \Delta(\mathcal{I})  +  \sma^{11/2}\smb^7\sqrt{({J}_\mathcal{I}+1)T_\mathcal{I}} + \sma^{11}\smb^{14} J_\mathcal{I} \right). $$
\end{lemma}

\begin{proof}
    For all iterations $i \in \mathcal{I}$, let $\interEll_i$ be the length of iteration $i$, $\jIter_i$ be the good join rate in iteration $i$, $\mathcal{D}_i$ be the good departure rate in iteration $i$, and $T_{i}$ be the adversarial resource spend rate in iteration $i$. Then, by Lemma \ref{lem:iter_purge}, for some constant $c$,  we have the total spending of the good IDs over all iterations in $\mathcal{I}$ is at most:
	\begin{align*}
		&\sum_{i \in \mathcal{I}}c\left(D_i\interEll_i + \sma^{\nicefrac{11}{2}}\smb^7 \sqrt{ (\jIter_i\interEll_i+1) T_i\interEll_i}+ \sma^{11}\smb^{14} \jIter_i\interEll_i \right)
	\end{align*}	
	 Dividing this by $\sum_{i \in \mathcal{I}} \interEll_i$ and using Lemma \ref{fact:cs1}, we get:
	\begin{align*}
	& \frac{c\sum_{i \in \mathcal{I}} D_i\interEll_i}{\sum_{i \in \mathcal{I}}\interEll_i} + c\sma^{{11}/{2}}\smb^7\sqrt{{ \frac{\sum_{i \in \mathcal{I}}(\jIter_i\interEll_i + 1)}{\sum_{i \in \mathcal{I}}\interEll_i}}}\\
	& \left(\sqrt{\frac{\sum_{i \in \mathcal{I}}T_i\interEll_i}{\sum_{i \in \mathcal{I}}\interEll_i}}\right) + c\sma^{11}\smb^{14}\frac{\sum_{i \in \mathcal{I}} \jIter_i\interEll_i}{\sum_{i \in \mathcal{I}}\interEll_i}\\
	 & = O\left( \Delta(\mathcal{I})  +  \sma^{11/2}\smb^7\sqrt{({J}_\mathcal{I}+1)T_\mathcal{I}} + \sma^{11}\smb^{14} J_\mathcal{I} \right)	
	\end{align*} 
	which yields the result. 
\end{proof}

We can now prove Theorem~\ref{thm:new-main-upper}:

\begin{proof}
 The resource cost bound follows immediately from Lemma~\ref{l:cost} by noting that $\Delta(\Iters)=0$ when $\Iters$ is all iterations, since the system is initially empty. Then, Lemma \ref{lem:badbounded} completes the proof, by showing that the fraction of bad is always less than $3\kappa$, which is no more than $1/6$, for $\kappa \leq 1/18$.
\end{proof}

\subsection{Parameters and Constants}\label{sec:parameter-constant-discussion}

The relationships between certain model parameters and the constants in \ergo and \goodJest is specified in detail within our analysis above. 

Here, we provide an {\it informal} discussion of some of these relationships and highlight the connections between (1) the parameters, bounds and  constants, and (2) our analysis.  We emphasize that this discussion is only a supplement to and not a replacement for any of the formal analysis above. 

\subsubsection{Bounds}\label{s:bounds-discussion}

\noindent{\bf Bound on {\boldmath{$\kappa$}}.}  
Our current paper is ``proof of concept" that it is possible to tolerate \emph{some} constant value of $\kappa$ efficiently, but we do not believe that $1/18$ is the largest possible constant.  This non-optimal bound is a weakness of our result.  However, it is  common for the first theoretical result in an area to have non-optimal constant bounds.  
\medskip

\noindent{\bf Bound on {\boldmath{$n_0$}}.} 
In Section~\ref{s:churn}, we require $n_0\geq \max\{6000, (720(\gamma+1))^{4/3}, (41\smb)^2 \}$.  The first two terms in the max are needed  to ensure a union bound to keep our probability of error polynomially small over the lifetime of the system (see Lemma~\ref{lem:a_uppbound}).  The third term, is necessary to lower bound the number of good IDs joining in any interval by a constant --- i.e. to ensure that $a \geq 8$ (see Lemma~\ref{lem:a-lowerbound}). \medskip

\noindent{\bf Bound on {\boldmath{$\epsilon$}}.}
Intuitively, if too many good IDs depart in the single round during which a purge is executed, then no bound on the fraction of bad IDs can be guaranteed; this motivates the role of $\epsilon$ used in the proof of Lemma~\ref{lem:bound_b}. The
bound $\epsilon < 1/12$ is required in Lemma~\ref{lem:badbounded} in order to ensure the fraction of bad never exceeds $3 \kappa$.

\medskip\smallskip
\subsubsection{Constants Used in Our Algorithms}

To minimize notation, we have prioritized use of
constants rather than introducing new variables---or  functions of variables---in our algorithms and analysis.  Here, we briefly provide some explanation for the constants chosen.

\medskip

\noindent{\bf \goodJest.} Why do we use the constant $5/12$ for  delineating intervals under \goodJest? This constant results from the product of (i) the $(1-1/6)=5/6$ lower bound on the fraction of good IDs guaranteed by \ergo, and (ii) the $1/2$ fraction of good IDs that change in each epoch.  The expected fraction of good IDs that change across an epoch is at least the product of these two quantities.  The reason that this implies that an interval overlaps two epochs is described in the proof of Lemma~\ref{lem:interval-epochs}. Hence, the constant $5/12$ is used.

\medskip\smallskip

\noindent{\bf \ergo.} 
In \ergo the fraction $1/11$ is used to delineate iterations.  This particular constant is a solution to a system of linear constraints.  Details of these constraints, which involve the upper bound $1/12$ on $\epsilon$, the constant $5/12$ used in \goodJest, and the upper bound $1/18$ on $\kappa$, are given in the proof of Lemma~\ref{lem:badbounded}.


\section{Experiments}\label{sec:experiments} 

We now report on several empirical contributions that complement our prior analysis by illustrating how the performance of our algorithms compares to that predicted by our upper bounds. First, in Section~\ref{section:empasym}, we measure the resource burning cost for \ergo as a function of the adversarial cost, and we compare this against the cost resulting from prior results in the literature. Second, in Section~\ref{sec:evalEst}, we evaluate the performance of the \goodJest algorithm by measuring the approximation factor for the join rate of good IDs. Third, in Section~\ref{section:heuristics}, we propose and implement several heuristics for \ergo.  All our experiments were written in MATLAB, and our source code can be found online~\cite{diksha-code}. 

\smallskip

\noindent
We use churn data from the following networks: \medskip

\noindent{\bf $\bullet$ Bitcoin.} This dataset records the join and departure events of IDs in the Bitcoin network, timestamped to the second, over roughly 7 days \cite{7140490}. \smallskip 

\noindent{\bf $\bullet$ BitTorrent.} This dataset simulates the join and departure events for the BitTorrent network to obtain a RedHat ISO image. We use the Weibull distribution with shape and scale parameters of 0.59 and 41.0, respectively, from \cite{Stutzbach:2006:UCP:1177080.1177105}.\smallskip 

\noindent{\bf $\bullet$ Ethereum.} This dataset simulates join and departure events of IDs for the Ethereum network. Based on a study in \cite{kim2017measuring}, we use the Weibull distribution with shape parameter of 0.52 and scale parameter of 9.8. \smallskip 

\noindent{\bf $\bullet$ Gnutella.} This dataset simulates join and departure events for the Gnutella network. Based on a study in~\cite{rowaihy2007limiting}, we use an exponential distribution with mean of $2.3$ hours for session time, and Poisson distribution with mean of $1$ ID per second for the arrival rate. 


\subsection{Evaluating \ergo}\label{section:empasym}

We compare the performance of \ergo against four resource burning based Sybil defenses: (1) \AlgB~\cite{pow-without}, (2) \AlgA~\cite{li:sybilcontrol}, (3) \AlgC (a name that uses the authors' initials)\cite{rowaihy2005limiting} and (4) ERGO-SF.  Throughout our experiments, we measure performance in terms of the resource burning cost of each algorithm.

\medskip\smallskip

\noindent\textbf{$\bullet$ \AlgB.} It is the same as \ergo, except the hardness of RB challenge assigned to joining IDs is always 1. Thus, \AlgB does not need knowledge of the good join rate and, therefore, has no estimation component like \goodJest. 

\medskip

\noindent\textbf{$\bullet$ \AlgA.} Each ID solves a RB challenge to join. Additionally, each ID tests its neighbors with a RB challenge every $0.5$ seconds, removing from its list of neighbors those IDs that fail to provide a solution within a fixed time period.  These tests are not coordinated between IDs.

\medskip 

\noindent\textbf{$\bullet$ \AlgC.}  Each ID solves a RB challenge to join. Additionally, each ID must solve RB challenges every {\boldmath{$W$}} seconds. We use Equation (4) from \cite{rowaihy2005limiting} to compute the value of spend rate per ID as $\frac{L}{W} = \frac{n}{ N_{attacker}} = \frac{T_{max}}{\AdvPower N}$, where $L$ is the cost to an ID per $W$ seconds, $n$ is the number of IDs that the adversary can add to the system and $N_{attacker}$ is the total number of attackers in the system. The total good spend rate is:
\begin{equation}\label{eq:REMP}
	 \mathcal{A}_{REMP} = (1-\AdvPower)N \times \frac{L}{W} = \frac{(1-\AdvPower)T_{max}}{\AdvPower}
\end{equation}
\noindent to guarantee that the fraction of bad IDs is less than half.

\smallskip  

\noindent\textbf{$\bullet$ ERGO-SF.} This simulates the combination of \ergo and the ML-based method \SybilFuse; recall our discussion of \SybilFuse in Section~\ref{sec:related-work}.  \ergo is changed so that (1) each joining ID is classified as good or bad; and (2) All IDs that are classified as bad are refused entry.  For (1), we assume a classification accuracy of $0.98$, which is the average accuracy reported in~\cite{gao2018sybilfuse} for experiments run over both synthetic and real-world data (Section IV - B, last paragraph).  

We reiterate that, by itself, classification cannot solve \defID, since classifier error allows the adversary to eventually accumulate a bad majority in the system. However, combining \ergo with \SybilFuse allows us to explore the performance benefits a classifier may provide.

\medskip

\noindent{\bf Setup.}
For all algorithms, we measure the spend-rate, which is based solely on the cost of solving RB challenges.  We assume a cost of $k$ for solving a $k$-hard RB  challenge.  We set $\AdvPower = 1/18$, and  let $T$ range over $[2^0,2^{20}]$, where for each value of $T$, the system is simulated for $10,000$ seconds. We assume that the adversary only solves RB challenges to add IDs to the system.  For REMP, we consider $T_{max} = 10^7$ to ensure correctness for all values of $T$ considered.

\begin{figure*}[t!]
\centering
	\includegraphics[trim = 1.8cm 6cm 1.8cm 6.5cm, height=5.5cm]{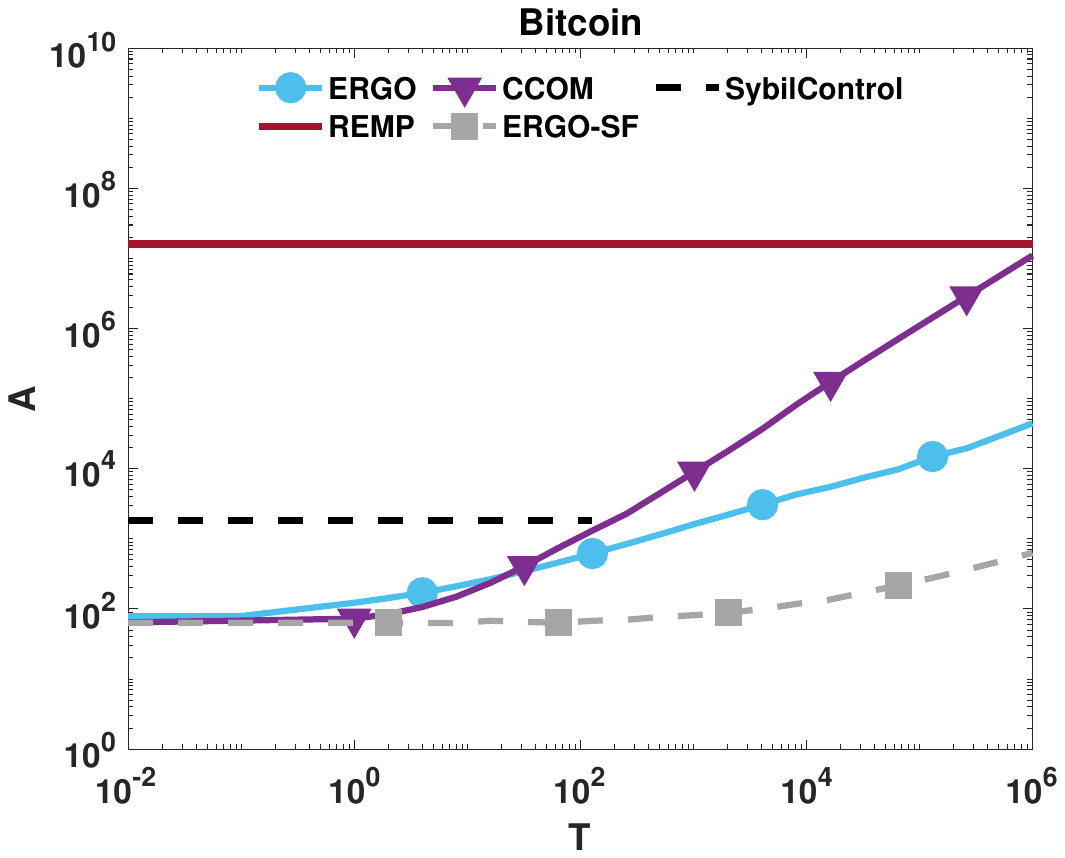}
	\includegraphics[trim = 1.8cm 6cm 1.8cm 6.5cm, height=5.5cm]{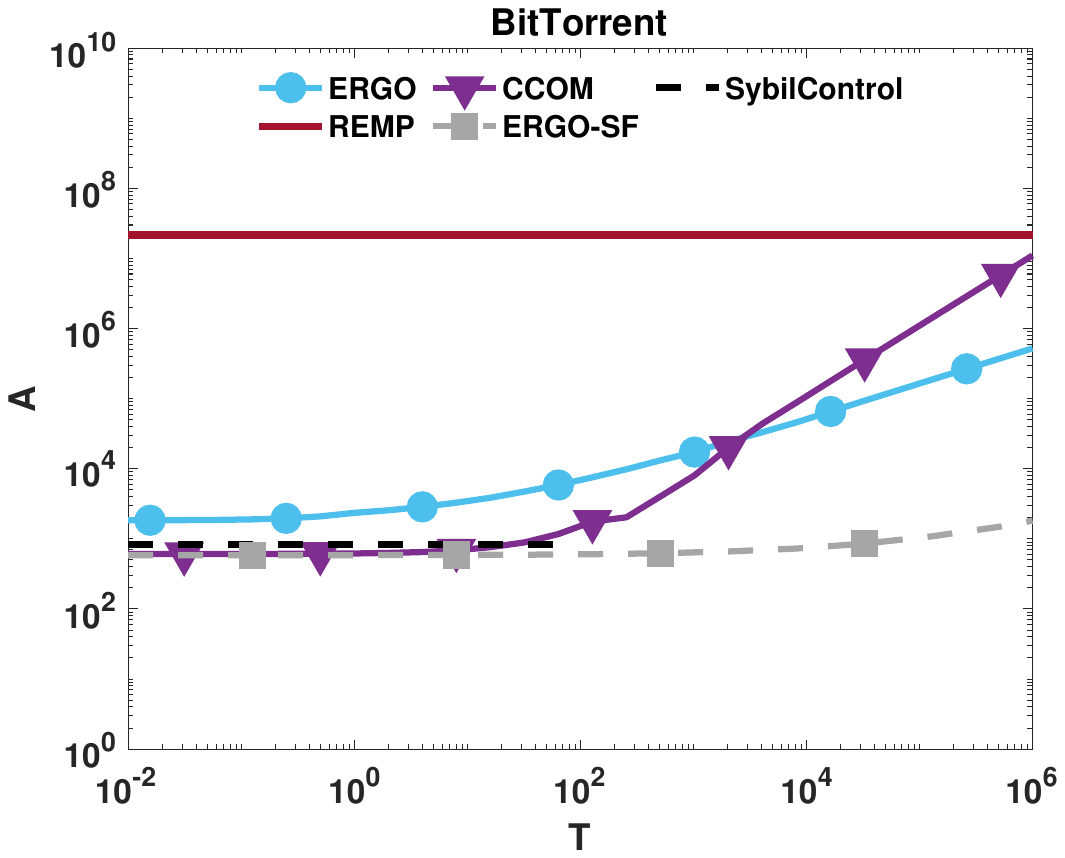}
	\includegraphics[trim = 1.8cm 6cm 1.8cm 6.5cm, height=5.5cm]{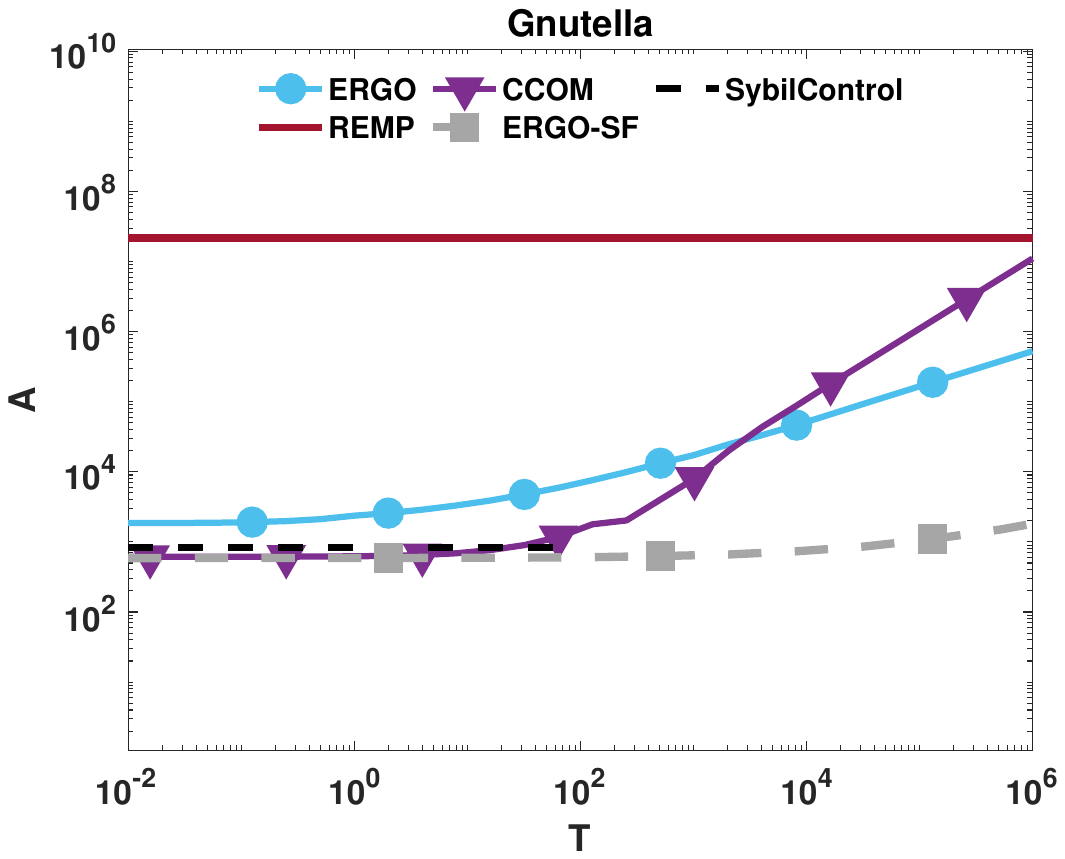}
 	\includegraphics[trim = 1.8cm 6cm 1.8cm 6.5cm, height=5.5cm]{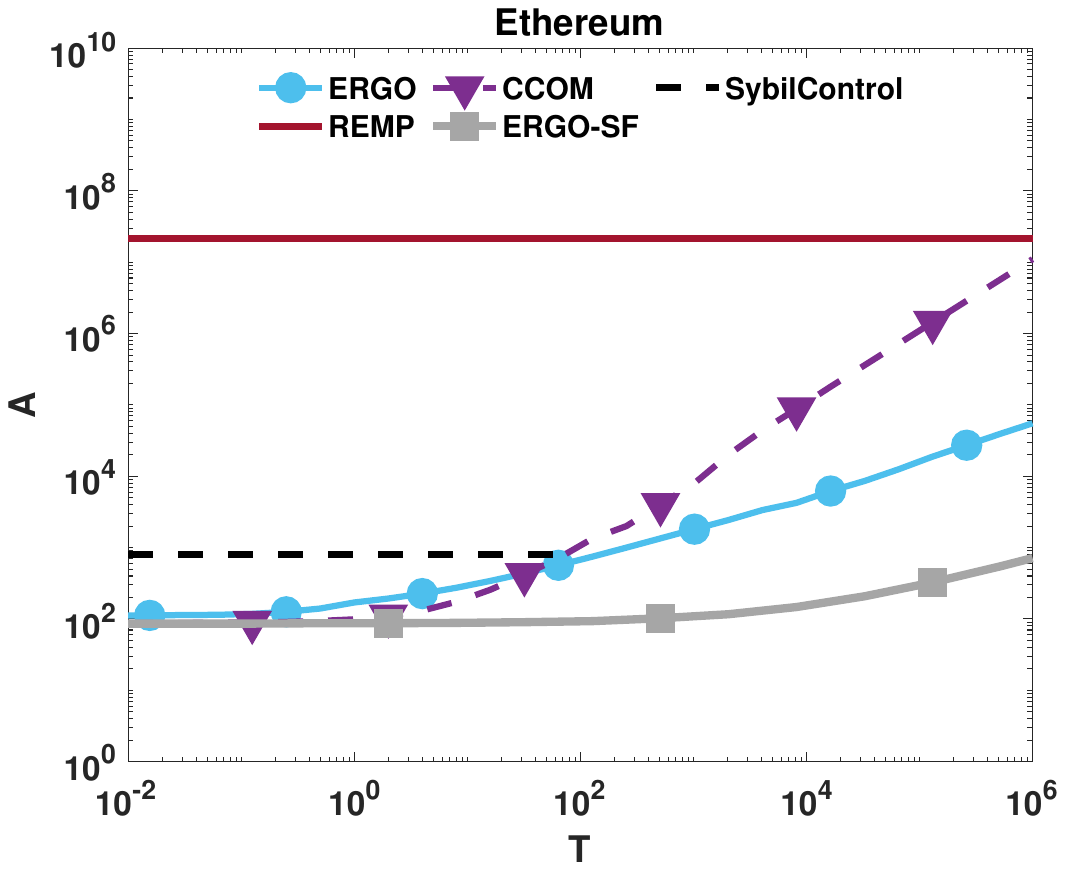}
\caption{Illustration of the good spend rate ($A$) versus adversarial spend rate ($T$).}
\label{fig:AvsT}
\end{figure*} 

\medskip
\noindent{\bf Results.} Figure \ref{fig:AvsT} illustrates our results; we omit error bars since they are negligible. The x-axis is the adversarial spend rate, $T$; and the y-axis is the good spend rate, $\mathcal{A}$.  

We cut off the plot of \AlgA~when the algorithm can no longer ensure that the fraction of bad IDs is less than $1/6$. We also note that REMP-$10^7$ only ensures a minority of bad IDs for up to $T = 10^7$. 

\ergo always has a spend rate as low as the other algorithms for $T \geq 100$, and significantly less than the other algorithms for large $T$, with improvements that grow to about $2$ orders of magnitude. Our heuristic improves further, allowing \ergo to outperform for all $T \geq 0$. This is illustrated by ERGO-SF, which reduces costs significantly, yielding improvements of up to three orders of magnitude during the most significant attack tested. The spend rate for \ergo is linear in $\sqrt{T}$, agreeing  with our theoretical analysis.  We emphasize that the benefits of \ergo are consistent over four disparate networks.
 
\subsection{Evaluating the performance of \goodJest}\label{sec:evalEst}


\begin{figure*}[t!]
\centering
	\includegraphics[trim = 1.8cm 6cm 1.8cm 6.5cm, height=5.5cm]{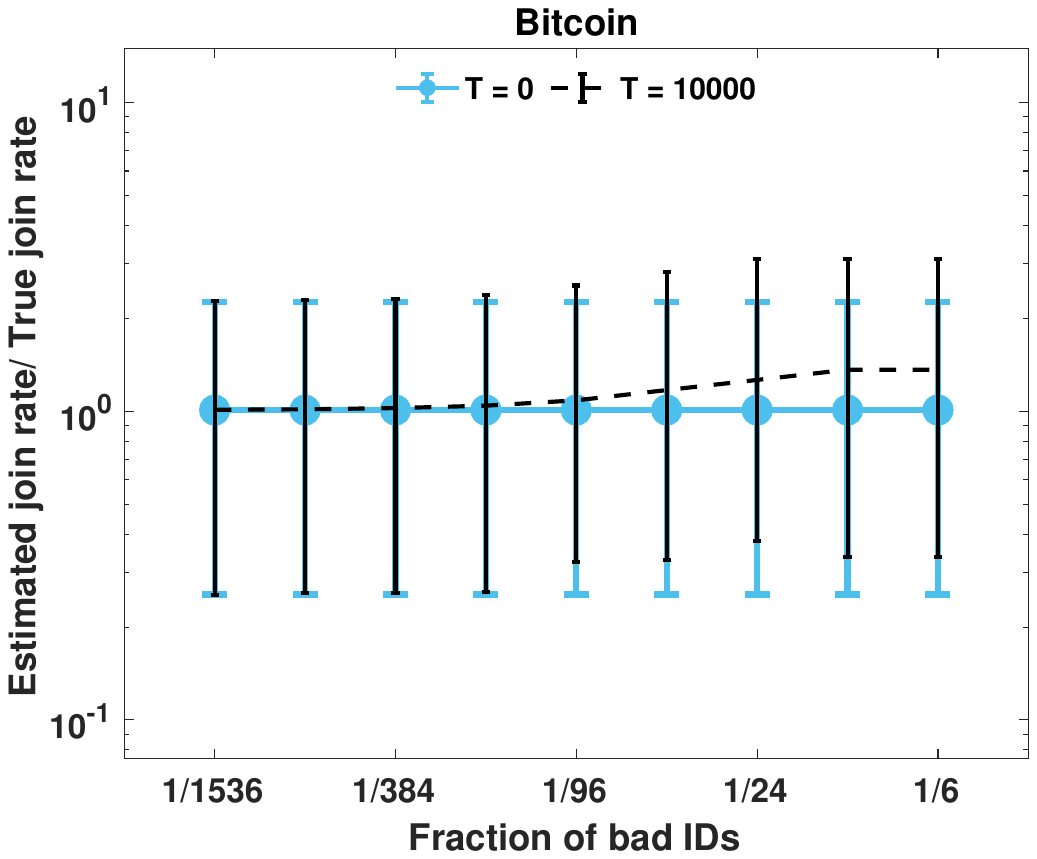}
	\includegraphics[trim = 1.8cm 6cm 1.8cm 6.5cm, height=5.5cm]{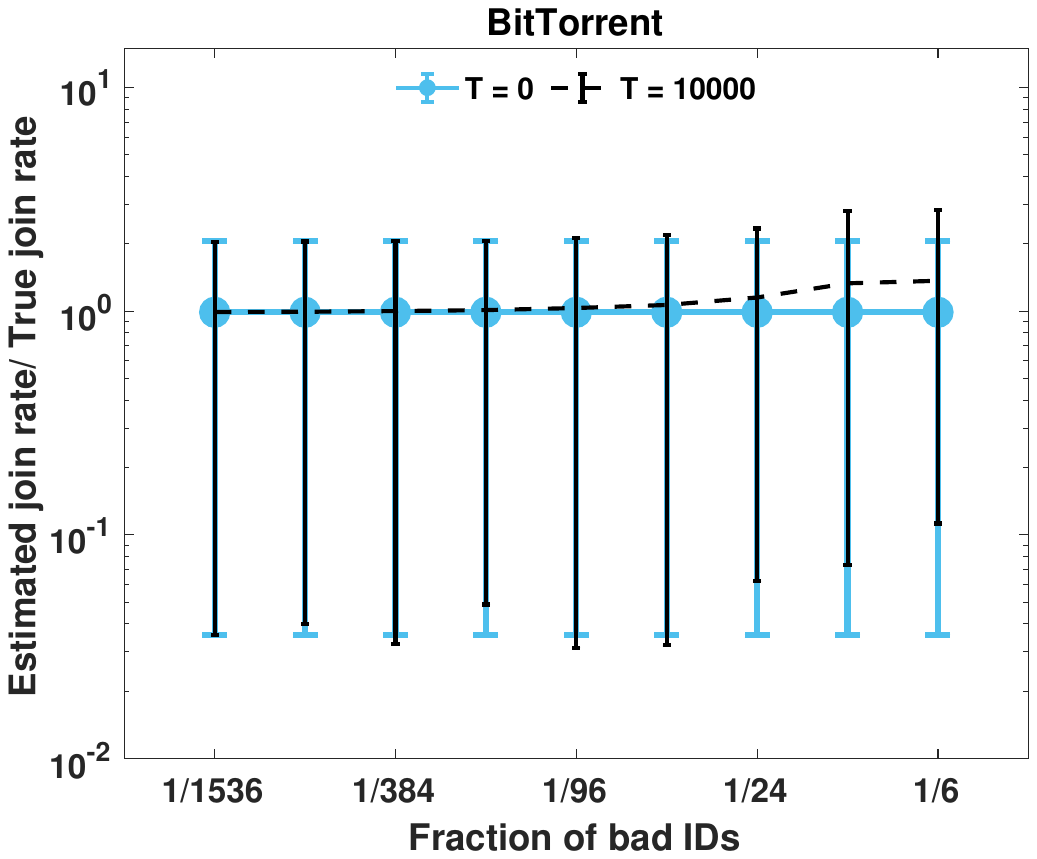}
\includegraphics[trim = 1.8cm 6cm 1.8cm 6.5cm, height=5.5cm]{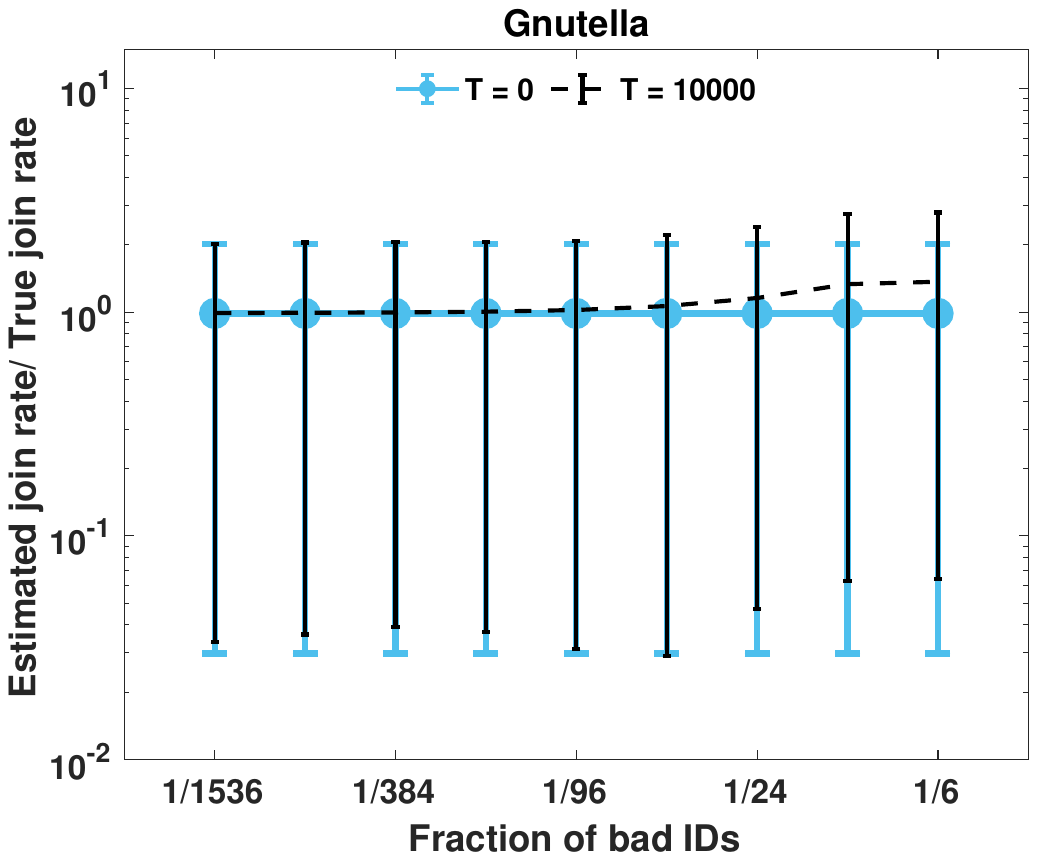}
	\includegraphics[trim = 1.8cm 6cm 1.8cm 6.5cm, height=5.5cm]{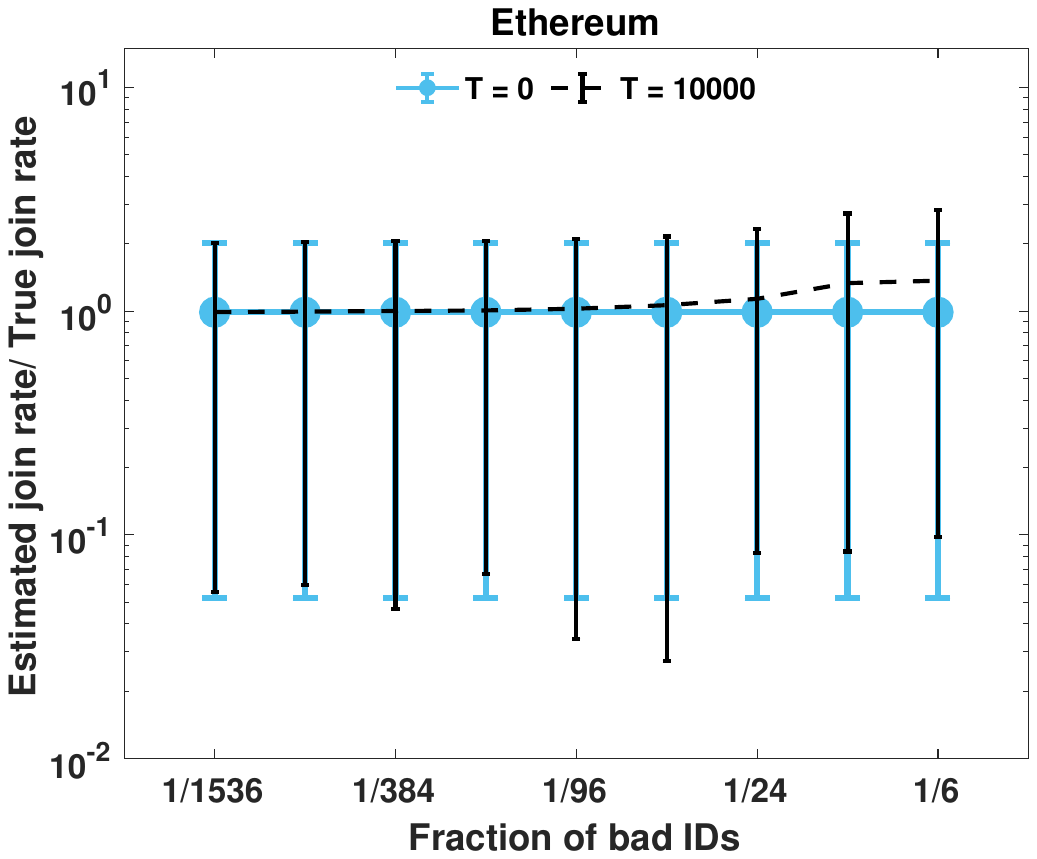}
\caption{Illustration for the ratio of \goodJest estimated to the true join rate for good IDs versus fraction of bad IDs.}
\label{fig:PI}
\end{figure*} 


Having witnessed the encouraging performance of \ergo (and the \ergo-based heuristic, \textsc{ERGO-SF}), we drill down further to examine the performance of \goodJest. We are interested in the accuracy of estimate for the good join rate over our data sets. For the Bitcoin network, the system initially consists of 9212 IDs, and the join and departure events are based off the dataset from Neudecker et al.~\cite{neudecker-atc16}. For the remaining networks, we initialize them with 10,000 IDs each, and simulate the join and departure events over 100,000 timesteps.  

In our simulations, all joins and leaves from the data sets are assumed to be good IDs. We experiment with different fractions of bad IDs that persist in the system; these fractions are $\{1/1500, 1/375,1/94,$ $1/24, 1/6 \}$. We note that $1/6$ actually exceeds the fraction under which our theoretical guarantees hold, but we include this value to observe the impact on performance. To test resilience, we also simulate an attack, where the adversary injects additional bad IDs at a constant rate that can be afforded when $T=10,000$.  For every interval, we measure the ratio of the estimate from \EstGoodJoin to the actual good join rate.

We report our results in Figure \ref{fig:PI}. These demonstrate the robustness of \EstGoodJoin. When $T = 0$, our estimate is always within range $(0.08,1.2)$ of the actual good join rate.  Moreover, even when $T=10,000$, our estimate is always within range $(0.08,4)$ of the actual good join rate.  

These results for \goodJest are encouraging, since they align qualitatively with our theoretical analysis. Indeed, these results suggest that the constants in the big-O notation of our analysis are not so large that they hamper performance in practice. 



\subsection{Heuristics}\label{section:heuristics}

\noindent{\bf Setup.} The previous section offers encouraging evidence that \goodJest is providing \ergo with an accurate estimate of the good join rate, thus facilitating an appropriate entrance cost.  Here, we consider  natural modifications to the other mechanism by which \ergo imposes resource burning---{\it purges}---with the aim of further reducing the resource burning performed by \ergo relative to the adversary.

We present four heuristics that aim to reduce the rate at which purges are conducted. We emphasize that none these heuristics, with the exception of Heuristic 3,  sacrifice the guarantees of Theorem~\ref{thm:new-main-upper}. Our modifications do not yield new theoretical results; however, our goal is to examine practical performance improvements. 

Below, we describe our heuristics and provide intuition for their design.
\medskip

\noindent\textbf{Heuristic 1:}  We align the computation of $\tilde{J}$ with the end of the iteration.  Specifically, at time $t'$ marking the end of an interval, we do not necessarily calculate $\tilde{J}$ immediately using $S(t')$, but rather we wait for the end of the current iteration, if it does not also occur at time $t'$. After the purge occurs, we then calculate $\tilde{J}$. Thus, the fraction of bad IDs in $S(t')$ is reduced to at most $\kappa$, and this may improve the estimate from \goodJest (recall Figure~\ref{alg:estGoodRate}). Since only bad IDs are removed from $S(t')$, this heuristic does not impact our claims in Theorem~\ref{t:JoinEst} regarding \goodJest.\medskip 

\noindent\textbf{Heuristic 2:} We use the symmetric difference to determine when to do a purge.  Specifically, if an iteration starts at time $\tau$, and $\tau'$ is the current time, and if $\vert S(\tau) \triangle S(\tau') \vert  \geq \vert S(\tau)\vert /11$, then a purge is executed.  This still ensures that the fraction of bad IDs can increase by no more than in our original specification, but it also decreases the purge frequency. For example, consider the case where the adversary causes a single bad ID to join and depart repeatedly, and there is no other churn. Notably,  this behavior does not threaten the bound of ensuring less than a $1/6$-fraction of bad IDs and so there is no need to purge. However, since \ergo uses the number of joins and departures to demarcate iterations, a purge will occur nonetheless. In contrast, using the symmetric difference heuristic will avoid this.
\medskip

\noindent\textbf{Heuristic 3:}  When the threshold condition for a purge holds in \ergo, we do an additional check, which is as follows.  Determine if the total join rate over the current iteration is less than some constant $c$ times the good ID join-rate estimate from the prior iteration.  If it is greater, then a purge is performed, since we have seen many more joins than expected based on the prior estimate. Else, we go to the next iteration {\it without} purging. In our experiments, we set $c = 1/11$.  We note that this heuristic can fail to enforce correctness in the case where $c<\alpha$.  However, in our experiments, we confirmed that for all data sets, the heuristic still always kept the total fraction of bad IDs less than $1/6$.

\medskip

\noindent\textbf{Heuristic 4:} This is \textbf{ERGO-SF}, which we described in Section~\ref{section:empasym}.

\medskip


\begin{figure*}[t!]
\centering
	\includegraphics[trim = 1.8cm 6cm 1.8cm 6.5cm, height=5.5cm]{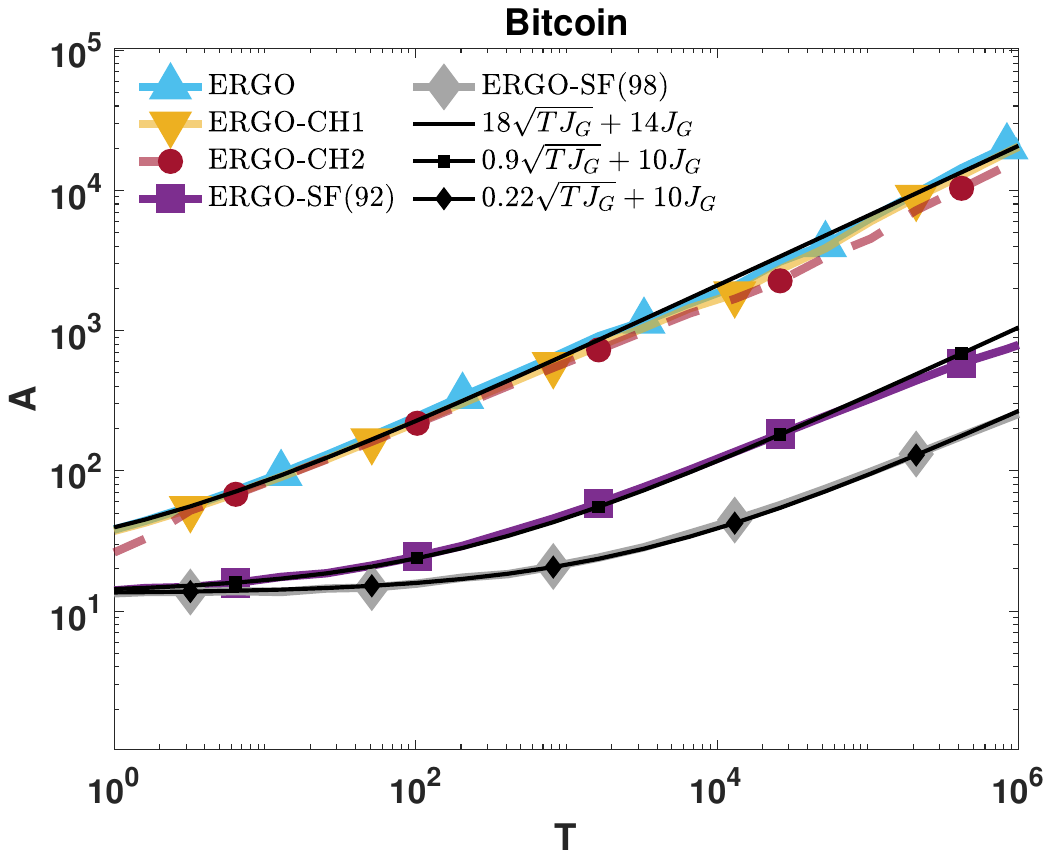} 
	\includegraphics[trim = 1.8cm 6cm 1.8cm 6.5cm, height=5.5cm] {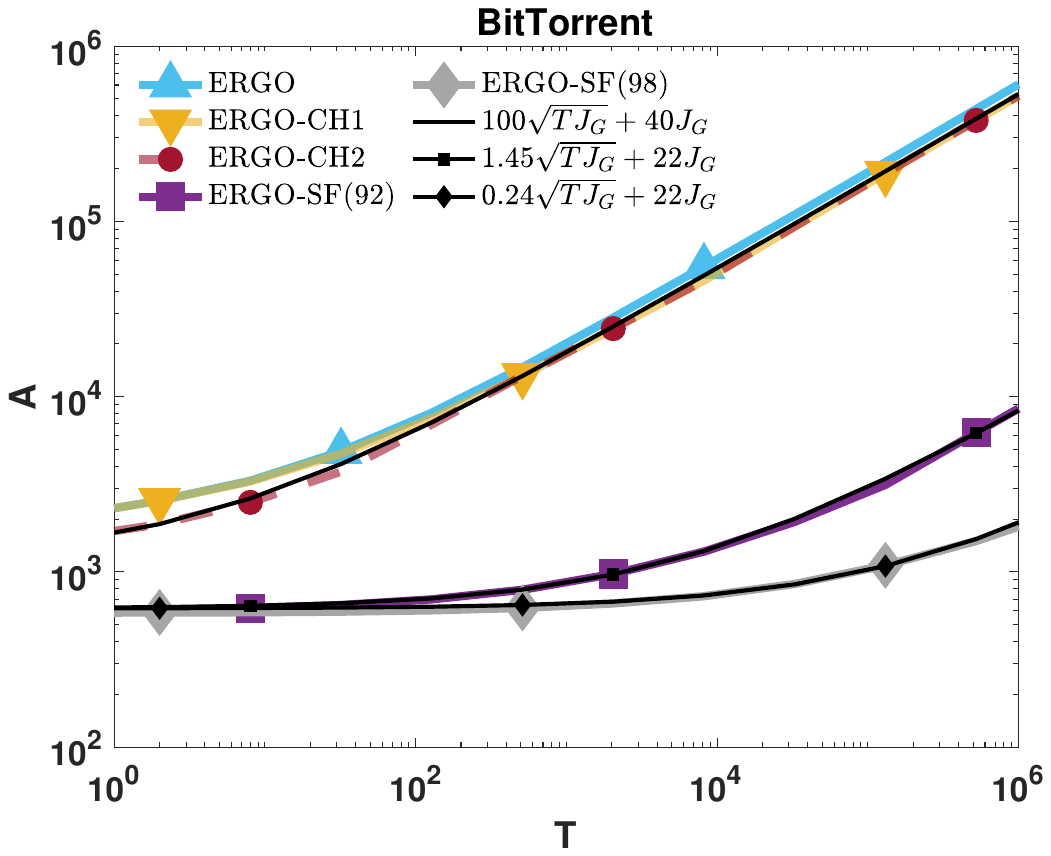}
	\includegraphics[trim = 1.8cm 6cm 1.8cm 6.5cm, height=5.5cm]{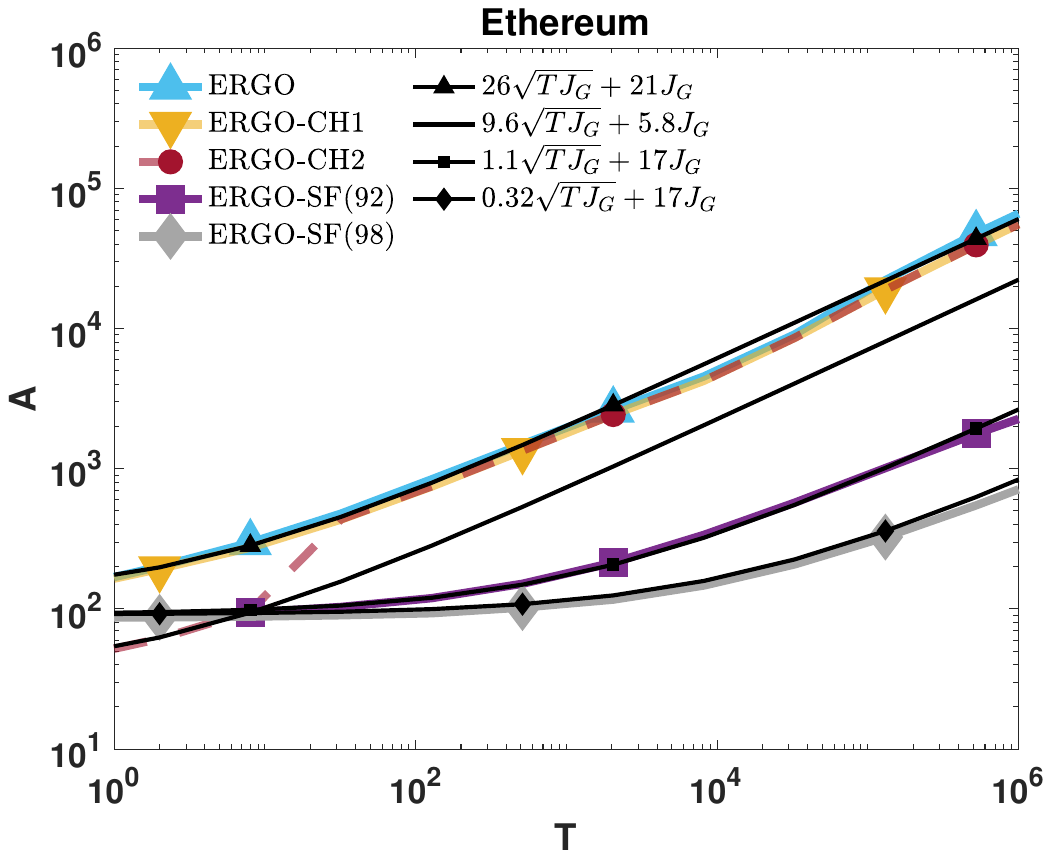} 
	\includegraphics[trim = 1.8cm 6cm 1.8cm 6.5cm, height=5.5cm]{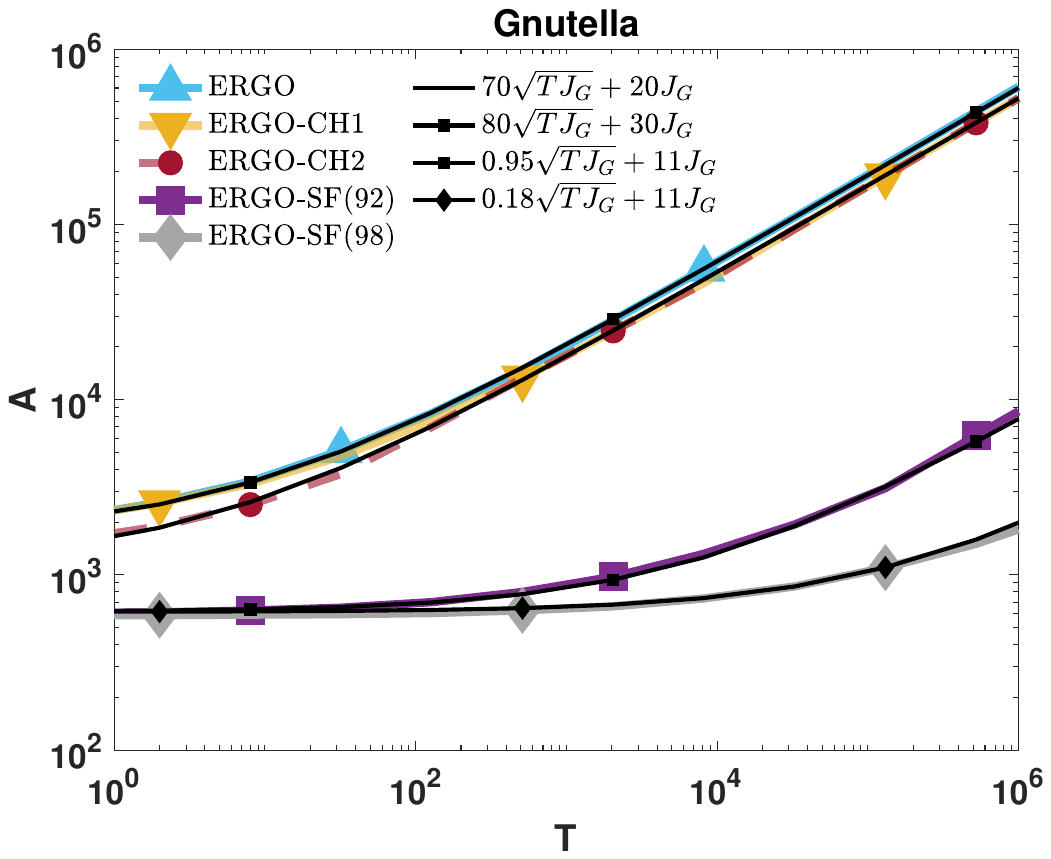} 
\caption{Algorithmic cost versus adversarial cost for \newAlg and heuristics.}
\label{fig:Heuristic}
\end{figure*}


We evaluate the performance of the above heuristics against \newAlg. The experimental setup is the same as Section~\ref{section:empasym}.  We define \defn{ERGO-CH1} using both Heuristic 1 and 2, and \defn{ERGO-CH2} using Heuristics 1, 2, and 3.  In the name, the ``CH" stands for combined heuristic. We define \defn{ERGO-SF(92)} and \defn{ERGO-SF(98)} as \newAlg using Heuristics 1, 2, 3, and 4, with the accuracy of \SybilFuse in Heuristic 4 set as 0.98 and 0.92, respectively. In the name, the ``SF" stands for \SybilFuse.  We simulated Heuristic 1 in conjunction with \newAlg, but, by itself, it did not yield  improvements, and so we omit it from our plots. \medskip

\noindent{\bf Results.}
Figure \ref{fig:Heuristic} illustrates our results. ERGO-SF(92) and ERGO-SF(98) reduce costs significantly during adversarial attack, with improvements of up to three orders of magnitude during the most significant attack tested. Again, these improvements are consistent across $4$ different types of data sets. 

Taken together, these results indicate one of the most important heuristics is the use of a classifier.  For all data sets but Ethereum, ERGO-SF(92) and ERGO-SF(98) always perform better than any other heuristic, and for all data sets, ERGO-SF(92) and ERGO-SF(98) perform best, with an increasing performance gap, for large values of $T$.

We note that the RB cost for Bitcoin and Ethereum when $T$ is small is smaller than for BitTorrent and Gnutella. This behavior is likely caused by the higher churn rates inherent to the latter networks as documented in prior literature~\cite{saroiu2001measurement,6688697,7140490,kiffer2021under}, which results in more frequent purges. 

Also related to the case for small $T$, we observe that Heuristic 3 yields good performance for the Ethereum data set and,  to a lesser degree, the Gnutella data set.  We would expect Heuristic 3 to work particularly well when (1) the good join rate does not change much from iteration to iteration; and (2) there are not many bad IDs joining the system.  
Based on our results for this heuristic, we speculate that the Ethereum network has a churn rate with less ``variability" than the other networks.


\section{A Lower Bound}\label{sec:lower}

In the previous sections, we derived upper bounds for our algorithms and explored their empirical performance. In this section, we provide a lower bound that applies to the class of algorithms which have the attributes $B1 - B3$ described below. 

\begin{itemize}
\item{\bf B1.} Each new ID must pay an entrance fee in order to join the system and this is defined by a \defn{cost function} {\boldmath{$f$}}, which takes as inputs the good join rate and the adversarial join rate.\smallskip
\item{\bf B2.} The algorithm executes over iterations, delineated by when the condition $a + d \geq  \delta\, n $ holds, for any fixed positive $\delta$, where $a$ and  $d$ are the number of arrivals and departures over the  iteration, and $n$ is the number of IDs in the system at the start of the iteration.\smallskip
\item {\bf B3.} At the end of each iteration, each ID must pay $\Omega(1)$ to remain in the system.
\end{itemize}

\smallskip
We emphasize that B1 captures any cost function where the cost during an iteration depends only on the good join rate and the bad join rate.  Our analysis of \ergo makes this assumption since its cost function depends on estimates of the good join rate and the total, i.e. the good-ID join rate plus the bad-ID join rate.  With regard to B2 and B3, recall that we wish to ensure that the fraction of bad IDs is always less than some constant fraction. It is hard to imagine an algorithm that preserves this invariant without an RB challenge being imposed on all IDs whenever the system membership changes significantly.

\subsection{Lower-Bound Analysis}
\noindent Restating in terms of the conditions above, we have the following result.\smallskip

\begin{theorem}Suppose an algorithm satisfies conditions B1-B3, then there exists an adversarial strategy that forces the algorithm to spend at a rate of $ \Omega(\sqrt{\advAveCost\,\joinRate} + \joinRate)$, where $\joinRate$ is the good ID join rate, and $\advAveCost$ is the algorithmic spending rate, both taken over the iteration.
\end{theorem}\label{t:lower-bound}
\begin{proof}
 Fix an iteration.  Let $n$ be the number of IDs in the system at the start of the iteration.  The adversarial will have bad IDs join uniformly at the maximum rate possible, and then have the bad IDs drop out during the purge.  In particular, let $\advRate$ be the rate at which bad IDs join, and let $f(\advRate, \joinRate)$ be the algorithm's entrance cost function based on $\advRate$ and $\joinRate$; we pessimistically assume that the algorithm knows both $\advRate$ and $\joinRate$ exactly.  Then $\advRate = T / f(\advRate, \joinRate)$.

We first calculate the algorithmic spending rate due to purge puzzle costs in the iteration.  Since the iteration ends after $O(n)$ join events (B2), and since each purge puzzle has a cost of $\Omega(1)$ (B3), spending on purge puzzles is asymptotically at least equal to the number of good and bad IDs join the system.  Thus, the spending due to purge puzzles is $\Omega (\joinRate + \advRate)$.
 		 	
\smallskip			
\noindent We now have two cases:\medskip

\noindent	
     \textbf{Case 1:} {\boldmath{$f(\advRate, \joinRate) \leq \advRate/\joinRate$}}.  In this case, we have:    
           \begin{eqnarray*}
	 	\advRate & \geq & \joinRate f(\advRate, \joinRate) =  \joinRate T/ \advRate
	 \end{eqnarray*}
where the above equality holds since $f(\advRate, \joinRate) = \advAveCost/\advRate$.  Solving for $\advRate$ we get:
	 \begin{eqnarray*}
	 	\advRate &\geq & \sqrt{\advAveCost\,\joinRate}
	 \end{eqnarray*} 	 

Since the algorithmic spending rate due to purge costs is $\Omega(\advRate+ \joinRate)$, the total algorithmic spending rate is:
      \begin{eqnarray*}
 \Omega(\advRate + \joinRate) & = & \Omega\left(\sqrt{\advAveCost\,\joinRate} + \joinRate\right)
	 \end{eqnarray*}

\noindent\textbf{Case 2:} {\boldmath{$f(\advRate, \joinRate) > \advRate/\joinRate$}}. In this case, we have:	 
	 \begin{eqnarray*}
	 	\advRate & < & \joinRate\,f(\advRate, \joinRate) = \joinRate T/\advRate
	 \end{eqnarray*}
The above equality follows since $f(\advRate, \joinRate) = $ $\advAveCost/\advRate$.  Solving for $\advRate$, we get:
	 \begin{eqnarray*}
	 	\advRate & < & \sqrt{\advAveCost\,\joinRate}.
	 \end{eqnarray*}

The spending rate for the algorithm due to entrance costs is $\Omega(\joinRate\, f(\advRate, \joinRate))$.  Adding in the spending rate for purge costs of $\Omega (\joinRate + \advRate)$, we get that the total spend rate of the algorithm is:
	 \begin{eqnarray*}
 \Omega(\joinRate\, f(\advRate, \joinRate) + (\joinRate + \advRate))
	 & = & \Omega(\joinRate\, \advAveCost/\advRate + (\joinRate + \advRate)) \\
	 & = & \Omega\left(\sqrt{\advAveCost\,\joinRate} + \joinRate\right).
	 \end{eqnarray*}
\end{proof}


\section{Decentralization} \label{s:committee}

When there is no centralized server, we can run our algorithms using a \defn{committee}: a $O(\log n_0)$ sized subset of IDs with a good majority. This committee takes over the responsibilities of the server, which means the committee runs \goodJest and \ergo in a robust, distributed fashion.  Below, we discuss the necessary modifications needed to maintain and use such a committee.

The results in this section are straight-forward applications of tools that have already been developed in prior literature, such as state-machine replication and committee election.

\medskip\smallskip


\noindent\textbf{Model Modifications.}\label{s:model-modifications}
Our communication model inherits the assumptions made in~\cite{abraham2019sync} and~\cite{awerbuch:towards2}.   These assumptions allow us to implement State Machine Replication and Committee Election, respectively.  

Thus, we assume synchronous communication, and that there exist secure and authenticated communication channels between all pairs of IDs in the committee.  
Additionally, we require a secure and authenticated communication  channel between each member of the committee and each ID in the system.  This enables challenge solutions to be sent to the committee, and election information to be sent to all participants.

Further, in order to use the algorithm by Rabin and Ben-Or~\cite{rabin1989verifiable} for committee selection (see below in Section~\ref{s:use-and-maintenance}), the adversary is oblivious to the private, random bits generated by any good ID.



\subsection{Committee and System Initialization}\label{sec:commMain}

\noindent
\textbf{\genID.} To initialize our system, we require a solution to  \genID (recall Section~\ref{s:churn}). \genID guarantees that at initialization, (1) all good IDs agree on the same set of IDs; and (2) at most a $\kappa$-fraction of IDs in that set are bad.  Additionally, \genID ensures that all good IDs agree on a committee of logarithmic size with a majority of good IDs.  There are several algorithms that solve \genID in our model, as defined in Section \ref{sec:model-main}~\cite{andrychowicz2015pow, hou2017randomized,katz2014pseudonymous,aggarwal2019bootstrapping}.  The algorithm in~\cite{aggarwal2019bootstrapping}, makes use of computational challenges, and runs in expected $O(1)$ rounds, and, in expectation, requires each good ID to send $O(n)$ bits, and solve $O(1)$ $1$-hard RB challenges.


\subsection{Use and Maintenance of Committee}\label{s:use-and-maintenance}

The committee makes use of \defn{State Machine Replication (SMR)}~(see \cite{attiya2004distributed,lynch1996distributed,abraham2019sync, bessani2014state})  to agree on an ordering of network events so as to execute \goodJest and \ergo in parallel. In order to scalably and correctly execute SMR, we maintain the invariant that the committee always has size $\Theta(\log{n_0})$ and a majority of good IDs.

To maintain this invariant, a new committee is elected by the old committee at the end of each iteration.  In particular, at the end of iteration $i$, the old committee selects a committee of size $C\log  N_i $, where $C>1$ is a sufficiently large constant. This committee selection process can be accomplished in our model via classic secure multiparty computation protocols; for example, see Rabin and Ben-Or~\cite{rabin1989verifiable}. Note that subsequent results can accomplish the same task more efficiently, but require cryptographic assumptions.  For example, Awerbuch and Scheideler~\cite{awerbuch:random} describe an algorithm specifically for random number generation that can be used by the existing committee to select new committee members. 


\medskip

With the modifications above, \newAlg gives the following guarantees:  

\begin{theorem}\label{thm:decent}
	For $\AdvPower \leq 1/18$,  \newAlg ensures that the good spend rate is: 
        $$O\left(\sma^{11/2}\smb^7 \sqrt{\advAveCost(\jAll + 1)}  + \sma^{11}\smb^{14}\jAll \right) .$$
	and guarantees:
	\begin{enumerate}
	\item The fraction of bad IDs in the system is less than $1/6$.
	\item The fraction of bad IDs in the committee is at most $1/8$.
	\end{enumerate}
\end{theorem}  

Much of the proof of Theorem~\ref{thm:decent} follows from earlier arguments, but we must still show that the committee does indeed always have size $\Theta(\log n_0)$ and a good majority.
Thus, the following lemma is useful.

\begin{lemma}\label{lem:maj_comm}
	The following holds for all iterations $i > 0$, \whp. \newAlg maintains a committee with at least a $7/8$ fraction of good IDs.  Moreover, this committee has size $\Theta(\log n_0)$.  
\end{lemma}
\begin{proof}
For iteration $i = 0$, the lemma holds true by properties of the \genID algorithm used to initialize the system (See Section \ref{sec:commMain}, and Lemma 6 of \cite{andrychowicz2014distributed}).
	
Fix an iteration $i > 0$. Recall that $B_i$ and $G_i$ are the number of bad and good IDs in the system at the end of iteration $i$, respectively, $N_i = B_i + G_i$. Recall that a new committee is elected by the existing committee at the end of an iteration by selecting $C\log N_i $ IDs independently and uniformly at random from the set $S_i$, for a sufficiently large constant $C>0$, defined concretely later in this proof. Let $X_G$ $(X_B)$ be random variables for the number of good (bad) IDs, respectively elected to the new committee at the end of iteration $i$. Then:
	\begin{align}
		E[X_G] & = \frac{ G_i }{N_i }C\log{N_i } \nonumber\\ 
		&= \left(1-\frac{ B_i }{N_i }\right)C\log{N_i } \nonumber \\
		&\geq  \left(1-\frac{1}{18(1-\epsilon)}\right)C\log N_i  \nonumber \\
		&\geq \left(1-\frac{1}{18(11/12)}\right)C\log N_i  \nonumber \\
		&= \frac{31}{33}C\log N_i  
	\end{align}
	
	 In the above, the third step follows from Lemma \ref{lem:bound_b} and fourth step since $\epsilon < 1/12$. Similarly, the expected number of bad IDs elected into the committee is:
	 \begin{align}
	     	E[X_B] & = \frac{ B_i }{N_i }C\log{N_i } \leq \frac{1}{18(1-\epsilon)}C\log N_i 
	 \end{align}
	 
	 Next, by Chernoff bounds~\cite{mitzenmacher2017probability}, for any constant $0 < \delta < 1$:
	 \begin{align*}
	& Pr \left( X_G \leq (1-\delta)\frac{31}{33}C\log N_i  \right)\\
	& \leq  \exp{\left\{ -\frac{31\delta^2 C\log N_i }{66} \right\}}\\
	 	& =  N_i ^{-\frac{31C\delta^2}{66}}\\
	 	& \leq {n_0^{-(\lifetime + 1)}}
	 \end{align*}
    The last step holds for all $C \geq \frac{66(\lifetime + 1)}{31\delta^2}$. For $\delta = 13/310$, we obtain that the number of good IDs in the committee is at least $(9/10) C\log{ N_i }$ w.h.p.
    
    Again:
    \begin{align*}
        & Pr \left( X_B \geq (1+\delta)\frac{2}{33}C\log N_i  \right)\\
        & \leq  \exp{\left\{ -\frac{2\delta^2 C\log N_i }{99} \right\}}\\
	 	& =  N_i ^{-\frac{2C\delta^2}{99}}\\
	 	& \leq {n_0^{-(\lifetime + 1)}}
    \end{align*}
    The last step holds for all $C \geq \frac{99(\lifetime + 1)}{2\delta^2}$. For $\delta = 13/20$, we obtain that the number of bad IDs in the committee is at most $(1/10) C\log{ N_i }$ w.h.p.

    \medskip
	 	 	
    \noindent Next, let $Y_g$ be the number of good IDs that depart from the committee during iteration $i$.  Note that the number of good departures in iteration $i$ is at most $N_{i-1} /11$. Then, since each departing good ID is selected independently and uniformly at random (See Section \ref{sec:model-main}), we obtain:
    
		\begin{align*}
			E[Y_g] &\leq \frac{ N_{i-1} }{11} \left(\frac{C\log N_i }{ N_i }\right)\\
			&\leq \left(\frac{11 N_i /10}{11}\right) \frac{C\log N_i }{ N_i }\\
            &=\frac{C}{10}\log N_i 		
		\end{align*}  
		
	The second step holds since over an iteration at most $ N_{i-1} /11$ IDs depart, so $ N_i  \geq 10 N_{i-1} /11$.
	
	\medskip
	
	Again, using Chernoff bounds, we have:
		\begin{align*}
			Pr \left( Y_g \geq (1+\delta')\frac{C}{10} \log N_i  \right) &\leq  \exp\left\{ -\frac{\delta'^2C}{30}\log N_i  \right\}\\
			& =  N_i ^{-\delta'^2C/30}\\
			& \leq n_0^{-(\lifetime+1)}
		\end{align*}
	where the first step holds for any constant $0< \delta' < 1$ and the last step holds for all $C \geq \frac{30(\lifetime + 1)}{\delta'^2}$. Thus, for $\delta' = 1/9$, the number of good IDs in the committee is always greater than $(9/10)C\log{ N_i } -  (1/9)C\log{ N_i } = (71/90)C\log{ N_i } $.

	Next, the fraction of good IDs in the committee is minimized when only good IDs depart from the committee. Thus, the fraction of good IDs in the committee is always at least:
	$$\frac{(71/90)C\log{ N_i }}{C\log N_i  - (1/9)C\log{ N_i }} > \frac{7}{8}$$

To bound the committee size, note that the committee always has a number of IDs at least:
	$$C\log{ N_i } - \frac{C}{10}\log{ N_i } - \frac{C}{9}\log{ N_i } = \Theta(\log n_0)$$

	Finally, we use a union bound over all $n_0^\lifetime$ iterations to show that the above facts hold, \whp, for all iterations. 
\end{proof}

The proof of Theorem \ref{thm:decent} now follows from Lemmas ~\ref{lem:badbounded},~\ref{l:cost}, and~\ref{lem:maj_comm}.


\section{Conclusion and Future Work}\label{sec:future} \ergo is a new Sybil-defense that efficiently employs resource-burning to limit Sybil IDs, despite high churn.  Specifically,  \ergo  guarantees (1) there is always a minority of Sybil IDs; and (2) despite a churn rate that can vary exponentially, the resource burning rate of good IDs is $O(\sqrt{TJ} + J)$, where $T$ is the resource burning rate of the adversary, and $J$ is the join rate of good IDs. Our experiments illustrate that \ergo significantly decreases algorithmic resource costs when compared to other Sybil defenses. We also show empirically that costs can be further reduced by combining \ergo with \SybilFuse,  thus illustrating the benefits of leveraging classification algorithms.  Finally,  we prove that the resource burning rate of $O(\sqrt{TJ} + J)$ is asymptotically optimal for a large class of algorithms.

There are two future research directions that we feel are important, which we discuss next.

\subsection{Incentives}
Under \ergo, recall that good IDs must periodically solve a $1$-hard puzzle to avoid being purged from the system. In practice,  as discussed in Section~\ref{sec:incentives},  the system may benefit from providing incentives for honest users to obey this protocol. Thus, an interesting direction for future work is devising a mechanism for incentivizing good IDs to solve these puzzles. For example, as is the case for cryptocurrencies like Bitcoin, we might make use of rewards.  For example, during the purge, competition for a reward could be used to ensure that IDs actually commit sufficient resources to remain in the system.  If challenges are proof-or-work based, the ID that finds the smallest solution during this period could receive units of cryptocurrency, such as that received for mining a block.  In this way, all good IDs would have positive incentives to commit resources. Hence, the difficulty of a $1$-hard puzzle could be tuned, based on measured computational effort, to automatically adjust to new, faster hardware or heterogeneous hardware distribution.   Formally analyzing the game-theoretic properties of such a scheme, using techniques similar to those used in~\cite{abraham2016solidus,margolin2007informant,liu2019survey-pub},  is left for future work.  

\subsection{Distributed Hash Tables}
Can we apply the results in this paper to build and maintain a Sybil-resistant distributed hash table (DHT)~\cite{urdaneta:survey}?  To the best of our knowledge, there is no such result that ensures the good IDs pay a cost that is slowly growing function of both the good churn rate and the cost payed by an attacker. Third, can a similar approach be used to mitigate distributed denial-of-service (DDoS) attacks at the application layer? Here, server resources can be exhausted by bad clients whose spurious jobs cannot be {\it a priori} distinguished from legitimate jobs. It seems plausible that a resource-burning approach similar to \ergo might offer a defense here too. Fourth, while our lower bound applies to a large class of algorithms, it would be interesting to establish a more general result.

\subsection{Tolerating $\kappa > 1/18$ }\label{s:larger-kappa}

In Section~\ref{sec:parameter-constant-discussion}, we discussed the interaction between various parameters and constants, and we highlighted that there is some flexibility in selecting their values such that our formal arguments hold with only cosmetic changes. Here, we elaborate on this aspect in order to highlight the challenges involved with  tolerating a larger value of $\kappa$, along with sketching ideas for how this might be accomplished.

To begin, first consider what happens if we wish to alter \goodJest such that an interval ends whenever:
$$\vert S(t') \triangle S(t) \vert \geq \frac{1}{2} \vert S(t') \vert.$$
That is, we wish to change the constant $5/12$ to be $1/2$. It suffices to change the definition of an {\it epoch} so that their boundaries occur when the symmetric difference between the sets of good IDs at the start and the end of the epoch exceeds $3/5$ (rather than $1/2$) times the number of good IDs at the start.  This can be seen by the proof for Lemma~\ref{lem:interval-epochs}, since now the key inequality in the argument is $\vert S(t_2)  \triangle S(t_0) \vert \geq  (3/5)(5/6) = 1/2$, which ends an epoch under this new definition, and the proof follows as before with this minor change. 

There are two reasons this concrete example is useful. First, it illustrates how the values we use in our arguments have some flexibility. Second, it is important to understanding why the structure of our current analysis cannot accommodate a significantly larger $\kappa$, which we now discuss.

What prevents \ergo from tolerating a larger $\kappa$ under our current analysis? While we have flexibility to choose values for our parameters and constants, this flexibility has limits.  Specifically, a key obstacle occurs in achieving Inequality~\ref{eq:uppersymm} within the proof of Lemma~\ref{lem:a-lowerbound}. Namely, we require a sufficiently large constant---currently this is $5/12$---in order to achieve Inequality~\ref{eq:uppersymm}. Consider what happens if we wish to tolerate $\kappa\leq 1/9$ (rather than $\kappa\leq 1/18$), and thus we will now have a strict bound of $3\kappa = 1/3$ on the fraction of bad IDs in the system at any time (recall Lemma~\ref{lem:badbounded}). Note that the chain of inequalities that leads to Inequality~\ref{eq:uppersymm} no longer holds, since the constant $5/12$ is too small. 

But can we not increase this constant as we did in the concrete example above? Unfortunately, if we wish to tolerate a strict bound of  $1/3$ on the fraction of bad IDs, the cosmetic change performed in that prior example is insufficient. This can be observed by simply following through the proofs for Lemma~\ref{lem:interval-epochs} and the chain of inequalities that leads to Inequality~\ref{eq:uppersymm}.  To be explicit, recall that an interval ends when:
$$\vert S(t_2)  \triangle S(t_0) \vert \geq  (E)(2/3) = I$$
\noindent{}where $E$ is the constant used for delineating epochs, and $I$ is the constant used for delineating intervals; recall that the latter is currently $5/12$ and is the value we wish to make larger. However, note that $I\leq 2/3$ no matter what value we choose for $E$. Therefore, for the argument of Lemma~\ref{lem:interval-epochs} to remain correct, we must end intervals whenever:
$$\vert S(t') \triangle S(t) \vert \geq \frac{2}{3} \vert S(t') \vert.$$
What are the implications for the Inequality~\ref{eq:uppersymm} used in the proof of  Lemma~\ref{lem:a-lowerbound}? The portion of the argument:
\begin{align}
		\vert G(t') \triangle G_t\vert \nonumber
		&\geq ...\\
  &\geq \left(\frac{5}{12}\right)\vert S(t')\vert  - \frac{\vert S(t')\vert }{6} - \frac{10}{7}\left(\frac{\vert S(t')\vert }{6}\right)\nonumber\\
		&\geq \frac{\vert S(t')\vert }{84}\nonumber
\end{align}
is critical. However, with our hypothetical modification, we have:
\begin{align}
		\vert G(t') \triangle G_t\vert \nonumber
		&\geq ...\\
  &\geq \left(\frac{2}{3}\right)\vert S(t')\vert  - \frac{\vert S(t')\vert }{3} - d\left(\frac{\vert S(t')\vert }{3}\right)\nonumber\\
		&\ngtr 0\nonumber
\end{align}
\noindent where $d$ is a constant now exceeding $10/7$, which can be seen from following through the proof of Lemma~\ref{lem:boundsize} using $2/3$ instead of $5/12$. Thus, we no longer achieve $\vert G(t') \triangle G_t\vert  = \Omega(S(t'))$, but rather end up with a negative value.

Tolerating larger values of $\kappa$ is an area of future work. It may be possible to sufficiently increase $I$ by allowing  for an interval to overlap {\it more} than two epochs; this would require significant changes to Lemma~\ref{lem:interval-epochs} and  other arguments in Section~\ref{s:analGoodJest}. Note that this such changes would likely also impact the accuracy of \goodJest, since if an interval overlaps more epochs, the exponents for $\sma$ and $\smb$ in Theorem~\ref{t:JoinEst}  will increase as a result.

\medskip\medskip

\noindent{\bf Acknowledgements.} This work is supported by the National Science Foundation grants CNS 1816076, CNS 1816250, CNS 2210299, CNS 2210300, and CCF 2144410.  We are grateful to the anonymous reviewers; their feedback greatly improved our manuscript.



\begin{thebibliography}{100}
\expandafter\ifx\csname url\endcsname\relax
  \def\url#1{\texttt{#1}}\fi
\expandafter\ifx\csname urlprefix\endcsname\relax\def\urlprefix{URL }\fi
\expandafter\ifx\csname href\endcsname\relax
  \def\href#1#2{#2} \def\path#1{#1}\fi

\bibitem{Gupta_Saia_Young_2021}
D.~Gupta, J.~Saia, M.~Young, Bankrupting {Sybil} despite churn, in: Proceedings
  of the 41st IEEE International Conference on Distributed Computing Systems
  (ICDCS), 2021, pp. 425--437.

\bibitem{Gupta_Saia_Young_2019}
D.~Gupta, J.~Saia, M.~Young, Peace through superior puzzling: An asymmetric
  {S}ybil defense, in: Proceedings of the $33rd$ IEEE International Parallel
  and Distributed Processing Symposium (IPDPS), 2019, pp. 1083--1094.

\bibitem{douceur02sybil}
J.~Douceur, The {Sybil} attack, in: Proceedings of the Second International
  Peer-to-Peer Symposium (IPTPS), 2002, pp. 251--260.

\bibitem{zhang2019double}
S.~Zhang, J.-H. Lee, Double-spending with a {Sybil} attack in the bitcoin
  decentralized network, IEEE transactions on Industrial Informatics 15~(10)
  (2019) 5715--5722.

\bibitem{heilman2015eclipse}
E.~Heilman, A.~Kendler, A.~Zohar, S.~Goldberg, Eclipse attacks on {Bitcoin}’s
  peer-to-peer network, in: 24th $\{$USENIX$\}$ Security Symposium
  ($\{$USENIX$\}$ Security 15), 2015, pp. 129--144.

\bibitem{bitcoin-sybil}
T.~Neudecker, {Bitcoin} cash ({BCH}) {Sybil} nodes on the {Bitcoin}
  peer-to-peer network,
  \url{http://dsn.tm.kit.edu/publications/files/332/bch_sybil.pdf} (2017).

\bibitem{mohaisen:sybil}
A.~Mohaisen, J.~Kim, The {Sybil} attacks and defenses: A survey, Smart
  Computing Review 3~(6) (2013) 480--489.

\bibitem{gupta:resource-burning}
D.~Gupta, J.~Saia, M.~Young, \href{https://arxiv.org/abs/2006.04865}{Invited
  paper: Resource burning for permissionless systems}, in: $27^{th}$
  International Conference on Structural Information and Communication
  Complexity (SIROCCO), 2020, pp. 19--44.
\newline\urlprefix\url{https://arxiv.org/abs/2006.04865}

\bibitem{dwork:pricing}
C.~Dwork, M.~Naor, Pricing via processing or combatting junk mail, in:
  Proceedings of the $12^{th}$ Annual International Cryptology Conference on
  Advances in Cryptology, 1993, pp. 139--147.

\bibitem{economistBC}
T.~Economist, {Why {Bitcoin} uses so much energy},
  \url{https://www.economist.com/the-economist-explains/2018/07/09/why-bitcoin-uses-so-much-energy}
  (2018).

\bibitem{arstechnica}
A.~Technica, Mining {Bitcoins} takes power, but is it an environmental
  disaster?,
  \url{http://arstechnica.com/business/2013/04/mining-bitcoins-takes-power-but-is-it-an-environmental-disaster/}
  (2013).

\bibitem{Stutzbach:2006:UCP:1177080.1177105}
D.~Stutzbach, R.~Rejaie,
  \href{http://doi.acm.org/10.1145/1177080.1177105}{Understanding churn in
  peer-to-peer networks}, in: Proceedings of the 6th ACM SIGCOMM Conference on
  Internet Measurement (IMC), ACM, New York, NY, USA, 2006, pp. 189--202.
\newblock \href {https://doi.org/10.1145/1177080.1177105}
  {\path{doi:10.1145/1177080.1177105}}.
\newline\urlprefix\url{http://doi.acm.org/10.1145/1177080.1177105}

\bibitem{falkner:profiling}
J.~Falkner, M.~Piatek, J.~P. John, A.~Krishnamurthy, T.~Anderson, Profiling a
  million user \textsc{DHT}, in: Proceedings of the $7^{th}$ ACM SIGCOMM
  Conference on Internet Measurement, 2007, pp. 129--134.

\bibitem{6688697}
L.~Wang, J.~Kangasharju, Measuring large-scale distributed systems: Case of
  {BitTorrent Mainline DHT}, in: IEEE 13th International Conference on
  Peer-to-Peer Computing (P2P), 2013, pp. 1--10.

\bibitem{synchrony:malkhi}
D.~Malkhi, {The BFT Lens: Hot-Stuff and Casper},
  \url{https://dahliamalkhi.github.io/posts/2018/03/bft-lens-casper/} (2018).

\bibitem{392384}
D.~L. {Mills}, {Improved Algorithms for Synchronizing Computer Network Clocks},
  IEEE/ACM Transactions on Networking 3~(3) (1995).

\bibitem{nakamoto:bitcoin}
S.~Nakamoto, {Bitcoin}: A peer-to-peer electronic cash system,
  \url{http://bitcoin.org/bitcoin.pdf} (2008).

\bibitem{andrychowicz2015pow}
M.~Andrychowicz, S.~Dziembowski, {PoW}-based distributed cryptography with no
  trusted setup, in: Proceedings of the Annual Cryptology Conference, Springer,
  2015, pp. 379--399.

\bibitem{GiladHMVZ17}
Y.~Gilad, R.~Hemo, S.~Micali, G.~Vlachos, N.~Zeldovich, Algorand: Scaling
  {Byzantine} agreements for cryptocurrencies, in: Proceedings of the 26th
  Symposium on Operating Systems Principles (SOSP), 2017, pp. 51--68.

\bibitem{urdaneta:survey}
G.~Urdaneta, G.~Pierre, M.~van Steen, A survey of {DHT} security techniques,
  ACM Computing Surveys 43~(2) (2011) 1--53.

\bibitem{stoica_etal:chord}
I.~Stoica, R.~Morris, D.~Karger, M.~F. Kaashoek, H.~Balakrishnan, Chord: A
  scalable peer-to-peer lookup service for internet applications, in:
  Proceedings of the Conference on Applications, Technologies, Architectures,
  and Protocols for Computer Communications (SIGCOMM), 2001, pp. 149--160.

\bibitem{guerraoui:highly}
R.~Guerraoui, F.~Huc, A.-M. Kermarrec, Highly dynamic distributed computing
  with {Byzantine} failures, in: Proceedings of the 2013 ACM Symposium on
  Principles of Distributed Computing (PODC), 2013, pp. 176--183.

\bibitem{JaiyeolaPSYZ17}
M.~O. Jaiyeola, K.~Patron, J.~Saia, M.~Young, Q.~M. Zhou, Tiny groups tackle
  {Byzantine} adversaries, in: Proceedings of the {IEEE} International Parallel
  and Distributed Processing Symposium, {IPDPS}, 2018, pp. 1030--1039.

\bibitem{fiat:making}
A.~Fiat, J.~Saia, M.~Young, {Making Chord robust to {Byzantine} attacks}, in:
  Proceedings of the $13^{th}$ European Symposium on Algorithms (ESA), 2005,
  pp. 803--814.

\bibitem{awerbuch_scheideler:group}
B.~Awerbuch, C.~Scheideler, Group spreading: A protocol for provably secure
  distributed name service, in: Proceedings of the $31^{st}$ International
  Colloquium on Automata, Languages, and Programming (ICALP), 2004, pp.
  183--195.

\bibitem{castro:secure}
M.~Castro, P.~Druschel, A.~Ganesh, A.~Rowstron, D.~S. Wallach, Secure routing
  for structured peer-to-peer overlay networks, in: Proceedings of the $5^{th}$
  Usenix Symposium on Operating Systems Design and Implementation (OSDI), 2002,
  pp. 299--314.

\bibitem{young:towards}
M.~Young, A.~Kate, I.~Goldberg, M.~Karsten, Towards practical communication in
  {Byzantine}-resistant {DHTs}, IEEE/ACM Transactions on Networking 21~(1)
  (2013) 190--203.

\bibitem{awerbuch:towards2}
B.~Awerbuch, C.~Scheideler, Towards scalable and robust overlay networks, in:
  Proceedings of the $6^{th}$ International Workshop on Peer-to-Peer Systems
  (IPTPS), 2007.

\bibitem{awerbuch:random}
B.~Awerbuch, C.~Scheideler, Robust random number generation for peer-to-peer
  systems, in: Proceedings of the $10th$ International Conference On Principles
  of Distributed Systems (OPODIS), 2006, pp. 275--289.

\bibitem{awerbuch:towards}
B.~Awerbuch, C.~Scheideler, Towards a scalable and robust {DHT}, in:
  Proceedings of the 18th ACM Symposium on Parallelism in Algorithms and
  Architectures (SPAA), 2006, pp. 318--327.

\bibitem{scheideler:how}
C.~Scheideler, How to spread adversarial nodes? {Rotate}!, in: Proceedings of
  the Thirty-Seventh Annual ACM Symposium on Theory of Computing (STOC), 2005,
  pp. 704--713.

\bibitem{wang2021ethna}
T.~Wang, C.~Zhao, Q.~Yang, S.~Zhang, S.~C. Liew, Ethna: Analyzing the
  underlying peer-to-peer network of {Ethereum} blockchain, IEEE Transactions
  on Network Science and Engineering (2021).

\bibitem{maymounkov:kademlia}
P.~Maymounkov, D.~Mazieres, Kademlia: A peer-to-peer information system based
  on the xor metric, Lecture Notes in Computer Science (2002) 53--65.

\bibitem{CromanDEGJKMSSS16}
K.~Croman, C.~Decker, I.~Eyal, A.~E. Gencer, A.~Juels, A.~E. Kosba, A.~Miller,
  P.~Saxena, E.~Shi, E.~G. Sirer, D.~Song, R.~Wattenhofer, On scaling
  decentralized blockchains {(A} position paper), in: Proceedings of the
  International Financial Cryptography and Data Security (FC), 2016, pp.
  106--125.

\bibitem{AspnesJK2005}
J.~Aspnes, C.~Jackson, A.~Krishnamurthy, Exposing computationally-challenged
  {Byzantine} impostors, Tech. rep., YALEU/DCS/TR-1332, Yale University
  {\url{http://www.cs.yale.edu/homes/aspnes/papers/tr1332.pdf}} (2005).

\bibitem{katz2014pseudonymous}
J.~Katz, A.~Miller, E.~Shi, \href{http://eprint.iacr.org/2014/857}{Pseudonymous
  secure computation from time-lock puzzles}, {IACR} Cryptol. ePrint Arch. 2014
  (2014) 857.
\newline\urlprefix\url{http://eprint.iacr.org/2014/857}

\bibitem{hou2017randomized}
R.~Hou, I.~Jahja, L.~Luu, P.~Saxena, H.~Yu, Randomized view reconciliation in
  permissionless distributed systems, in: Proceedings of the IEEE Conference on
  Computer Communications (INFOCOM), 2018, pp. 2528--2536.

\bibitem{aggarwal2019bootstrapping}
A.~Aggarwal, M.~Movahedi, J.~Saia, M.~Zamani, Bootstrapping public blockchains
  without a trusted setup, in: Proceedings of the 2019 ACM Symposium on
  Principles of Distributed Computing (PODC), ACM, 2019, pp. 366--368.

\bibitem{lynch1996distributed}
N.~A. Lynch, Distributed algorithms, Elsevier, San Francisco, CA, 1996.

\bibitem{benor_goldwasser_wigderson:completeness}
M.~Ben-Or, S.~Goldwasser, A.~Wigderson, Completeness theorems for
  non-cryptographic fault-tolerant distributed computing, in: Proceedings of
  the Twentieth ACM Symposium on the Theory of Computing (STOC), 1988, pp.
  1--10.

\bibitem{gao2018sybilfuse}
P.~Gao, B.~Wang, N.~Z. Gong, S.~R. Kulkarni, K.~Thomas, P.~Mittal, Sybilfuse:
  Combining local attributes with global structure to perform robust {Sybil}
  detection, in: 2018 IEEE Conference on Communications and Network Security
  (CNS), 2018, pp. 1--9,
  \url{https://www.princeton.edu/~pmittal/publications/sybilfuse-cns18.pdf}.

\bibitem{augustine2012towards}
J.~Augustine, G.~Pandurangan, P.~Robinson, E.~Upfal, Towards robust and
  efficient computation in dynamic peer-to-peer networks, in: Proceedings of
  the Twenty-Third Annual ACM-SIAM Symposium on Discrete Algorithms, SIAM,
  2012, pp. 551--569.

\bibitem{augustine2015distributed}
J.~Augustine, G.~Pandurangan, P.~Robinson, E.~Upfal, Distributed agreement in
  dynamic peer-to-peer networks, Journal of Computer and System Sciences 81~(7)
  (2015) 1088--1109.

\bibitem{augustine2015fast}
J.~Augustine, G.~Pandurangan, P.~Robinson, Fast byzantine leader election in
  dynamic networks, in: International Symposium on Distributed Computing,
  Springer, 2015, pp. 276--291.

\bibitem{jacobs2013stochastic}
T.~Jacobs, G.~Pandurangan, Stochastic analysis of a churn-tolerant structured
  peer-to-peer scheme, Peer-to-peer Networking and Applications 6~(1) (2013)
  1--14.

\bibitem{augustine2015enabling}
J.~Augustine, G.~Pandurangan, P.~Robinson, S.~Roche, E.~Upfal, Enabling robust
  and efficient distributed computation in dynamic peer-to-peer networks, in:
  2015 IEEE 56th Annual Symposium on Foundations of Computer Science, IEEE,
  2015, pp. 350--369.

\bibitem{augustine2013storage}
J.~Augustine, A.~R. Molla, E.~Morsy, G.~Pandurangan, P.~Robinson, E.~Upfal,
  Storage and search in dynamic peer-to-peer networks, in: Proceedings of the
  twenty-fifth annual ACM symposium on Parallelism in algorithms and
  architectures, 2013, pp. 53--62.

\bibitem{augustine2016distributed}
J.~Augustine, G.~Pandurangan, P.~Robinson, Distributed algorithmic foundations
  of dynamic networks, ACM SIGACT News 47~(1) (2016) 69--98.

\bibitem{ko2008using}
S.~Y. Ko, I.~Hoque, I.~Gupta, Using tractable and realistic churn models to
  analyze quiescence behavior of distributed protocols, in: 2008 Symposium on
  Reliable Distributed Systems, IEEE, 2008, pp. 259--268.

\bibitem{aguilera2004pleasant}
M.~K. Aguilera, A pleasant stroll through the land of infinitely many
  creatures, ACM Sigact News 35~(2) (2004) 36--59.

\bibitem{liben2002analysis}
D.~Liben-Nowell, H.~Balakrishnan, D.~Karger, Analysis of the evolution of
  peer-to-peer systems, in: Proceedings of the twenty-first annual symposium on
  Principles of distributed computing, 2002, pp. 233--242.

\bibitem{gummadi2003measurement}
K.~P. Gummadi, R.~J. Dunn, S.~Saroiu, S.~D. Gribble, H.~M. Levy, J.~Zahorjan,
  Measurement, modeling, and analysis of a peer-to-peer file-sharing workload,
  in: Proceedings of the nineteenth ACM symposium on Operating systems
  principles, 2003, pp. 314--329.

\bibitem{imtiaz2019churn}
M.~A. Imtiaz, D.~Starobinski, A.~Trachtenberg, N.~Younis, Churn in the bitcoin
  network: Characterization and impact, in: 2019 IEEE International Conference
  on Blockchain and Cryptocurrency (ICBC), IEEE, 2019, pp. 431--439.

\bibitem{john:soft}
R.~John, J.~P. Cherian, J.~J. Kizhakkethottam, A survey of techniques to
  prevent {Sybil} attacks, in: Proceedings of the International Conference on
  Soft-Computing and Networks Security (ICSNS), 2015, pp. 1--6.

\bibitem{newsome:sybil}
J.~Newsome, E.~Shi, D.~Song, A.~Perrig, The {Sybil} attack in sensor networks:
  Analysis \& defenses, in: Proceedings of the 3rd International Symposium on
  Information Processing in Sensor Networks (IPSN), 2004, pp. 259--268.

\bibitem{6503215}
L.~Wang, J.~Kangasharju, Real-world {Sybil} attacks in {BitTorrent Mainline
  DHT}, in: Proceedings of the IEEE Global Communications Conference
  (GLOBECOM), 2012, pp. 826--832.

\bibitem{Yang:2011:USN:2068816.2068841}
Z.~Yang, C.~Wilson, X.~Wang, T.~Gao, B.~Y. Zhao, Y.~Dai, {Uncovering social
  network Sybils in the wild}, in: Proceedings of the 2011 ACM SIGCOMM
  Conference on Internet Measurement Conference (IMC), 2011, pp. 259--268.

\bibitem{yu:sybilguard}
H.~Yu, M.~Kaminsky, P.~B. Gibbons, A.~Flaxman, Sybilguard: Defending against
  {Sybil} attacks via social networks, in: Proceedings of the 2006 Conference
  on Applications, Technologies, Architectures, and Protocols for Computer
  Communications (SIGCOMM), 2006, pp. 267--278.

\bibitem{mohaisen:improving}
A.~Mohaisen, S.~Hollenbeck, Improving social network-based {Sybil} defenses by
  rewiring and augmenting social graphs, in: Proceedings of the $14^{th}$
  International Workshop on Information Security Applications (WISA), 2014, pp.
  65--80.

\bibitem{wei:sybildefender}
W.~Wei, F.~Xu, C.~C. Tan, Q.~Li, {SybilDefender}: A defense mechanism for
  {Sybil} attacks in large social networks, IEEE Transactions on Parallel \&
  Distributed Systems 24~(12) (2013) 2492--2502.

\bibitem{misra2016sybilexposer}
S.~Misra, A.~S.~M. Tayeen, W.~Xu, Sybilexposer: An effective scheme to detect
  {Sybil} communities in online social networks, in: 2016 IEEE International
  Conference on Communications (ICC), IEEE, 2016, pp. 1--6.

\bibitem{sherr:veracity}
M.~Sherr, M.~Blaze, B.~T. Loo, Veracity: Practical secure network coordinates
  via vote-based agreements, in: Proceedings of the USENIX Annual Technical
  Conference, 2009, pp. 13 --13.

\bibitem{liu:mason}
Y.~Liu, D.~R. Bild, R.~P. Dick, Z.~M. Mao, D.~S. Wallach, The {Mason} test: A
  defense against {Sybil} attacks in wireless networks without trusted
  authorities, IEEE Transactions on Mobile Computing 14~(11) (2015) 2376--2391.

\bibitem{Gil-RSS-15}
S.~Gil, S.~Kumar, M.~Mazumder, D.~Katabi, D.~Rus, Guaranteeing spoof-resilient
  multi-robot networks, in: Proceedings of Robotics: Science and Systems, Rome,
  Italy, 2015, pp. 1383--1400.

\bibitem{danezis:sybil}
G.~Danezis, C.~Lesniewski-laas, M.~F. Kaashoek, R.~Anderson, Sybil-resistant
  \textsc{DHT} routing, in: Proceedings of the $10^{th}$ European Symposium On
  Research In Computer Security (ESORICS), 2005, pp. 305--318.

\bibitem{scheideler:shell}
C.~Scheideler, S.~Schmid, A distributed and oblivious heap, in: Proceedings of
  the 36th Internatilonal Collogquium on Automata, Languages and Programming:
  Part II, ICALP '09, 2009, pp. 571--582.

\bibitem{li:sybilcontrol}
F.~Li, P.~Mittal, M.~Caesar, N.~Borisov, {SybilControl}: Practical {Sybil}
  defense with computational puzzles, in: Proceedings of the Seventh ACM
  Workshop on Scalable Trusted Computing, 2012, pp. 67--78.

\bibitem{moran2019simple}
T.~Moran, I.~Orlov, Simple proofs of space-time and rational proofs of storage,
  in: Annual International Cryptology Conference, Springer, 2019, pp. 381--409.

\bibitem{shoker2017sustainable}
A.~Shoker, Sustainable blockchain through proof of exercise, in: 2017 IEEE 16th
  International Symposium on Network Computing and Applications (NCA), IEEE,
  2017, pp. 1--9.

\bibitem{ball2018proofs}
M.~Ball, A.~Rosen, M.~Sabin, P.~N. Vasudevan, Proofs of work from worst-case
  assumptions, in: Proceedings of the Annual International Cryptology
  Conference (CRYPTO), Springer, 2018, pp. 789--819.

\bibitem{von2003captcha}
L.~Von~Ahn, M.~Blum, N.~J. Hopper, J.~Langford, {CAPTCHA}: Using hard {AI}
  problems for security, in: Proceedings of the International Conference on the
  Theory and Applications of Cryptographic Techniques, Springer, 2003, pp.
  294--311.

\bibitem{moradi2015captcha}
M.~Moradi, M.~Keyvanpour, {CAPTCHA and its alternatives: A review}, Security
  and Communication Networks 8~(12) (2015) 2135--2156.

\bibitem{baird2005scattertype}
H.~S. Baird, M.~A. Moll, S.-Y. Wang, Scattertype: A legible but hard-to-segment
  {CAPTCHA}, in: Proceedings of the Eighth International Conference on Document
  Analysis and Recognition (ICDAR), 2005, pp. 935--939.

\bibitem{von2008recaptcha}
L.~Von~Ahn, B.~Maurer, C.~McMillen, D.~Abraham, M.~Blum, Recaptcha: Human-based
  character recognition via web security measures, Science 321~(5895) (2008)
  1465--1468.

\bibitem{monica:radio}
D.~M\'{o}nica, L.~Leitao, L.~Rodrigues, C.~Ribeiro, On the use of radio
  resource tests in wireless ad-hoc networks, in: Proceedings of the $3rd$
  Workshop on Recent Advances on Intrusion-Tolerant Systems, 2009, pp. 21--26.

\bibitem{gilbert:sybilcast}
S.~Gilbert, C.~Zheng, Sybilcast: Broadcast on the open airwaves, in:
  Proceedings of the $25^{th}$ Annual ACM Symposium on Parallelism in
  Algorithms and Architectures (SPAA), 2013, pp. 130--139.

\bibitem{gilbert:who}
S.~Gilbert, C.~Newport, C.~Zheng, Who are you? {Secure} identities in ad hoc
  networks, in: Proceedings of the $28^{th}$ International Symposium on
  Distributed Computing (DISC), 2014, pp. 227--242.

\bibitem{Kiayias2017}
A.~Kiayias, A.~Russell, B.~David, R.~Oliynykov, Ouroboros: A provably secure
  proof-of-stake blockchain protocol, in: Proceedings of the $37th$ Annual
  International Cryptology Conference, Vol. 10401 of Lecture Notes in Computer
  Science, Springer, 2017, pp. 357--388.

\bibitem{ethereum-pos}
O.~Kharif, Bye-bye, miners! {How} {Ethereum’s} big change will work,
  \url{https://www.bloombergquint.com/quicktakes/bye-bye-miners-how-ethereum-s-big-change-will-work-quicktake}
  (2021).

\bibitem{posdm}
CoinDesk, Vulnerable? {Ethereum's Casper} tech takes criticism at {Curacao}
  event,
  \url{https://www.coindesk.com/fundamentally-vulnerable-ethereums-casper-tech-takes-criticism-curacao}
  (2018).

\bibitem{gupta2017proof}
D.~Gupta, J.~Saia, M.~Young, Proof of work without all the work, arXiv preprint
  arXiv:1708.01285 (2017).

\bibitem{gilbert:making}
S.~Gilbert, M.~Young, {Making Evildoers Pay: Resource-Competitive Broadcast in
  Sensor Networks}, in: Proceedings of the $31^{th}$ Symposium on Principles of
  Distributed Computing (PODC), 2012, pp. 145--154.

\bibitem{gilbert:near}
S.~Gilbert, V.~King, S.~Pettie, E.~Porat, J.~Saia, M.~Young, {(Near)} optimal
  resource-competitive broadcast with jamming, in: Proceedings of the $26^{th}$
  ACM Symposium on Parallelism in Algorithms and Architectures (SPAA), 2014,
  pp. 257--266.

\bibitem{king:conflict}
V.~King, J.~Saia, M.~Young, Conflict on a communication channel, in:
  Proceedings of the $30^{th}$ Symposium on Principles of Distributed Computing
  (PODC), 2011, pp. 277--286.

\bibitem{bender:how}
M.~A. Bender, J.~T. Fineman, S.~Gilbert, M.~Young, How to scale exponential
  backoff: Constant throughput, polylog access attempts, and robustness, in:
  Proceedings of the $27^{th}$ Annual ACM-SIAM Symposium on Discrete Algorithms
  (SODA), 2016, pp. 636--654.

\bibitem{ICALP15}
V.~Dani, M.~Movahedi, J.~Saia, M.~Young, Interactive communication with unknown
  noise rate, in: Proceedings of the Colloquium on Automata, Languages, and
  Programming (ICALP), 2015, pp. 464--486.

\bibitem{daniICJournal17}
V.~Dani, T.~Hayes, M.~Movahedi, J.~Saia, M.~Young, Interactive communication
  with unknown noise rate, Information and Computation 261 (2017) 464--486.

\bibitem{aggarwal2016secure}
A.~Aggarwal, V.~Dani, T.~Hayes, J.~Saia, Secure one-way interactive
  communication, in: Proceedings of the 15th International Conference on
  Distributed Computing and Networking (ICDCN), 2017.

\bibitem{gilbert:resource}
S.~Gilbert, V.~King, J.~Saia, M.~Young, Resource-competitive analysis: A new
  perspective on attack-resistant distributed computing, in: Proceedings of the
  $8^{th}$ ACM International Workshop on Foundations of Mobile Computing, 2012,
  pp. 1--6.

\bibitem{zamani2017torbricks}
M.~Zamani, J.~Saia, J.~Crandall, Torbricks: Blocking-resistant tor bridge
  distribution, in: International Symposium on Stabilization, Safety, and
  Security of Distributed Systems (SSS), Springer, 2017, pp. 426--440.

\bibitem{Bender:2015:RA:2818936.2818949}
M.~A. Bender, J.~T. Fineman, M.~Movahedi, J.~Saia, V.~Dani, S.~Gilbert,
  S.~Pettie, M.~Young, Resource-competitive algorithms, SIGACT News 46~(3)
  (2015) 57--71.

\bibitem{dubhashi:concentration}
D.~Dubhashi, A.~Panconesi, {Concentration of Measure for the Analysis of
  Randomized Algorithms}, 1st Edition, Cambridge University Press, 2009.

\bibitem{horn2012matrix}
R.~A. Horn, C.~R. Johnson, Matrix analysis, Cambridge University Press,
  Cambridge, UK, 2012.

\bibitem{diksha-code}
D.~Gupta, Source code link,
  \url{https://www.dropbox.com/sh/ao19eqnresc530k/AAAUF5oPqBJo5ql_CT7T9WbNa?dl=0},
  last updated on August 25, 2022 (2022).

\bibitem{7140490}
T.~Neudecker, P.~Andelfinger, H.~Hartenstein, A simulation model for analysis
  of attacks on the {Bitcoin} peer-to-peer network, in: Proceedings of the
  IFIP/IEEE International Symposium on Integrated Network Management (IM),
  2015, pp. 1327--1332.

\bibitem{kim2017measuring}
S.~K. Kim, Z.~Ma, S.~Murali, J.~Mason, A.~Miller, M.~Bailey, Measuring
  {Ethereum} network peers, in: Proceedings of the Internet Measurement
  Conference, 2018, pp. 91--104.

\bibitem{rowaihy2007limiting}
H.~Rowaihy, W.~Enck, P.~McDaniel, T.~La~Porta, Limiting {Sybil} attacks in
  structured {P2P} networks, in: Proceedings of the IEEE Intl. Conference on
  Computer Communications (INFOCOM), 2007, pp. 2596--2600.

\bibitem{pow-without}
D.~Gupta, J.~Saia, M.~Young, Proof of work without all the work, in:
  Proceedings of the $19^{th}$ International Conference on Distributed
  Computing and Networking (ICDCN), 2018, pp. 1--10.

\bibitem{rowaihy2005limiting}
H.~Rowaihy, W.~Enck, P.~Mcdaniel, T.~La~Porta, Limiting {Sybil} attacks in
  structured peer-to-peer networks, in: IEEE {INFOCOM} Mini-Symposium, 2005.

\bibitem{neudecker-atc16}
T.~Neudecker, P.~Andelfinger, H.~Hartenstein, Timing analysis for inferring the
  topology of the {Bitcoin} peer-to-peer network, in: Proceedings of the 13th
  IEEE International Conference on Advanced and Trusted Computing (ATC), 2016,
  pp. 358--367.

\bibitem{saroiu2001measurement}
S.~Saroiu, P.~K. Gummadi, S.~D. Gribble, Measurement study of peer-to-peer file
  sharing systems, in: Multimedia computing and networking 2002, Vol. 4673,
  SPIE, 2001, pp. 156--170.

\bibitem{kiffer2021under}
L.~Kiffer, A.~Salman, D.~Levin, A.~Mislove, C.~Nita-Rotaru, Under the hood of
  the ethereum gossip protocol, in: International Conference on Financial
  Cryptography and Data Security, Springer, 2021, pp. 437--456.

\bibitem{abraham2019sync}
I.~Abraham, D.~Malkhi, K.~Nayak, L.~Ren, M.~Yin, Sync {HotStuff}: Simple and
  practical synchronous state machine replication, IACR Cryptology ePrint
  Archive (2019).

\bibitem{rabin1989verifiable}
T.~Rabin, M.~Ben-Or, Verifiable secret sharing and multiparty protocols with
  honest majority, in: Proceedings of the $21^{st}$ Annual ACM symposium on
  Theory of Computing, 1989, pp. 73--85.

\bibitem{attiya2004distributed}
H.~Attiya, J.~Welch, Distributed computing: Fundamentals, Simulations, and
  Advanced Topics, Vol.~19, John Wiley \& Sons, 2004.

\bibitem{bessani2014state}
A.~Bessani, J.~Sousa, E.~E. Alchieri, State machine replication for the masses
  with {BFT-SMART}, in: Proceedings of the 44th Annual IEEE/IFIP International
  Conference on Dependable Systems and Networks, IEEE, 2014, pp. 355--362.

\bibitem{andrychowicz2014distributed}
M.~Andrychowicz, S.~Dziembowski, D.~Malinowski, {\L}.~Mazurek, Modeling
  {Bitcoin} contracts by timed automata, in: Proceedings of the International
  Conference on Formal Modeling and Analysis of Timed Systems, Springer, 2014,
  pp. 7--22.

\bibitem{mitzenmacher2017probability}
M.~Mitzenmacher, E.~Upfal, Probability and Computing: Randomization and
  Probabilistic Techniques in Algorithms and Data Analysis, Cambridge
  University Press, 2017.

\bibitem{abraham2016solidus}
I.~Abraham, D.~Malkhi, K.~Nayak, L.~Ren, A.~Spiegelman, {Solidus: An
  Incentive-compatible Cryptocurrency Based on Permissionless Byzantine
  Consensus}, CoRR, abs/1612.02916 (2016).

\bibitem{margolin2007informant}
N.~B. Margolin, B.~N. Levine, Informant: Detecting {Sybils} using incentives,
  in: International Conference on Financial Cryptography and Data Security,
  Springer, 2007, pp. 192--207.

\bibitem{liu2019survey-pub}
Z.~Liu, N.~C. Luong, W.~Wang, D.~Niyato, P.~Wang, Y.-C. Liang, D.~I. Kim, A
  survey on blockchain: A game theoretical perspective, IEEE Access 7 (2019)
  47615--47643.

\end{thebibliography}

\end{document}